\documentclass{article}

\usepackage{prelude}

\usepackage{subcaption}

\usepackage{tikz}
\usetikzlibrary{positioning,calc,decorations.pathreplacing}
\tikzstyle{qubit}=[circle,draw,fill,thick,
inner sep=0pt,minimum size=2mm]
\tikzstyle{mqubit}=[circle,draw,fill=white,
inner sep=0pt,minimum size=2mm]
\tikzstyle{rqubit}=[circle,draw=black,shade,ball color=red,
inner sep=0pt,minimum size=4mm]
\tikzstyle{bqubit}=[circle,draw=black,shade,ball color=blue,
inner sep=0pt,minimum size=4mm]
\tikzstyle{bqubits}=[circle,draw=black,shade,ball color=blue,
inner sep=0pt,minimum size=1.75mm]
\tikzstyle{heavy}=[ultra thick]

\crefname{condition}{condition}{conditions}
\crefname{part}{part}{parts}

\newcommand\nbd\nobreakdash

\newcommand{\cA}{\mathcal{A}}
\newcommand{\cM}{\mathcal{M}}
\newcommand{\cEs}{\cE_{\operatorname{state}}}
\newcommand{\cEm}{\cE_{\operatorname{meas}}}

\newcommand{\Herm}{\operatorname{Herm}}
\newcommand{\rest}{\operatorname{rest}}

\long\def\onlyshort#1{}
\newcommand\shortcite[1]{~\cite{#1}}
\newcommand\citesupp{}
\newcommand\textcitesupp{\cref{part:technical}}

\begin{document}


\title{Universal quantum Hamiltonians}

\author{Toby Cubitt$^1$, Ashley Montanaro$^2$ and Stephen Piddock$^2$\\
{\small $^1$ Department of Computer Science, University College London, UK\\
$^2$ School of Mathematics, University of Bristol, UK.}
}
\date{\today}

\maketitle

\begin{abstract}
  Quantum many-body systems exhibit an extremely diverse range of phases and physical phenomena.
However, we prove that the entire physics of any quantum many-body system can be replicated by certain simple, ``universal'' spin-lattice models.
We first characterise precisely what it means for one quantum system to simulate the entire physics of another.
We then fully classify the simulation power of all two-qubit interactions, thereby proving that certain simple models can simulate all others, hence are universal.
Our results put the practical field of analogue Hamiltonian simulation on a rigorous footing, and take a step towards justifying why error correction may not be required for this application of quantum information technology.


\end{abstract}

\clearpage

\renewcommand\baselinestretch{0.94}\normalfont
\tableofcontents
\clearpage




\part{Extended overview}
The properties of any physical system are captured in its Hamiltonian, which describes all the possible energy configurations of the system.
Amongst the workhorses of theoretical many-body and condensed matter physics are spin-lattice Hamiltonians, in which the degrees of freedom are quantum spins arranged on a lattice, and the overall Hamiltonian is built up from few-body interactions between these spins.
Although these are idealised, toy models of real materials, different spin-lattice Hamiltonians are able to model a wide variety of different quantum phases and many-body phenomena: phase transitions\shortcite{Sachdev_book}, frustration\shortcite{frustrated}, spontaneous symmetry-breaking\shortcite{symmetry-breaking}, gauge symmetries\shortcite{gauge-symmetry}, quantum magnetism\shortcite{magnetism}, spin liquids\shortcite{spin-liquids}, topological order\shortcite{Kitaev-model}, and more.
In this work, we prove that there exist particular, simple spin models that are universal: they can replicate to any desired accuracy the entire physics of any other quantum many-body system (including systems composed not only of spins, but also bosons and fermions).
This implies, in particular, that the ground state, full energy spectrum and associated excited states, all observables, correlation functions, thermal properties, time-evolution, and also any local noise processes are reproduced by the universal model.

Note that this is a very different notion of ``universality'' from that of universality classes in condensed matter and statistical physics\shortcite{Cardy_book}.
Universality classes capture the fact that, if we repeatedly ``zoom out'' or course-grain the microscopic degrees of freedom of a many-body system, models that are microscopically different become increasingly similar (converge to the same limit under this ``renormalisation group flow''), and their macroscopic properties turn out to fall into one of a small handful of possible classes.
The ``universality'' we are concerned with here\shortcite{Science} has a completely different and unrelated meaning.
It is closer to the notion of universality familiar from computing.
A universal computer can carry out any possible computation, including simulating completely different types of computer.
Universal models are able to produce any many-body physics phenomena, including reproducing the physics of completely different many-body models.

One might expect that universal models must be very complicated for their phase diagram to encompass all possible many-body physics.
In fact, some of the models we show to be universal are amongst the simplest possible.
Nearest-neighbour Heisenberg interactions on a square lattice give rise to 2D models with the simplest possible local degrees of freedom (qubits), short-range, two-body interactions, and the largest possible local symmetry (full $SU(2)$ invariance).
Yet our results prove that, if all the coupling strengths can be varied individually, this model is universal.
Thus it can replicate in a rigorous sense the full physics of models with higher spatial dimensions, long-range interactions, other symmetries, higher-dimensional spins, and even bosons and fermions.

In addition to the new relationships this establishes between apparently very different quantum many-body models, with implications for our fundamental understanding of quantum many-body physics, there are also potential practical applications of our results in the field of analogue quantum simulation.
There is substantial interest nowadays in using one quantum many-body system to simulate the physics of another, and one of the most important applications of quantum computers is anticipated to be the simulation of quantum systems\shortcite{georgescu14,cirac12}.

Two quite different notions of Hamiltonian simulation are studied in the literature.
The first concerns simulating the time-dynamics of a Hamiltonian on a quantum computer using an algorithm originally proposed by Lloyd\shortcite{Lloyd}, and refined and improved in the decades since\shortcite{Berry07,Berry14a,Berry15,low16}.
This is the quantum computing equivalent of running a numerical simulation on a classical computer.
However, it requires a scalable, fault-tolerant, digital quantum computer.
Except for small-scale proof-of-principle demonstrations, this is beyond the reach of current technology.
The second notion, called ``physical'' or ``analogue'' -- in the sense of ``analogous'' -- Hamiltonian simulation, involves directly engineering the Hamiltonian of interest and studying its properties experimentally.
(Akin to building a model of an aerofoil and studying it in a wind tunnel.)
This form of Hamiltonian simulation is already being performed in the laboratory using a variety of technologies, including optical lattices, ion traps, superconducting circuits and others\shortcite{somaroo99,NatPhys_simulation_review,georgescu14}.
Just as it is easier to study a scale model of an aerofoil in a wind tunnel than an entire aeroplane, the advantage of artificially engineering a Hamiltonian that models a material of interest, rather than studying that material directly, is that it is typically easier to measure and manipulate the artificially-engineered system.
It is possible to measure the state of a single atom in an optical lattice\shortcite{Bloch_single-site,Grainer_single-site,Chin_single-site}; it is substantially harder to measure e.g.\ the state of a single electron spin in a 2D layer within a cuprate superconductor.

Many important theoretical questions regarding analogue quantum simulation remain open, despite its practical significance and experimental success\shortcite{somaroo99,NatPhys_simulation_review,georgescu14}.
Which systems can simulate which others?
How can we characterise the effect of errors on an analogue quantum simulator?
(Highlighted in the 2012 review article\shortcite{cirac12} as one of the key questions in this field.)
On a basic level, what should the general definition of analogue quantum simulation itself be?
The notions of simulation and universality we develop here enable us to answer all these questions.

This computationally-inspired notion of physical universality has its origins in earlier work on ``completeness'' of the partition function of certain classical statistical mechanics models \shortcite{Va08,De09a,Ka12}.
Recent results by one of us and De las Cuevas built on those ideas to establish the more stringent notion of universality for classical spin systems\shortcite{Science}.
Related, more practically-focused notions have also been explored in recent work motivated by classical Hamiltonian engineering experiments\shortcite{LHZ15}.
Here we consider the richer and more complex setting of quantum Hamiltonians, which requires completely different techniques.

For our explicit constructions that establish the existence of universal Hamiltonians, we are able to draw on a long literature in the field of Hamiltonian complexity\shortcite{Kempe-Kitaev-Regev,Oliveira-Terhal,AGIK,Biamonte-Love,Schuch-Verstraete,Gottesman-Irani,Cubitt-Montanaro,Bravyi-Hastings}, studying the computational complexity of estimating ground state energies.
These results per se only concern the ground state energy, and moreover only the computational complexity of estimating it.
Nonetheless, the ``perturbative gadget'' techniques developed to prove Hamiltonian complexity results\shortcite{Kempe-Kitaev-Regev,Oliveira-Terhal} turn out to be highly useful in constructing the full physical simulations required for our results.
By combining our new and precise mathematical understanding of analogue Hamiltonian simulation with these Hamiltonian complexity techniques, we are able to design new ``gadgets'' that transform one many-body Hamiltonian into another whilst preserving its entire physics and local structure, as required to construct universal models.

\section{Hamiltonian simulation}
We start by establishing precisely what it means for one quantum many-body system to simulate another.
Any non-trivial simulation of one Hamiltonian $H$ with another $H'$ will involve encoding the first within the second in some way.
We want this encoding $H' = \cE(H)$ to ``replicate all the physics'' of the original $H$.
To reproduce all static, dynamic and thermodynamic properties, the encoding $\cE$ needs to fulfil a long list of operational requirements:

\begin{enumerate}[label=\arabic*.,ref=\arabic*]
\setlength{\itemsep}{0pt}
\item \label{short:Hermitian-preserving}%
  $\mathcal{E}(H)$ should be a valid Hamiltonian: $\cE(H) = \cE(H)^\dg$.
\item \label{short:spectrum-preserving}%
  $\mathcal{E}$ should reproduce the complete energy spectrum: $\spec(\cE(H)) = \spec(H)$.
  More generally, $\mathcal{E}(M)$ should preserve the outcomes (eigenvalues) of any measurement $M$: $\spec(\cE(M)) = \spec(M)$.
\item \label{short:local_interactions}%
  Individual interactions in the Hamiltonian should be encoded separately: $\cE(\sum_i\alpha_i h_i) = \sum_i\alpha_i\cE(h_i)$. Otherwise, encoding would require solving the full many-body Hamiltonian, obviating any need to simulate it.
\item \label{short:observables}
  There should exist a corresponding encoding of states, $\cE_{\operatorname{state}}$, such that measurements on states are simulated correctly: for any observable $A$, $\tr(\cE(A) \cE_{\operatorname{state}}(\rho)) = \tr(A \rho)$.
\item \label{short:partition-function}
  $\cE$ should preserve the partition function (potentially up to a physically unimportant constant rescaling): $Z_{H'}(\beta) = \tr(e^{-\beta\cE(H)}) = c \tr(e^{-\beta H}) = c\, Z_H(\beta)$.
\item \label{short:time-evolution}
  Time-evolution according to $\cE(H)$ should simulate time-evolution according to $H$: $e^{-iH't} \mathcal{E}_{\operatorname{state}}(\rho) e^{iH't} = \mathcal{E}_{\operatorname{state}}( e^{-iHt} \rho e^{iHt})$.
\item \label{short:errors}
  Any error or noise process on the $\cE(H)$ system should correspond to some error or noise process on the $H$ system: for any superoperator $\cN'$, there should exist a superoperator $\cN$ such that $\mathcal{N}'(\mathcal{E}_{\operatorname{state}}(\rho)) = \mathcal{E}_{\operatorname{state}}(\mathcal{N}(\rho))$.
\end{enumerate}

Using Jordan- and C*-algebra techniques, we prove (see \textcitesupp{}) that, remarkably, the very basic requirements \labelcref{short:Hermitian-preserving,short:spectrum-preserving,short:local_interactions} already imply that all other operational requirements are satisfied too.
Furthermore, any encoding map $\mathcal{E}$ that satisfies them must have a particularly simple mathematical form:
\begin{equation}
  \cE(H) = U (H^{\oplus p} \oplus \bar{H}^{\oplus q}) U^\dg
\end{equation}
for some unitary $U$ and non-negative integers $p$, $q$ such that $p+q \ge 1$. ($\bar{H}$ denotes complex conjugation of $H$.)

This characterisation of Hamiltonian encodings holds if the entire simulation is to \emph{exactly} replicate all the physics of the original.
But in practice no simulation will ever be exact.
What if the simulator Hamiltonian $H'$ only replicates the physics of the original Hamiltonian $H$ up to some approximation?
As long as this approximation can be made arbitrarily accurate, $H'$ will be able to replicate the entire physics of $H$ to any desired precision.

Moreover, it suffices if the physics of $H$ is replicated within some well-isolated subspace of $H'$, even if $H'$ behaves nothing like $H$ outside that subspace.
An important case is when the simulation occurs within the subspace of states with energy below some cut-off $\Delta$, especially if this energy cut-off can be made as large as desired (see \cref{short:fig:hamsim}).
Due to energy conservation, any initial state with energy less than $\Delta$ will be unaffected by the high-energy sector.
Indeed, as long as the cut-off is larger than the maximum eigenvalue of $H$, $H'$ will be able to simulate all possible states of $H$.
This also holds for all thermodynamic properties; any error in the partition function due to the high-energy sector is exponentially suppressed with increasing $\Delta$.
In practice, one is often interested only in low-temperature properties of a quantum many-body Hamiltonian, as these are the properties relevant to quantum phases and phase transitions.
In that case, the energy cut-off does not even need to be large, merely sufficiently above the lowest excitation energy.
Thus we need to generalise our characterisation to encompass approximate simulation of $H$ in the low-energy subspace of $H'$.

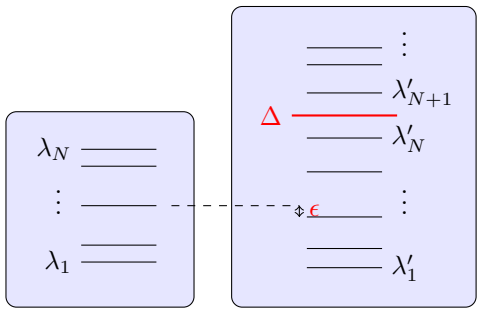
\begin{figure}
\begin{center}
\begin{tikzpicture}[yscale=0.75]

\filldraw[rounded corners,fill=blue!10] (-1,-0.6) rectangle (1.5,2);
\filldraw[rounded corners,fill=blue!10] (2,-0.6) rectangle (5.25,3.4);

\begin{scope}[yscale=1.5]
\draw (0,0) node[left] {$\lambda_1$} -- (1,0);
\draw (0,0.15) -- (1,0.15);
\draw (0,0.55) -- (1,0.55);
\node at (-0.3,0.5) {$\vdots$};
\draw (0,0.85) node[left] {$\lambda_N$} -- (1,0.85);

\draw[dashed] (1.2,0.55) -- (2.9,0.55);
\draw[<->] (2.9,0.4) -- (2.9,0.55);
\node[red] at (3.1,0.51) {$\epsilon$};
\end{scope}

\begin{scope}[xshift=3cm,yscale=1.5]
\draw (0,-0.05) -- (1,-0.05) node[right] {$\lambda'_1$};
\draw (0,0.12) -- (1,0.12);
\draw (0,0.4) -- (1,0.4);
\node at (1.3,0.5) {$\vdots$};
\draw (0,0.8) -- (1,0.8) node[right] {$\lambda'_N$};
\draw[thick,red] (-0.2,1.1) node[left] {$\Delta$} -- (1.2,1.1);
\draw (0,1.5) -- (1,1.5) node[right] {$\lambda'_{N+1}$};
\draw (0,1.75) -- (1,1.75);
\draw (0,1.9) -- (1,1.9);
\node at (1.3,2) {$\vdots$};
\end{scope}

\end{tikzpicture}
\end{center}
\caption{Simulating one Hamiltonian within the low-energy space of another. $H'$ (on right) simulates $H$ (on left) to precision $(\eta,\epsilon)$ below energy cut-off $\Delta$.}
\label{short:fig:hamsim}
\end{figure}

Finally, for a good simulation we would also like the encoding to be \emph{local}, in the sense that each subsystem of the original Hamiltonian corresponds to a distinct subset of particles in the simulator.
This will enable us to map local observables on the original system to local observables on the simulator system, as well as to efficiently prepare states of the simulator.

By making all the above mathematically precise, we show\citesupp{} that this leads to the following rigorous notion of Hamiltonian simulation, which encompasses both exact simulation (as a special case) and, more generally, approximate simulation within a low-energy subspace (also see \cref{short:fig:hamsim}):

\begin{definition}[Analogue Hamiltonian simulation]
  \label{short:def:simulation}\hfill\newline
  A many-body Hamiltonian $H'$ \emph{simulates} a Hamiltonian $H$ to precision $(\eta,\epsilon)$ below an energy cut-off $\Delta$ if there exists a local encoding $\mathcal{E}(H) = V (H\ox P + \bar{H} \ox Q) V^\dg$, where $V= \bigotimes_i V_i$ for some isometries $V_{i}$ acting on 0 or 1 qudits of the original system each, and $P$ and $Q$ are locally orthogonal projectors, such that:
  \begin{enumerate}
  \item \label{short:def:simulation:encoding}
    There exists an encoding $\widetilde{\cE}(H) = \widetilde{V}(H\ox P + \bar{H}\ox Q)\widetilde{V}^\dg$ such that $\widetilde{\cE}(\1) = P_{\le \Delta(H')}$ and \mbox{$\|\widetilde{V} - V\| \le \eta$};
  \item \label{short:def:simulation:Hamiltonian}
    $\| H'_{\le \Delta} - \widetilde{\mathcal{E}}(H)\| \le \epsilon$.
  \end{enumerate}
\end{definition}
Here, we write $H'_{\le \Delta} = P_{\le \Delta(H')} H'$ where $P_{\le \Delta(H')}$ denotes the projector onto the subspace spanned by eigenvectors of $H'$ with eigenvalues below $\Delta$.

The first requirement \labelcref{short:def:simulation:encoding} states that, to good approximation (i.e.\ within error $\eta$), the local encoding $\cE$ approximates an encoding $\widetilde{\cE}$ onto low-energy states of $H'$.
The second requirement \labelcref{short:def:simulation:Hamiltonian} says that the map $\widetilde{\cE}$ gives a good simulation of $H$ (i.e.\ within error $\epsilon$).
Note that if $\eta=\epsilon=0$ and $\Delta\to\infty$, the simulation is exact.
Increasing the accuracy of the simulation will typically require expending more ``effort'', e.g.\ by increasing the energy of the interactions.

\cref{short:def:simulation} requires the simulating subspace to be the low-energy sector.
All our simulations achieve this.
But it is worth noting that \cref{short:def:simulation} can readily be generalised to other types of subspace, by replacing $P_{\le \Delta(H')}$ by a projector onto the subspace of interest.
Physically relevant examples might include symmetric subspaces, superselection sectors, or invariant subspaces of another Hamiltonian.
Constructing interesting simulations in such subspaces is an interesting direction for future research.

Our definition of Hamiltonian simulation, which follows from physical requirements, turns out to be a refinement of a definition of simulation introduced in prior work\shortcite{Bravyi-Hastings} in the context of Hamiltonian complexity theory.
There are two important differences.
We allow the encoding map $\cE$ to be anything that satisfies the physical requirements \labelcref{short:Hermitian-preserving,short:spectrum-preserving,short:local_interactions} from above, which can be more complicated than a single isometry.
On the other hand, we restrict $\cE$ to be local, since we require simulations to preserve locality.
A notion of universal analogue quantum simulation was also discussed -- though not formally defined -- in~\cite{hauke12}, along with some requirements that a quantum simulator should satisfy.
Our requirements encompass these.

Our notion of Hamiltonian simulation is strong enough to imply that all our requirements \labelcref{short:Hermitian-preserving,short:spectrum-preserving,short:local_interactions,short:observables,short:time-evolution,short:partition-function,short:errors} are indeed satisfied: all static, dynamic and thermodynamic properties are preserved up to any desired precision (see the next section and \textcitesupp{} for rigorous statements).

We are usually interested in simulating entire quantum many-body models, rather than individual Hamiltonians.
By ``model'', we mean very generally here any family of Hamiltonians. E.g.\ the 2D Heisenberg model consists of all Hamiltonians with nearest-neighbour Heisenberg interactions on a 2D square lattice of some given size, with uniform coupling strengths. The 2D Heisenberg model with variable couplings is another, more general model, consisting of all 2D Heisenberg Hamiltonians with any values for the individual coupling strengths.

When we say that a model $A$ can simulate another model $B$, we mean it in the following strong sense: any Hamiltonian $H$ on $n$ qudits (i.e.\ $d$-dimensional spins) from model $B$ can be simulated by some Hamiltonian $H'$ on $m$ qudits from model $A$, and this simulation can be done to \emph{any} precision $\eta,\epsilon$ with as large an energy cut-off $\Delta$ as desired.
The simulation is \emph{efficient} if each qudit of the original system is encoded into a constant number of qudits in the simulator (i.e.\ each $V_i$ in \cref{short:def:simulation} maps to $O(1)$ qudits); $P\ket{\psi}=\ket{\psi}$ for some state $\ket{\psi}$ that can be constructed efficiently; $H'$ is efficiently computable from $H$, and the energy overhead and number of qubits of the simulation scales at most polynomially (i.e.\ $\|H'\| = \poly(n,1/\eta,1/\epsilon,\Delta)$ and $m = \poly(n,1/\eta,1/\epsilon,\Delta)$).

\section{Consequences of simulation}

We arrived at a rigorous notion of Hamiltonian simulation by requiring the simulation to approximate the entire physics to arbitrary accuracy.
This is clearly very strong.
Just as exact simulation preserves all physical properties perfectly, approximate simulation preserves all physical properties approximately.
First, all energy levels are preserved up to any desired precision $\epsilon$.
Second, by locality of $\cE$, for any local observable $A$ on the original system there is a local observable $A'$ on the simulator and a local map $\cE_{\operatorname{state}}(\rho)$ such that applying $A'$ to $\cE_{\operatorname{state}}(\rho)$ perfectly reproduces the effect of $A$ applied to $\rho$.
This applies to all local observables, all order parameters (including topological order), and all correlation functions.
Thus all these static properties of the original Hamiltonian are reproduced by the simulation.

Third, Gibbs states of the original system correspond to Gibbs states of the simulator, and the partition function of $H$ is reproduced by $H'$, up to a physically irrelevant constant rescaling and an error that can be exponentially suppressed by increasing the energy cut-off $\Delta$ and improving the precision $\epsilon$.
More precisely, if the original and simulator Hamiltonians have local dimension $d$, then\citesupp{}
\begin{equation*}
    \frac{|\mathcal{Z}_{H'}(\beta)-(p+q)\mathcal{Z}_H(\beta)|}{(p+q)\mathcal{Z}_H(\beta)} \le \frac{d^{m-n} e^{-\beta\Delta}}{(p+q) e^{-\beta \|H\|}}+(e^{\epsilon\beta}-1) .
\end{equation*}
Since it is able to reproduce the partition function to any desired precision, all thermodynamic properties of the original Hamiltonian are reproduced by the simulation.
Finally, all dynamical properties are also reproduced to any desired precision.
More precisely, the error in the simulated time-evolution grows only linearly in time (which is optimal without active error correction), and can be suppressed to any desired level by improving the approximation accuracy $\epsilon$ and $\eta$:
\begin{equation}\label{short:eq:time-evolution}
  \| e^{-iH't} \mathcal{E}_{\operatorname{state}}(\rho) e^{iH't} - \mathcal{E}_{\operatorname{state}}( e^{-iHt} \rho e^{iHt}) \|_1 = O(t\epsilon + \eta).
\end{equation}

We can also derive some important consequences for simulation errors.
A recurring criticism of analogue Hamiltonian simulation is that, because it does not implement any error-correction, errors will accumulate over time and swamp the simulation.
A common counter-argument is that any real physical system is itself always subject to noise and errors.
If the properties of its Hamiltonian are sensitive to noise, the behaviour of the real physical system will include the effects of this, so from a physical perspective it is in fact fine to simulate this noisy system rather than an artificial, perfect, error-corrected system.

There is truth to both sides.
In the absence of error-correction, errors \emph{will} accumulate over time, as \cref{short:eq:time-evolution} shows.
It is also true that the same will happen in the original physical system, so this may not matter for simulating physical properties.
But only if noise and errors in the simulation closely mimic the noise and errors experienced by the real physical system we are trying to simulate.

With our precise definition of Hamiltonian simulation in hand, we can take a first step towards a rigorous version of this argument.
Most natural noise models are local: physical errors tend to act on nearby particles, not across the entire system.
The definition of Hamiltonian simulation we arrived at immediately implies that local errors in the original system correspond to local errors in the simulator.
But we can go further.
We prove\citesupp{} that, under a reasonable physical assumption, a local error in the simulator approximates arbitrarily well the encoded version of some local error on the original system.
To make this precise, note that if we take the energy cut-off $\Delta$ to be large enough, errors on the simulator system are unlikely to take the simulated state out of the low-energy space of $H'$.
Assume that this happens with probability at most $\delta$, for some $\delta \le \eta$.
Then for any noise operation $\mathcal{N}'$ acting on $\ell$ qudits of the simulator, there is always some noise operation $\mathcal{N}$ on at most $\ell$ qudits of the original system (which we can easily write down\citesupp{}) such that, for any state $\rho$, the effect of $\mathcal{N'}$ on the simulator approximates (again, to any desired precision) the effect of $\mathcal{N}$ on the original system:
\begin{equation*}
  \mathcal{E}_{\operatorname{state}}(\mathcal{N}(\rho)) = \mathcal{N}'(\mathcal{E}_{\operatorname{state}}(\rho)) + O(\sqrt{\eta})
\end{equation*}
where $\cN$ and $\cN'$ are superoperators.
The fact we can prove the result this way around is crucial: it shows that any local noise and errors in our simulator just reproduce the effects of local noise and errors in the original physical system.
This is much stronger than merely showing that errors on the original system can be simulated.

This is as strong a result as one can hope for in a fully general, abstract description of Hamiltonian simulation.
But it still falls far short of a full justification of the lack of error-correction in analogue quantum simulation.
Fully justifying this would require characterising all the noise and error processes occurring in the particular Hamiltonian simulator implementation, then determining whether these faithfully reproduce the effects of the natural noise and error processes in the physical system it is being used to simulate.
Our results provide the mathematical framework required to carry out the latter; the former is an experimental physics challenge.
Even then, the validity of this argument rests on the validity of the noise characterisation and model.
Ultimately, determining whether or not a simulation is accurate always comes down to testing its predictions in the laboratory.

\section{Universal Hamiltonians}

The notion of Hamiltonian simulation we have arrived at is extremely demanding.
It is not a priori clear whether any interesting simulations exist at all.
In fact, not only do such simulations exist, we prove that there are even \emph{universal} quantum simulators.
A model is ``universal'' if it can simulate \emph{any} Hamiltonian whatsoever, in the strong sense of simulation discussed above.
Depending on the target Hamiltonian, this simulation may or may not be efficient.
Typically, the simulation will be efficient for target Hamiltonians with local interactions in the same (or lower) spatial dimension.
Whereas, whilst universal models can also simulate Hamiltonians in higher spatial dimensions with only modest (polynomial) system-size overhead, this comes at an exponential cost in energy.
More precisely, any interaction graph that is spatially sparse can be simulated efficiently by any of the universal models in 2D, whereas the complete graph can be simulated with polynomial space-overhead but exponential energy-overhead (see \textcitesupp{} for details).

Remarkably, even certain simple 2D quantum spin-lattice models are universal.
To show this, we in fact prove a still stronger result.
We completely classify all two-qubit interactions (i.e.\ nontrivial interactions between two spin-1/2 particles) according to their simulation ability\citesupp{}.
This classification tells us which two-qubit interactions are universal.
The universal class turns out to be identical to the class of QMA-complete two-qubit interactions from quantum complexity theory\shortcite{Cubitt-Montanaro}, where QMA is the quantum analogue of the complexity class NP\shortcite{Kitaev-Shen-Vyalyi}.

The classification also shows that there are two other classes of two-qubit interaction, with successively weaker simulation ability.
Combining our Hamiltonian simulation results with previous work\shortcite{Bravyi-Hastings}, we find that there is a class of two-qubit interactions that can simulate any stoquastic Hamiltonian, i.e.\ any Hamiltonian whose off-diagonal entries in the standard basis are non-positive.
This is the class of Hamiltonians believed not to suffer from the sign-problem in numerical Monte-Carlo calculations.
Another class is able, by previous work\shortcite{Science}, to simulate any \emph{classical} Hamiltonian, i.e.\ any Hamiltonian that is diagonal in the standard basis.

The 2D Heisenberg- and XY-models with variable coupling strengths are important examples which we show fall into the first category, hence are universal simulators.
The 2D (quantum) Ising model with transverse fields falls into the second category, so can simulate any other stoquastic Hamiltonian\shortcite{Bravyi-Hastings}.
The 2D classical Ising model with fields falls into the third category, so is an example of a universal classical Hamiltonian simulator\shortcite{Science}.

\section{Universality classification}

We now summarise the proof of the universality classification result (see\citesupp{} for full technical details).
This involves chaining together a number of steps, the most important of which are shown in \cref{short:fig:reductions}.
In fact, most of the technical difficulty lies in proving universality of the Heisenberg and XY interactions, as these have the most restrictive symmetries of all two-qubit interactions.
Once these are shown to be universal, recently developed techniques\shortcite{Cubitt-Montanaro,Piddock-Montanaro} show that any other Hamiltonian from the universal category can simulate one of these two (this step is omitted from the illustration in \cref{short:fig:reductions}).
Hence, by universality of the Heisenberg or XY interactions, such Hamiltonians can also simulate any other Hamiltonian.

\begin{figure}
\begin{center}
\begin{tikzpicture}[yscale=0.8,font=\sffamily,every node/.style={draw,fill=yellow!30,inner sep=2pt},label/.style={draw=none,fill=none,font={\sffamily,\scriptsize}},>=stealth]
\node (heisenberg) at (-2,-0.5) {Heisenberg interactions};
\node (xy) at (2,-0.5) {XY interactions};
\node (xxzz) at (0,-1.5) {2\nbd-local Pauli interactions with no $\sigma_y$'s};
\node (4local) at (0,-2.5) {Arbitrary $(2k+1)$\nbd-local terms with no $\sigma_y$'s};
\node (real) at (0,-3.5) {Arbitrary real $2k$\nbd-local qubit Hamiltonian};
\node (klocal) at (0,-4.5) {Arbitrary $k$\nbd-local qubit Hamiltonian};
\node (bosons) at (-2,-5.5) {Bosons};
\node (qudits) at (0.3,-5.5) {Qudits};
\node (fermions) at (2,-5.5) {Fermions};
\draw[->] (heisenberg.south) -- node[draw=none,fill=none,midway,anchor=east,left=10pt] {\footnotesize \color{blue} Step 1} (xxzz);
\draw[->] (xy.south) -- node[draw=none,fill=none,midway,anchor=west,right=10pt] {\footnotesize \color{blue} Step 1}  (xxzz);
\draw[->] (xxzz) -- node[draw=none,fill=none,midway,anchor=east] {\footnotesize \color{blue} Step 2} (4local);
\draw[->] (4local) -- node[draw=none,fill=none,midway,anchor=east] {\footnotesize \color{blue} Step 3} (real);
\draw[->] (real) -- node[draw=none,fill=none,midway,anchor=east] {\footnotesize \color{blue} Step 4} (klocal);
\draw[->] (klocal) -- node[draw=none,fill=none,midway,anchor=east,left=8pt] {\footnotesize \color{blue} Step 5} (bosons);
\draw[->] (klocal) -- node[draw=none,fill=none,midway,anchor=east] {\footnotesize \color{blue} Step 5} (qudits);
\draw[->] (klocal) -- node[draw=none,fill=none,midway,anchor=west,right=8pt] {\footnotesize \color{blue} Step 5} (fermions);
\end{tikzpicture}
\end{center}
\caption{Part of the sequence of simulations used in this work. An arrow from one box to another indicates that a Hamiltonian of the first type can simulate a Hamiltonian of the second type.}
\label{short:fig:reductions}
\end{figure}

\paragraph{Step 1}
The Heisenberg interaction $h_\mathrm{Heis} = \sigma_x\ox\sigma_x + \sigma_y\ox\sigma_y + \sigma_z\ox\sigma_z$ (where $\sigma_{x,y,z}$ are the Pauli matrices) has full local rotational symmetry.
Mathematically, this is equivalent to invariance under arbitrary simultaneous local unitary rotations $U\ox U$.
The XY interaction $h_{XY} = \sigma_x\ox\sigma_x + \sigma_y\ox\sigma_y$ is invariant under arbitrary rotations in the z-plane, i.e.\ $U\ox U$ with $U=e^{i \theta \sigma_z}$ for any angle $\theta$.
Any Hamiltonian composed of just one of these types of interaction inherits the corresponding symmetry.
Thus all its eigenspaces also necessarily have this symmetry.
Yet if it is to be universal, it must simulate Hamiltonians without this symmetry.

To overcome the symmetry restriction, we develop more complicated simulations based around the use of ``perturbative gadgets'' (a technique originally introduced to prove QMA-completeness results in Hamiltonian complexity theory\shortcite{Kempe-Kitaev-Regev,Oliveira-Terhal}).
In a perturbative gadget, a heavily weighted term $C H_0$ (for some large constant $C$) dominates the overall Hamiltonian $H'=C H_0 + H_1$ such that the low-energy part of $H'$ is approximately just the ground space of $H_0$.
Within this low-energy subspace, an effective Hamiltonian is generated by $H_1$ and can be calculated using a precise version of perturbation theory\shortcite{Bravyi-Hastings}, which accounts rigorously for the approximation errors resulting from neglecting the higher-order terms.
The first-order term in the perturbative expansion is given by $H_1$ projected into the ground space of $H_0$, as one might expect.
But if this term vanishes, then the more complicated form of higher order terms may be exploited to generate more interesting effective interactions.

In order to break the symmetry of the Heisenberg and XY interactions, it is necessary for the encoded Hamiltonian to act not on the physical qubits of the system, but on qubits encoded into a subspace of multiple physical qubits.
To achieve this, we design a four-qubit gadget where the strong $H_0$ term, consisting of equally weighted interactions across all pairs of qubits, has a two-fold degenerate ground space.
This two-dimensional space can be used to encode a qubit.
This gadget is used repeatedly to encode all qubits of the systems separately, as illustrated in \cref{short:fig:logphys}.
We then add less heavily weighted interactions acting between qubits in different gadgets, in order to generate effective interactions between the encoded qubits.
This allows us to generate any two-qubit interaction that does not involve any $\sigma_y$ terms.

\begin{figure}
\begin{center}
\begin{tikzpicture}[scale=0.73]

\foreach \x in {1,2,...,4} {
\draw (\x,3) -- (\x,5);
}
\foreach \x in {3,4,5} {
\draw (1,\x) -- (4,\x);
}
\foreach \x in {1,2,...,4} {
\foreach \y in {3,4,5} {
\node[rqubit] at (\x,\y) {};
}
}

\begin{scope}[xshift=0.4cm,xscale=0.9]
\filldraw[thick,left color=red!50,right color=blue!50] (5,3.6) -- (7,3.6) -- (7,3.25) -- (8,4) -- (7,4.75) -- (7,4.4) -- (5,4.4) -- cycle;
\end{scope}

\begin{scope}[xshift=7.5cm]

\begin{scope}[xshift=-0.25cm,yshift=-0.8cm,scale=1.21]
\foreach \x in {1,2,...,4} {
\draw (\x,3) -- (\x,5);
}
\foreach \x in {3,4,5} {
\draw (1,\x) -- (4,\x);
}
\end{scope}

\begin{scope}[xshift=-0.5cm,yshift=-1cm,scale=1.2]
\foreach \x in {1,2,...,4} {
\foreach \y in {3,4,5} {
\begin{scope}[xshift=\x cm,yshift=\y cm,scale=0.25]
\fill[white] (0,-0.4) rectangle (2,1.732+0.5);
\draw[thick] (0,0) -- (2,0) -- (1,1.732) -- (1,0.577) -- (0,0) -- (1,1.732); \draw[thick] (2,0) -- (1,0.577);
\end{scope}
}
}
\foreach \x in {1,2,...,4} {
\foreach \y in {3,4,5} {
\begin{scope}[xshift=\x cm,yshift=\y cm,scale=0.25]
\node[bqubits] at (0,0) {};
\node[bqubits] at (2,0) {};
\node[bqubits] at (1,1.732) {};
\node[bqubits] at (1,0.577) {};
\end{scope}
}
}

\end{scope}
\end{scope}

\end{tikzpicture}
\caption{Schematic illustrating simulation of one Hamiltonian with another. Each logical (red) qubit is encoded within 4 physical (blue) qubits, forced into their ground space by strong pairwise interactions. Interactions between the physical qubits implement effective interactions between the logical qubits. An error on a physical qubit only affects one logical qubit.}
\label{short:fig:logphys}
\end{center}
\end{figure}

\paragraph{Steps 2+3}
The next steps use simpler perturbation gadgets, in which $H_0$ is used to project a system of ancilla qubits into a fixed state, such that the effective Hamiltonian that this generates couples the remaining qubits.
This type of gadget is known in the Hamiltonian complexity literature as a \emph{mediator} qubit gadget\shortcite{Oliveira-Terhal}, because the ancilla qubits are seen to ``mediate'' an effective interaction between the other qubits in the system.
Previously known gadgets of this type\shortcite{Oliveira-Terhal} allow many-body interactions to be simulated using two-body interactions.
We combine these with a new mediator gadget\citesupp{} to show how two-qubit Hamiltonians without $\sigma_y$'s can simulate all real local Hamiltonians.

\paragraph{Step 4}
There is still a more basic obstacle to overcome.
All matrix elements of $h_\mathrm{Heis}$ or $h_{XY}$ are real numbers (in the standard basis).
Thus any Hamiltonian built out of these interactions is also real (hence the lack of $\sigma_y$'s so far).
Yet if it is to be universal, it must simulate Hamiltonians with complex matrix elements.

A simple encoding overcomes this restriction, by adding an additional qubit and encoding the real and imaginary parts of $H$ separately, controlled on the state of the ancilla qubit.
The Hamiltonian $H'=\operatorname{Re}(H)\oplus \operatorname{Im}(H)$ is clearly real and is easily seen to be an encoding of $H$, since $H'=H \otimes \proj{+_y} + \bar{H} \otimes \proj{-_y}$, where $\ket{\pm_y} = (\ket{0} \pm i\ket{1})/\sqrt{2}$.
To make this encoding local, it can be adjusted to a simulation where there is an ancilla qubit for each qubit of the system, but these ancillas are forced by additional strong local interactions to be in $\operatorname{span}\{\ket{+_y}^{\otimes n},\ket{-_y}^{\otimes n}\}$.

\paragraph{Step 5}
Finally, higher-dimensional spins (qudits) can be simulated by encoding each qudit into $\lceil\log_2 d\rceil$ qubits in the obvious way. And to simulate indistinguishable particles, one can verify that standard techniques for mapping fermions or bosons to spin systems give the required simulations\citesupp{}.

To show that Hamiltonians with arbitrary long-range interactions can be simulated with a 2D lattice model, there is a further step: embedding an arbitrary interaction pattern within a square lattice.
This can be achieved by effectively drawing the long-range interactions as lines on the lattice, and using further perturbative gadgets to remove crossings between lines\shortcite{Oliveira-Terhal}.
This step requires multiple rounds of perturbation theory, which can result in the final Hamiltonian containing local interaction strengths that scale exponentially in the number of particles.
Thus the final simulation, whilst efficient in terms of the number of particles and interactions, is not necessarily efficient in terms of energy cost for arbitrary Hamiltonians.
For example, we do not know how to construct an energy-efficient simulation of a 3D lattice Hamiltonian using a 2D lattice model, nor do we necessarily expect it to be possible.
However, full efficiency is recovered when the original Hamiltonian is spatially sparse\shortcite{Oliveira-Terhal} (a class which encompasses all 2D lattice Hamiltonians).

\section{Conclusions}
We close by highlighting some of the limitations of our results, and possible future directions.
First, whilst our strong notion of simulation preserves locality in the sense that a few-particle observable in the original system will correspond to a few-particle observable in the simulator, simulating e.g.\ a 3D system in a 2D system necessarily means that the corresponding observables in the simulation will not always be on nearby particles.
Also, to simulate higher-dimensional systems in 2D, our constructions require very large coupling strengths.

From the analogue Hamiltonian engineering perspective, our results show that surprisingly simple types of interactions suffice for building a universal Hamiltonian simulator.
Together with the ability to prepare simple initial states, these would even suffice to construct a universal quantum computer, or to perform universal adiabatic quantum computation\citesupp{}.
(However, error correction and fault-tolerance, which are essential for scalable quantum computation, would require additional active control.)
The converse point of view is that, as these apparently restrictive models turn out to be universal, simulating them on a quantum computer may be more difficult than previously thought.

Furthermore, our mathematical constructions require extremely precise control over the strengths of individual local interactions across many orders of magnitude.
Though some degree of control is possible in state-of-the-art experiments\shortcite{NatPhys_simulation_review,georgescu14}, the requirements of our current universal models are beyond what is currently feasible.
On the other hand, it is already possible to experimentally engineer more complex interactions than those we have shown to be universal.
Now we have shown that universal models exist, and need not be extremely complex, it may be possible to construct other universal models tailored to particular experimental setups.

From a fundamental physics perspective, an important limitation of our current results is that our universal models are not translationally invariant.
Although we show there are universal models in which all interactions have an identical form, our proofs rely heavily on the fact that the \emph{strengths} of these interactions can differ from site to site.
Classic results showing that local symmetries together with translational-invariance can restrict the possible physics\shortcite{Mermin-Wagner,Hohenberg} suggest breaking translational-invariance may be crucial for universality.
On the other hand, much of the intuition for our proofs comes from Hamiltonian complexity, where recent results have shown that translational-invariance is no obstacle to complexity\shortcite{Gottesman-Irani,wheelbarrow}.

In light of our results, determining the precise boundary between simplicity and universality in quantum many-body physics is now an important open question for future research.



\clearpage

\part{Technical content}\label{part:technical}
\section{Notation and terminology}

As usual, $\cB(\cH)$ denotes the set of linear operators acting on a Hilbert space $\cH$.
For conciseness, we sometimes also use the notation $\cM_n$ for the set of all $n\times n$ matrices with complex entries.
$\Herm_n$ denotes the subset of all $n\times n$ Hermitian matrices.
$\1$ denotes the identity matrix.
For integer $n$, $[n]$ denotes the set $\{1,\dots,n\}$.

If $R,R'$ are rings, a \keyword{ring homomorphism} $\phi:R\to R'$ is a map that is both additive and multiplicative: $\forall a,b\in R: \phi(ab) = \phi(a)\phi(b)$ and $\phi(a+b) = \phi(a) + \phi(b)$.
Similarly, a \keyword{ring anti-homomorphism} is an additive map that is anti-multiplicative: $\phi(ab) = \phi(b)\phi(a)$.
If $\phi(\1)=\1$, we say the map is \keyword{unital}.

For a ring $R$, the corresponding \keyword{Jordan ring} $R_j$ is the ring obtained from $R$ by replacing multiplication with Jordan multiplication $\{ab\} := ab + ba$.
A \keyword{Jordan homomorphism} $\phi$ on $R$ is an additive map such that $\forall a,b\in R: \phi(ab+ba) = \phi(a)\phi(b)+\phi(b)\phi(a)$.
If $R$ is not of characteristic 2, this is equivalent to the constraint that $\forall a \in R: \phi(a^2) = \phi(a)^2$.
Note that any ring homomorphism is a Jordan homomorphism, but the converse is not necessarily true.

$\spec(A)$ denotes the spectrum of $A\in\cM_n$, i.e.\ the set of values $\lambda\in\C$ such that $A - \lambda\1$ is not invertible.
(This of course coincides with the set of eigenvalues, ignoring multiplicities.)
We say that $\phi:\cM_n\to\cM_m$ is \keyword{invertibility-preserving} if $\phi(A)$ is invertible in $\cM_m$ for all invertible $A\in\cM_n$.
We say that $\phi$ is \keyword{spectrum-preserving} if $\spec(\phi(A)) = \spec(A)$ for all $A\in\cM_n$.

For an arbitrary Hamiltonian $H \in \mathcal{B}(\C^d)$, we let $P_{\le \Delta(H)}$ denote the orthogonal projector onto the subspace $S_{\le \Delta(H)} := \linspan \{ \ket{\psi} : H\ket{\psi}=\lambda\ket{\psi}, \lambda \le \Delta \}$.
We also let $H'|_{\le \Delta(H)}$ denote the restriction of some other arbitrary Hamiltonian $H'$ to $S_{\le \Delta(H)}$, and write $H|_{\le \Delta} := H|_{\le \Delta(H)}$ and $H_{\le \Delta} := H P_{\le \Delta(H)}$.

We say that a Hamiltonian $H \in \mathcal{B}((\C^d)^{\ox n})$ is \keyword{$k$-local} if it can be written as a sum of terms such that each $h_i$ acts non-trivially on at most $k$ subsystems of $(\C^d)^{\ox n}$.
That is, $h_i\in\cB((\C^d)^{\ox k})$ and $H = \sum_i h_i\ox\1$ where the identity in each term in the sum acts on the subsystems where that $h_i$ does not.
An operator on a composite Hilbert space ``acts trivially'' on the subsystems where it acts as identity, and ``acts non-trivially'' on the remaining subsystems.
We will often employ a standard abuse of notation, and implicitly exend operators on subsystems to the full Hilbert without explicitly writing the tensor product with identity, allowing us e.g.\ to write simply $H=\sum h_i$.
We say that $H$ is \keyword{local} if it is $k$-local for some $k$ that does not depend on $n$\footnote{Technically, this makes sense only for families of Hamiltonians $H$, where we consider $n$ to be growing.}.

We let $X$, $Y$, $Z$ denote the Pauli matrices and often follow the condensed-matter convention of writing $XX$ for $X \ox X$ etc.
For example, $XX+YY+ZZ$ is short for $X \ox X + Y \ox Y + Z \ox Z$ and is known as the Heisenberg (exchange) interaction.
The XY interaction is $XX + YY$.

Let $M$ be a $k$-qudit Hermitian matrix.
We say that $U \in SU(d)$ \emph{locally diagonalises} $M$ if $U^{\ox k} M (U^\dg)^{\ox k}$ is diagonal.
We say that a set $\mathcal{S}$ of Hermitian matrices is simultaneously locally diagonalisable if there exists $U \in SU(d)$ such that $U$ locally diagonalises $M$ for all $M \in \mathcal{S}$.
Note that matrices in $\mathcal{S}$ may act on different numbers of qudits, so can be of different sizes.

We will often be interested in families of Hamiltonians.
For a subset $\mathcal{S}$ of interactions (Hermitian matrices on a fixed number of qudits), we define the family of $\mathcal{S}$-Hamiltonians to be the set of Hamiltonians which can be written as a sum of interaction terms where each term is either picked from $\mathcal{S}$, with an arbitrary positive or negative real weight, or is an arbitrarily weighted identity term.
For example, $H$ is a $\{ZZ\}$-Hamiltonian if it can be written in the form $H = \alpha \1 + \sum_{i < j} \beta_{ij} Z_i Z_j$ for some $\alpha,\beta_{ij} \in \R$.
A \keyword{model} is a (possibly infinite) family of Hamiltonians.
Typically the Hamiltonians in a model will be related in some way, e.g.\ all Hamiltonians with nearest-neighbour Heisenberg interactions on an arbitrarily large 2D lattice (the ``2D Heisenberg model'').


\section{Hamiltonian encodings}
Any non-trivial simulation of one Hamiltonian with another will involve encoding the first within the second in some way.
Write $H' = \cE(H)$ for some ``encoding'' map $\cE$ that encodes a Hamiltonian $H$ into some Hamiltonian $H'$.
Any such encoding should fulfil at least the following basic requirements.
First, any observable on the original system should correspond to an observable on the simulator system.
Second, the set of possible values of any encoded observable should be the same as for the corresponding original observable.
In particular, the energy spectrum of the Hamiltonian should be preserved.
Third, the encoding of a probabilistic mixture of observables should be the same as a probabilistic mixture of the encodings of the observables.

To see why this last requirement holds, imagine that we are asked to encode observable $A$ with probability $p$, and observable $B$ with probability $1-p$.
Then, for any state $\rho$ on the simulator system, the expected value of the encoded observable acting on $\rho$ should be the same as the corresponding probabilistic mixture of the expected values of the encoded observables $A$ and $B$ acting on $\rho$.
In order for this to hold for all states $\rho$, we need the mixture of observables $pA + (1-p)B$ to be encoded as the corresponding probabilistic mixture of encodings of $A$ and $B$.

These operational requirements correspond to the following mathematical requirements on the encoding map $\mathcal{E}$:
\begin{enumerate}[label=\arabic{*}.,ref=\arabic{*}]
\item
  $\cE(A) = \cE(A)^\dg$ for all $A \in \Herm_n$.
\item
  $\spec(\cE(A)) = \spec(A)$ for all $A \in \Herm_n$.
\item
  $\cE(pA+(1-p)B)=p\cE(A)+(1-p)\cE(B)$ for all $A,B \in \Herm_n$ and all $p \in [0,1]$.
\end{enumerate}

Of course, there are many other desiderata that we would like $\mathcal{E}$ to satisfy, such as preserving the partition function, measurement outcomes, time-evolution, local errors, and others.
For the Hamiltonian itself, we almost certainly want $\cE$ to not only be convex, but also real-linear: $\cE(\sum_i\alpha_i h_i) = \sum_i\alpha_i\cE(h_i)$, so that a Hamiltonian expressed as a sum of terms can be encoded by encoding the terms separately.
However, we will see later that meeting just the above three basic requirements necessarily implies also meeting all these other operational requirements (which we will make precise).

It turns out there is a simple and elegant characterisation of what such encodings have to look like.
To prove this, we will need the following theorem concerning Jordan ring homomorphisms.

\begin{theorem}[follows from~\cite{Jacobson-Rickart}, Theorem 4 and~\cite{Martindale}, Theorem 2]
  \label{Jordan_homomorphisms}
  For any $n \ge 2$, any Jordan homomorphism of the Jordan ring $\Herm_n$ can be extended in one and only one way to a homomorphism of the matrix ring $\mathcal{M}_n$.
\end{theorem}

\Cref{Jordan_homomorphisms} was shown by Jacobson and Rickart for $n \ge 3$~\cite{Jacobson-Rickart}, and by Martindale for $n = 2$~\cite{Martindale}, in each case in a far more general setting than we need here.

\begin{lemma}\label{invertibility-preserving}
  Any unital, invertibility-preserving, real-linear map $\phi:\Herm_n\to\Herm_m$ is a Jordan homomorphism.
\end{lemma}

\begin{proof}
  The argument is standard (see e.g.~\cite{Hou-Semrl}).

  $\phi(H-\lambda\1) = \phi(H)-\lambda\1$, thus $\spec(\phi(H))\subseteq\spec(H)$ since $\phi$ is invertibility-preserving.
  In particular, $\spec(\phi(P))\in\{0,1\}$ for every projector $P$.
  Since $\phi(P)$ is also Hermitian, this implies $\phi(P)$ is a projector.

  By the spectral decomposition, any $H\in\Herm_n$ can be decomposed as $H = \sum_i\lambda_i P_i$ where $P_i$ are mutually orthogonal projectors and $\lambda_i\in\R$.
  For $i\neq j$, $P_i+P_j$ is a projector, thus $\phi(P_i+P_j)$ is a projector and $(\phi(P_i+P_j))^2 = \phi(P_i) + \phi(P_j)$, so that $\phi(P_i)\phi(P_j)+\phi(P_i)\phi(P_j) = 0$.
  Therefore, $\phi(H)^2 = \sum_i\lambda_i^2\phi(P_i)^2 + \sum_{i\neq j}\lambda_i\lambda_j\phi(P_i)\phi(P_j) = \sum_i\lambda_i^2\phi(P_i) = \phi(H^2)$.
\end{proof}

\begin{theorem}[Encodings]\label{encoding}
  For any map $\cE:\Herm_n\to\Herm_m$, the following are equivalent:
  \begin{enumerate}
  \item \label[part]{encoding:operational_subset}%
    For all $A,B \in \Herm_n$, and all $p \in [0,1]$:
    \begin{enumerate}[label=\arabic{*}.,ref=\ref{encoding:operational_subset}\arabic{*}]
    \item \label[condition]{Hermitian-preserving}%
      $\cE(A) = \cE(A)^\dg$
    \item \label[condition]{spectrum-preserving}%
      $\spec(\cE(A)) = \spec(A)$
    \item \label[condition]{local_interactions}%
      $\cE(pA+(1-p)B) = p\cE(A)+(1-p)\cE(B)$.
    \end{enumerate}

  \item \label[part]{encoding:mathematical}%
    There exists a unique extension $\cE': \cM_n \to \cM_m$ such that $\cE'(H) = \cE(H)$ for all $H \in \Herm_n$ and, for all $A,B\in\cM_n$ and $x\in\R$:
    \begin{enumerate}[label=\alph{*}.,ref=\ref{encoding:mathematical}\alph{*}]
    \item $\cE'(\1) = \1$ \label[condition]{encoding:mathematical:unital}
    \item $\cE'(A^\dg) = \cE'(A)^\dg$ \label[condition]{encoding:mathematical:dagger}
    \item $\cE'(A+B) = \cE'(A) + \cE'(B)$ \label[condition]{encoding:mathematical:sum}
    \item $\cE'(AB) = \cE'(A)\cE'(B)$ \label[condition]{encoding:mathematical:product}
    \item $\cE'(xA) = x\cE'(A)$. \label[condition]{encoding:mathematical:real-linear}
    \end{enumerate}

  \item \label[part]{encoding:characterisation}%
    There exists a unique extension $\cE': \cM_n \to \cM_m$ such that $\cE'(H) = \cE(H)$ for all $H \in \Herm_n$ with $\cE'$ of the form
    \begin{equation}
    \label{eq:characterisation}
      \cE'(M) = U \left(M^{\oplus p} \oplus \bar{M}^{\oplus q}\right) U^\dg
    \end{equation}
    for some non-negative integers $p$, $q$ and unitary $U\in\cM_m$, where $M^{\oplus p} := \bigoplus_{i=1}^p M$ and $\bar{M}$ denotes complex conjugation.
  \end{enumerate}
  We call a map $\cE$ satisfying \labelcref{encoding:operational_subset,encoding:mathematical,encoding:characterisation} an \keyword{encoding}.
\end{theorem}

Note that \labelcref{encoding:characterisation} is basis-independent, despite the occurrence of complex conjugation; taking the complex conjugation with respect to a different basis is equivalent to modifying $U$, which just gives another encoding.
Given that $\mathcal{E}'$ is unique, for the remainder of the paper we simply identify $\mathcal{E}'$ with $\mathcal{E}$.
In particular, this allows us to assume that $\cE$ is of the form specified in \cref{encoding:characterisation}.
The characterisation \cref{eq:characterisation} can equivalently be written as
\begin{equation}\label{eq:characterisationproj}
  \cE'(M) = U \left(M \ox P + \bar{M} \ox Q \right) U^\dg
\end{equation}
for some orthogonal projectors $P$ and $Q$ such that $P+Q=\1$; this alternative form will sometimes be useful below. We think of the system on which $P$ and $Q$ act as an ancilla, and often label this ``extra'' subsystem by the letter $E$.

\begin{proof}
  \labelcref{encoding:operational_subset}~$\Rightarrow$~\labelcref{encoding:mathematical}:\\
  We first show that $\mathcal{E}$ is a Jordan homomorphism.
  \Cref{Hermitian-preserving} states that $\cE$ preserves $\Herm_n$, and \cref{spectrum-preserving} implies that $\cE$ is unital and invertibility-preserving on $\Herm_n$, with $\cE(0)=0$.
  We next check that $\cE(0)=0$ together with \cref{local_interactions} are equivalent to real-linearity of $\cE$.
  For any $\lambda<0$, setting $p=\lambda/(\lambda-1)$, $B=pA/(p-1)$ and using \cref{local_interactions} gives
  \begin{equation}\label{eqn:homogeneous}
    0 = \cE(0)=p\cE(A)+(1-p)\cE(pA/(p-1)) \: \Leftrightarrow\: \lambda\cE(A)=\cE(\lambda A).
  \end{equation}
  Apply \cref{eqn:homogeneous} to $\lambda A$ to get $\cE(\lambda^2 A)=\lambda^2\cE(A)$, showing that $\cE$ is homogeneous for all real scalars.
  Additivity follows by combining \cref{local_interactions} and homogeneity: $\cE(A+B)=\cE(2A)/2+\cE(2B)/2=\cE(A)+\cE(B)$.
  Therefore $\cE$ is also real-linear so by \cref{invertibility-preserving} $\cE$ is a Jordan homomorphism.

  By \cref{Jordan_homomorphisms}, there exists a unique homomorphism $\mathcal{E}': \cM_n \to \cM_m$ such that $\cE'(H) = \cE(H)$ for all $H \in \Herm_n$.
  As $\cE'$ agrees with $\cE$ on $\Herm_n$, it satisfies \labelcref{encoding:mathematical:unital}.
  As $\cE'$ is a homomorphism, it satisfies \labelcref{encoding:mathematical:sum,encoding:mathematical:product} by definition; this also implies that $\cE'(xA) = \cE'(x\1)\cE'(A) = \cE(x\1)\cE'(A) = x \cE'(A)$ for any $x \in \R$, so \labelcref{encoding:mathematical:real-linear} holds.

  We finally prove \labelcref{encoding:mathematical:dagger}.
  It is sufficient to show that $\cE'(i \1)^\dg = -\cE'(i \1)$, because if this holds we can expand any matrix $A \in \mathcal{M}_n$ as $A = B + i C$ for some Hermitian matrices $B$ and $C$ to obtain
  \begin{align}
    \cE'(A^\dg)
      &= \cE'(B - i C) = \cE'(B) - \cE'(C) \cE'(i\1) = \cE'(B)^\dg + \cE'(C)^\dg \cE'(i \1)^\dg\\
      &= \cE'(B + iC)^\dg  = \cE'(A)^\dg.
  \end{align}
  To show $\cE'(i \1)^\dg = -\cE'(i \1)$, we first write $i\1$ as a linear combination of products of Hermitian matrices.
  That this can be done is an immediate consequence of the fact that $\mathcal{M}_n$ is the enveloping associative ring of $\Herm_n$.
  However, it can also be seen explicitly by writing
  \begin{equation}
    i \proj{j} = \proj{j} (i \ketbra{j}{k} - i \ketbra{k}{j})(\ketbra{j}{k} + \ketbra{k}{j})
  \end{equation}
  for any $j$, and some $k \neq j$; summing this product over $j$, we obtain $i\1$.
  Thus we can write $i\1 = \sum_j A_j B_j C_j$ for Hermitian matrices $A_j$, $B_j$, $C_j$.
  By taking adjoints on both sides, it follows that $-i\1 = \sum_j C_j B_j A_j$.
  So we have
  \begin{align}
    \cE'(i \1)^\dg
      &= \cE'\Big(\sum_j A_j B_j C_j\Big)^\dg = \Big(\sum_j \cE(A_j) \cE(B_j) \cE(C_j)\Big)^\dg\\
      &= \sum_j \cE(C_j) \cE(B_j) \cE(A_j) =  \cE'\Big(\sum_j C_j B_j A_j\Big)\\
      &= \cE'(-i\1) = -\cE'(i\1).
  \end{align}

  \noindent\labelcref{encoding:mathematical}~$\Rightarrow$~\labelcref{encoding:characterisation}:\\
  Existence and uniqueness of $\mathcal{E}'$ were already shown in the previous part. In the proof of the remaining claim, for readability we just use $\mathcal{E}$ to denote this unique extension. First define the complex structure $J:= \cE(i\1) \equiv \cE(i)$ (where the latter notation is a convenient shorthand). We have
  \begin{equation}
    J^2 = \cE(i)\cE(i) = \cE(i^2) = \cE(-1) = -\1,
  \end{equation}
  thus $J$  has eigenvalues $\pm i$.
  Furthermore,
  \begin{equation}
    J^\dg = \cE(i)^\dg = \cE(i^\dg) = -\cE(i) = -J,
  \end{equation}
  so $J$ is anti-Hermitian, hence diagonalisable by a unitary transformation.

  For any $A\in\Herm_n$, we have
  \begin{equation}
    J\cE(A) = \cE(i)\cE(A) = \cE(iA) = \cE(Ai) = \cE(A)J,
  \end{equation}
  so that $[\cE(A),J] = 0$.
  Thus $\cE(A)$ and $J$ are simultaneously diagonalisable for all $A$.
  $\HS = \HS_+\oplus \HS_-$ therefore decomposes into a direct sum of the $\pm i$ eigenspaces of $J$, on which $\cE(A) = A_+\oplus A_-$ acts invariantly.

  Now, restricting to either of these invariant subspaces,
  \begin{gather}
    \cE(A)|_\pm = A_\pm \\
    \cE(iA)|_\pm = J A_\pm = \pm i A_\pm \\
    \cE(AB)|_\pm   = \cE(A)\cE(B)|_\pm = A_\pm B_\pm \\
    \cE(A^\dg)|_\pm = \cE(A)^\dg|_\pm = A_\pm^\dg.
  \end{gather}
  Thus $\cE = \cE_+\oplus\cE_-$ decomposes into a direct sum of a *\nbd-representation $\cE_+(A) := \cE(A)|_+$ and an anti\nbd-*\nbd-representation\footnote{By ``anti\nbd-*\nbd-representation'' we mean an anti-linear algebra homomorphism, \emph{not} a *\nbd-antihomomorphism (which would be a linear map preserving adjoints that reverses the order of multiplication).}
  $\cE_-(A) := \cE(A)|_-$.
  Since for any vector $\ket{\psi}\in\C^m$, $\cE_\pm(\1)\ket{\psi} = \1\ket{\psi} = \ket{\psi}$, these (anti\nbd-)*\nbd-representations are necessarily non-degenerate.

  By a standard result on the representations of finite-dimensional C*\nbd-algebras \cite[Corollary~III.1.2]{Davidson_C*}, any non-degenerate *\nbd-representation of $\cM_n$ is unitarily equivalent to a direct sum of identity representations.
  If $\phi$ is an anti-*\nbd-homomorphism, let $\varphi(A) := \overline{\phi(A)}$.
  Then $\varphi(iA)= \overline{\phi(iA)} = \overline{-i\phi(A)} = i\varphi(A)$, $\varphi(A+B) = \varphi(A) + \varphi(B)$, $\varphi(A^\dg) = \varphi(A)^\dg$, and $\varphi(AB) = \varphi(A)\varphi(B)$.
  Thus $\phi(A) = \overline{\varphi(A)}$ where $\varphi$ is a *\nbd-homomorphism.
  Therefore, any non-degenerate anti-*\nbd-representation is unitarily equivalent to a direct sum of complex conjugates of identity representations, which completes the argument.

  \labelcref{encoding:characterisation}~$\Rightarrow$~\labelcref{encoding:operational_subset} can readily be verified directly.
\end{proof}

The above theorem characterises encodings of observables.
This immediately tells us how to encode physical systems themselves, expressed as Hamiltonians: since the Hamiltonian itself is an observable, the encoding map must have the same characterisation.

It is easy to see from the characterisation in \cref{encoding:characterisation} of the \namecref{encoding} that any encoding preserves all interesting physical properties of the original Hamiltonian.
For example, the set of eigenvalues is preserved, up to possibly duplicating each eigenvalue the same number of times, implying preservation of the partition function (up to an unimportant constant factor).
It is also easy to see that any encoding $\mathcal{E}$ properly encodes arbitrary quantum channels: if $\{E_k: \sum_k E_k^\dg E_k = \1\}$ are the Kraus operators of the channel, then
\begin{equation}
  \label{eq:channels}
  \sum_k\cE(E_k)^\dg\cE(E_k) = \1.
\end{equation}


\subsection{A map on states, $\cEs$}
We now show that, for any encoding $\mathcal{E}$, there exists a corresponding map $\mathcal{E}_{\operatorname{state}}$ that encodes quantum states $\rho$ such that encoded observables $\mathcal{E}(A)$ applied to encoded states $\mathcal{E}_{\operatorname{state}}(\rho)$ have correct expectation values.

First, note that for any observable $A$ and any state $\rho'$ on the simulator system, we have
\begin{align}
  \tr (\cE(A) \rho')
  &=\tr[U(A\ox P +\bar{A}\ox Q)U^\dg \rho']\\
  &=\tr[(A\ox \1)(\1\ox P)U^\dg\rho'U)]+\tr[(\bar{A}\ox \1)(\1\ox Q)U^\dg\rho'U]\\
  &=\tr[AF(\rho')]+\tr[\bar{A}\;\overline{B(\rho')})]=\tr(A\rho)
  \label{eqn:Fmeas}
\end{align}
where
\begin{gather} \label{eq:fandb}
  F(\rho')=\tr_E[(\1\ox P) U^\dg\rho'U],\qquad
  B(\rho')=\overline{\tr_E[(\1\ox Q) U^\dg\rho'U]},\\
  \rho=F(\rho')+B(\rho')
\end{gather}
and we label the second subsystem $E$ as discussed after \cref{eq:characterisationproj}. Note that $F(\rho')$ and $B(\rho')$ are both positive but not necessarily normalised, but $\rho$ is normalised.

Therefore any map $\mathcal{E}_{\operatorname{state}}(\rho)$ on states $\rho$ such that $\rho = F(\mathcal{E}_{\operatorname{state}}(\rho))+B(\mathcal{E}_{\operatorname{state}}(\rho))$ will preserve measurement outcomes appropriately.
One natural choice is
\begin{equation}
  \label{eq:estate}
  \mathcal{E}_{\operatorname{state}}(\rho) =
  \begin{cases}
    U(\rho \ox \sigma)U^\dg \text{ for some $\sigma$ such that $P\sigma = \sigma$} & \text{if $P \neq 0$}\\
    U(\bar{\rho} \ox \sigma)U^\dg \text{ for some $\sigma$ such that $Q\sigma = \sigma$}  & \text{otherwise.}
  \end{cases}
\end{equation}
Then in the former case $F(\mathcal{E}_{\operatorname{state}}(\rho)) = \rho$, $B(\mathcal{E}_{\operatorname{state}}(\rho)) = 0$; and in the latter case the roles of $F$ and $B$ are reversed.

We note that for $\cE_{\operatorname{meas}}$ to be practically implementable, we need both that the unitary $U$ is not too complex and that the state $\sigma$ is easy to prepare. This will be formalised when we introduce the notion of an efficient simulation (\Cref{dfn:sim}).

We now show that $\mathcal{E}_{\operatorname{state}}$ simulates time-evolution correctly too.
We have
\begin{align}
  F(e^{-i\cE(H)t}\rho'e^{i\cE(H)t}) &=e^{-iHt}F(\rho')e^{iHt}, \\
  B(e^{-i\cE(H)t}\rho'e^{i\cE(H)t}) &=e^{iHt}B(\rho')e^{-iHt}.
\end{align}
This is why they are labelled with the letters $F$ and $B$: the $F$ part evolves \emph{forwards} in time while the $B$ part evolves \emph{backwards} in time.
Taking $\rho' = \mathcal{E}_{\operatorname{state}}(\rho)$, we have proven the following result.

\begin{proposition}
  \label{prop:encodingswork}
  For any encoding $\mathcal{E}$, the corresponding map $\mathcal{E}_{\operatorname{state}}$ satisfies the following:
  \begin{enumerate}
  \item \label[requirement]{expectation_values}%
    $\tr\left(\cE(A)\cE_{\operatorname{state}}(\rho)\right) = \tr(A\rho)$
  \item \label[requirement]{time_evolution}%
    For any time $t$,
    \begin{equation}
      e^{-i \mathcal{E}(H) t} \cE_{\operatorname{state}}(\rho) e^{i \mathcal{E}(H) t} =
      \begin{cases}
        \cE_{\operatorname{state}}(e^{-iHt} \rho e^{iHt}) & \text{if $p \ge 1$}\\
        \cE_{\operatorname{state}}(e^{iHt} \rho e^{-iHt}) & \text{if $p = 0$.}
      \end{cases}
    \end{equation}
  \end{enumerate}
\end{proposition}

It is worth highlighting the last point.
We see that if $p \ge 1$, evolving according to $\mathcal{E}(H)$ for time $t$ simulates evolving according to $H$ for time $t$, as we would expect; but that if $p = 0$, we simulate evolution according to $H$ for time $-t$.
That is, if our encoding only includes copies of $\bar{H}$, we simulate evolution backwards in time.
To avoid this issue, we define the concept of a \emph{standard} encoding as one where $p \ge 1$, and hence which is able to simulate evolution forward in time.

\begin{definition}[Standard encoding]
  An encoding $\mathcal{E}(M) = U (M^{\oplus p} \oplus \bar{M}^{\oplus q}) U^\dg$ is a \emph{standard encoding} if $p\ge 1$.
\end{definition}

\subsubsection{Gibbs-preserving state mappings}
The choice of $\cE_{\operatorname{state}}$ in \cref{eq:estate} is convenient, as it allows us to use the same mapping $\cE$ for both the Hamiltonian and for observables.
However, it does not map Gibbs states $e^{-\beta H}/\tr(e^{-\beta H})$ of the original system to Gibbs states $e^{-\beta' H'}/\tr(e^{-\beta' H'})$ of the simulator.
If we have limited ability to manipulate or prepare states of the simulator, it may be difficult to prepare a state of the form \cref{eq:estate}.
At equilibrium, the system will naturally be in a Gibbs state.
From this perspective, it would be more natural if the state mapping identified Gibbs states of the original system with Gibbs states of the simulator.

An alternative choice of $\cE_{\operatorname{state}}$ \emph{does} map Gibbs states to Gibbs states:
\begin{equation}\label{eq:estate_Gibbs}
  \cE_{\operatorname{state}}(\rho) =
 \frac{\cE(\rho)}{\tr[\cE(\rho)]}=\frac{1}{p+q}U(\rho \ox P +\bar{\rho} \ox Q)U^{\dag}
\end{equation}
where $p=\tr(P)$ and $q=\tr(Q)$.
However, to obtain the correct measurement outcome probabilities, we now need to choose a slightly different mapping for observables:\footnote{The Hamiltonian is of course also an observable.
  With this choice of state mapping, to construct the simulator Hamiltonian we must still use the mapping $H' = \cE(H)$.
  But if we want to \emph{measure} the Hamiltonian -- i.e.\ carry out the measurement on the simulator that corresponds to measuring the energy of the original system -- we must measure $\cE_{\operatorname{meas}}(H)$.}
\begin{equation}\label{eq:emeas_Gibbs}
  \cE_{\operatorname{meas}}(A) =
  \begin{cases}
    \frac{p+q}{p}U(A \ox P)U^\dg  \text{ if } P \neq 0 \\
    \frac{p+q}{q}U(\bar{A} \ox Q)U^\dg  \text{ otherwise.}
  \end{cases}
\end{equation}

For simplicity, in the remainder of the paper we will state and prove our results for the choice of state mapping $\cE_{\operatorname{state}}$ from \cref{eq:estate}, so that both Hamiltonians and observables are encoded by $\cE$.
However, our results also go through with the appropriate minor modifications for the choice of Gibbs-preserving $\cE_{\operatorname{state}}$ from \cref{eq:estate_Gibbs}, where the simulator Hamiltonian is still constructed using $\cE$ but observables are encoded by the $\cEm$ from \cref{eq:emeas_Gibbs}.

Note that $\cEm$ has been chosen so that measuring $\cEm(A)$ will only pick up the $F(\rho')$ part of a state $\rho'$ on the simulator. We therefore include results concerning the behaviour of $F$, in order to cover the choice of $\cEs$ given in \cref{eq:estate_Gibbs}, as well other mappings on states.


\subsection{The complex-to-real encoding}

The only nontrivial encoding (as opposed to simulation, q.v.) that we will need to use is an encoding of complex Hamiltonians as real Hamiltonians.

\begin{lemma}\label{complex-to-real enc}
  There exists an encoding $\varphi$ such that for any Hamiltonian $H\in\cB(\C^d)$, the encoded Hamiltonian $H'=\varphi(H)\in\cB(\R^{2d})$ is \emph{real}.
\end{lemma}

\begin{proof}
  This follows from the canonical Hilbert space isomorphism $\C^d \simeq \R^{2d}$ where the latter is endowed with a linear complex structure~$J$.

  Concretely, let
  \begin{equation}
    J:= \begin{pmatrix} 0 & \1_d \\ -\1_d & 0 \end{pmatrix}=i Y \ox \1_d
  \end{equation}
  where where $\1_d$ is the $d\times d$ identity matrix, and define the mapping
  \begin{equation}
    \begin{aligned}
    \varphi:\quad &\cB(\C^d) &\to\quad &\cB(\R^{2d})\\
    &\varphi(M) &=\quad &\real M\oplus\real M + J \imag M\oplus\imag M.
    \end{aligned}
  \end{equation}

  To see that $\varphi$ is indeed a valid encoding, we can either verify directly that it satisfies all the properties listed in \cref{encoding:operational_subset} of \cref{encoding}, or observe that
  \begin{equation}
    \varphi(M) = U (M\oplus\bar{M}) U^\dg
    \quad\text{where}\quad
    U = \frac{1}{\sqrt{2}}\begin{pmatrix} \phantom{i}\1 & \phantom{-i}\1 \\ i\1 & -i\1 \end{pmatrix}
      = \frac{1}{\sqrt{2}}\begin{pmatrix} 1 & \phantom{-}1 \\ i & -i \end{pmatrix} \ox \1,
  \end{equation}
  which is manifestly of the form given in \cref{encoding:characterisation} of \cref{encoding}.
  The \namecref{complex-to-real enc} follows by setting
  \begin{equation}
    H' = \varphi(H) = \real(H)\oplus\real(H) + J \imag(H)\oplus\imag(H).
  \end{equation}
\end{proof}

When applied to a Hamiltonian on a system of $n$ qubits, the encoding of \cref{complex-to-real enc} is local (see \cref{sec:local_encodings}).
Indeed, it produces a Hamiltonian $H'$ on $n+1$ qubits, given by
\begin{equation}
  H'=\ket{+_y}\bra{+_y} \ox H +\ket{-_y}\bra{-_y} \ox \bar{H}
\end{equation}
where $\ket{\pm_y}=(\ket{0}\pm i \ket{1})/\sqrt{2}$ are the eigenstates of $Y$.
It is easy to see that $H'$ is real since $\overline{\ket{+_y}}=\ket{-_y}$.
Any complex $k$\nbd-local interaction is mapped to a $(k+1)$\nbd-local interaction involving the additional qubit.

This additional qubit therefore has a special significance in the construction, which leads to two unwanted consequences.
Firstly, the interaction graph of $H'$ is in general more complicated than that of $H$.
Any geometric locality or spatial sparsity in the original Hamiltonian $H$ is lost, as all complex local terms are mapped to interactions in $H'$ that involve this additional qubit.
Secondly, an error on this single additional qubit would mix the spaces where $H$ and $\bar{H}$ act.
This could lead to unusual errors when simulating the time evolution of $\rho$ under $H$ with the simulator $H'$.

In \cref{complex-to-real sim} below we give an alternative to this encoding that avoids these problems.


\subsection{Local encodings}\label{sec:local_encodings}

So far, we have considered encodings of arbitrary Hamiltonians, with no additional structure.
However, in Hamiltonian simulation, we are typically interested in many-body Hamiltonians composed of local interactions between subsets of particles.
That is, Hamiltonians $H \in \cB((\C^d)^{\ox n})$ with $H = \sum_i h_{S_i}$, where the local terms $h_{S_i} \in \cB((\C^d)^{\ox\abs{S_i}})$ act on subsets $S_i$ of the particles (implicitly extended to $\cB((\C^d)^{\ox n})$ in the sum by tensoring with identity on the rest of the space, as usual).

In this case, we typically want our encoding to be \emph{local}, i.e.\ it should map local observables to local observables, and consequently
\begin{equation}
  \cE(h_{S_i}\ox\1) = h'_{S'_i}\ox\1
\end{equation}
so that the simulation $H' = \cE(H) = \sum_i h'_{S'_i}$ is itself a local Hamiltonian.

\begin{definition}[Local encoding]\label{def:local_encoding}
  Let $\cE:\cB(\bigotimes_{i=1}^n\HS_i) \to \cB(\bigotimes_{i=1}^{n'}\HS'_i)$ be an encoding, and let $\{S'_i\}_{i=1}^n$ be subsets of $[n']$.
  We say that the encoding is \emph{local} with respect to $\{S'_i\}$ if for any operator $A\in\cB(\HS_i)$, $\cE(A\ox\1)$ acts non-trivially only on $S'_i$.
\end{definition}

\begin{theorem}
  \label{local_encoding}
  Let $\cE:\cB(\bigotimes_{i=1}^n\HS_i) \to \cB(\bigotimes_{i=1}^{n'}\HS'_i)$ be a local encoding with respect to $\{S'_i\}$.
  Denote $Q_0 = \bigcup_{i,j} S'_i\cap S'_j$ and $Q_i = S'_i \setminus Q_0$ (see \cref{fig:local-encodings}).
  Then there exist decompositions $\HS_{Q_0} \simeq {E_0} \ox (\bigotimes_i\HS_i^{(\operatorname{in})})$ and $\HS_{Q_i} \simeq E_i\ox\HS_i^{(\operatorname{out})}$, together with identifications $\HS_i \simeq \HS_i^{(\operatorname{in})}\ox\HS_i^{(\operatorname{out})}$ and a decomposition $E_0 = \bigoplus_{\alpha}\left(\bigotimes_{i=0}^n E_{0.i}^{(\alpha)}\right)$, such that the encoding takes the form
  \begin{multline}\label{eq:local_encoding_characterisation}
    \cE(M) = \\
    U_{Q_0}
     \bigoplus_{\alpha} \left(
       \Bigl[\bigotimes_i U_{(i)}^{(\alpha)}\Bigr]
      \biggl( M^{(\alpha)} \ox \1_{E_1, \dots ,E_n}\ox \1_{E_{0.1}^{(\alpha)},\dots,E_{0.n}^{(\alpha)}} \biggr)
      \Bigl[\bigotimes_i U_{(i)}^{(\alpha)}\Bigr]^\dg
     \right)
    U_{Q_0}^\dg
  \end{multline}
where $U_{Q_0}$ acts non-trivially only on $\cH_{Q_0}$, each $U_{Q_i^+}^{(\alpha)}$ acts on $\HS_{Q_i^+}^{(\alpha)}:=\cH_i \ox E_i \ox E^{(\alpha)}_{0.i}$, and each $M^{(\alpha)} = M$ or $\bar{M}$.
\end{theorem}
\Cref{local_encoding} implies that the locality structure of an encoding is fully determined by how it maps 1-local operators.
Note that any of the Hilbert spaces in the decomposition could be one-dimensional.

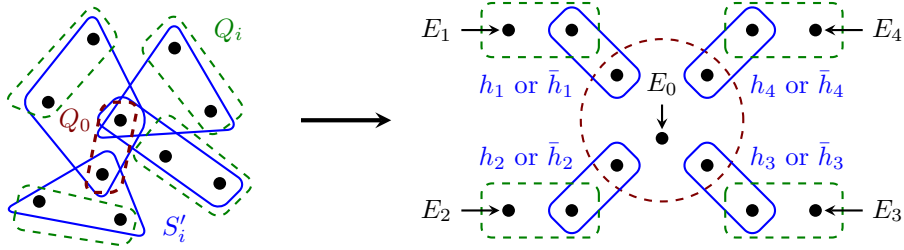
\begin{figure}
  \begin{center}
    \begin{tikzpicture}[scale=1.2]
      \fill (0,0) circle (2pt); \fill (1,0.1) circle (2pt); \fill (0.6,0.8) circle (2pt); \fill (-0.8,0.2) circle (2pt); \fill (-0.3,0.9) circle (2pt); \fill (-0.2,-0.6) circle (2pt); \fill (0.5,-0.4) circle (2pt); \fill (1.1,-0.7) circle (2pt); \fill (-0.9,-0.9) circle (2pt); \fill (0,-1.1) circle (2pt);
      \draw[thick,blue,rounded corners] (-1.3,-1) -- (0.3,-1.3) -- (-0.1,-0.3) -- cycle;
      \draw[thick,blue,rounded corners] (-0.3,1.2) -- (0.3,0) -- (-0.2,-0.9) -- (-1.1,0.2) -- cycle;
      \draw[thick,blue,rounded corners] (0.6,1.1) -- (1.3,-0.1) -- (-0.3,-0.2) -- cycle;
      \draw[thick,blue,rounded corners] (-0.3,0) -- (0,0.3) -- (1.4,-0.7) -- (1.1,-1) -- cycle;
      \draw[thick,green!50!black,rounded corners,dashed] (-1.3,0.2) -- (-0.8,0) -- (0,0.9) -- (-0.3,1.3) -- cycle;
      \draw[thick,green!50!black,rounded corners,dashed] (0.2,0.8) -- (0.6,1.2) -- (1.4,0.1) -- (1,-0.2) -- cycle;
      \draw[thick,green!50!black,rounded corners,dashed] (0.1,-0.4) -- (0.5,0.1) -- (1.5,-0.7) -- (1.1,-1.1) -- cycle;
      \draw[thick,green!50!black,rounded corners,dashed] (-1.2,-1.2) -- (-1.1,-0.7) -- (0.2,-1) -- (0.1,-1.4) -- cycle;
      \draw[very thick,red!50!black,rounded corners,dashed] (-0.2,0.2) -- (0.2,0.2) -- (0,-0.8) -- (-0.4,-0.8) -- cycle;
      \node at (1.2,1) {\color{green!50!black} $Q_i$};
      \node at (-0.5,0) {\color{red!50!black} $Q_0$};
      \node at (0.6,-1.2) {\color{blue} $S'_i$};

      \draw[->,>=stealth,ultra thick] (2,0) -- (3,0);

      \begin{scope}[xshift=6cm]
        \fill (0,-0.2) circle (2pt); \fill (0.5,0.5) circle (2pt); \fill (0.5,-0.5) circle (2pt); \fill (-0.5,0.5) circle (2pt); \fill (-0.5,-0.5) circle (2pt); \fill (1,1) circle (2pt); \fill (1.7,1) circle (2pt);  \fill (1,-1) circle (2pt); \fill (1.7,-1) circle (2pt);  \fill (-1,1) circle (2pt); \fill (-1.7,1) circle (2pt);  \fill (-1,-1) circle (2pt); \fill (-1.7,-1) circle (2pt);
        \draw[thick, red!50!black, dashed] (0,0) circle (0.9);
        \draw[thick, blue, rounded corners] (0.2,0.5) -- (0.5,0.2) -- (1.3,1) -- (1,1.3) -- cycle;
        \draw[thick, blue, rounded corners] (0.2,-0.5) -- (0.5,-0.2) -- (1.3,-1) -- (1,-1.3) -- cycle;
        \draw[thick, blue, rounded corners] (-0.2,0.5) -- (-0.5,0.2) -- (-1.3,1) -- (-1,1.3) -- cycle;
        \draw[thick, blue, rounded corners] (-0.2,-0.5) -- (-0.5,-0.2) -- (-1.3,-1) -- (-1,-1.3) -- cycle;
        \draw[thick,green!50!black,dashed,rounded corners] (0.7,0.7) -- (0.7,1.3) -- (2,1.3) -- (2,0.7) -- cycle;
        \draw[thick,green!50!black,dashed,rounded corners] (0.7,-0.7) -- (0.7,-1.3) -- (2,-1.3) -- (2,-0.7) -- cycle;
        \draw[thick,green!50!black,dashed,rounded corners] (-0.7,0.7) -- (-0.7,1.3) -- (-2,1.3) -- (-2,0.7) -- cycle;
        \draw[thick,green!50!black,dashed,rounded corners] (-0.7,-0.7) -- (-0.7,-1.3) -- (-2,-1.3) -- (-2,-0.7) -- cycle;
        \node (E0) at (0,0.4) {$E_0$};
        \node (E1) at (-2.5,1) {$E_1$};
        \node (E2) at (-2.5,-1) {$E_2$};
        \node (E4) at (2.5,1) {$E_4$};
        \node (E3) at (2.5,-1) {$E_3$};
        \draw [->,>=stealth,thick] (E0) -- (0,-0.1);
        \draw [->,>=stealth,thick] (E1) -- (-1.8,1);
        \draw [->,>=stealth,thick] (E2) -- (-1.8,-1);
        \draw [->,>=stealth,thick] (E3) -- (1.8,-1);
        \draw [->,>=stealth,thick] (E4) -- (1.8,1);
        \node at (-1.5,0.4) {\color{blue} $h_1$ or $\bar{h}_1$};
        \node at (-1.5,-0.4) {\color{blue} $h_2$ or $\bar{h}_2$};
        \node at (1.5,-0.4) {\color{blue} $h_3$ or $\bar{h}_3$};
        \node at (1.5,0.4) {\color{blue} $h_4$ or $\bar{h}_4$};
      \end{scope}
    \end{tikzpicture}
    \caption[Any local encoding can be decoupled into disjoint subsystems by local unitaries.]{%
      Any local encoding can be decoupled into disjoint subsystems by local unitaries on the $Q_i$ systems.
      Each subsystem encodes one of the qudits of the original system.
      Here $S'_i$ denotes the subsystems encoding qudit $i$ as a direct sum of identity and conjugate representations.
      $E_i$ denotes ancilla subsystems.}
    \label{fig:local-encodings}
  \end{center}
\end{figure}

The characterisation in \cref{local_encoding} shows that the most general possible encoding of local Hamiltonians looks very like the complex-to-real encoding from \cref{complex-to-real enc}.
Up to local unitaries, local encodings are just direct sums of product encodings, with a classical ancilla that determines whether to take the complex conjugate of all the local interactions or not.

To prove \cref{local_encoding}, we will need the following (slightly generalised) lemma from \cite{Aharonov-Eldar}, which is itself a special case of a result from \cite{Bravyi-Vyalyi}:
\begin{lemma}[Lemma 3.3 of \cite{Aharonov-Eldar}]
  \label{Bravyi-Vyalyi}
  Let $\HS=\bigotimes_{i=0}^n \HS_i$ be a Hilbert space and let $\cA_{0,k}$, $k \in \{1,\dots , n\}$, be sets of matrices which act non-trivially only on $\HS_0$ and $\HS_k$, such that matrices from different sets all commute.
  Then there exists a direct sum decomposition of $\HS_0$
  \begin{equation}
    \HS_0=\bigoplus_{\alpha} \HS_0^{(\alpha)}
  \end{equation}
  such that inside each subspace $ \HS_0^{(\alpha)}$ there is a tensor product structure
  \begin{equation}
    \HS_0^{(\alpha)}=\bigotimes_{i=0}^n \HS_{0.i}^{(\alpha)},
  \end{equation}
  and any element $A \in \cA_{0,k}$ preserves the subspaces $\HS^{(\alpha)} := \HS_0^{(\alpha)} \otimes \bigotimes_{i=1}^n \HS_i$.
  Moreover $A|_{\HS^{(\alpha)}}$ acts non-trivially only on $\HS_{0.k}^{(\alpha)} \ox \HS_k$.
\end{lemma}
In \cite{Aharonov-Eldar} this lemma is stated only in terms of single operators $H_{0,k}$ rather than sets of operators $\cA_{0,k}$, but the proof from \cite{Aharonov-Eldar} or \cite{Bravyi-Vyalyi} can be easily seen to generalise to this case.

\begin{proof}[{of \cref{local_encoding}}]
  Let $\cA_i = \langle\cE(A_i\ox\1) : A \in \cB(\HS_i)\rangle$ be the algebra generated by the operators $\{\cE(A_i\ox\1)\}$.
  By assumption, $\cA_i$ acts non-trivially only on $\HS_{Q_0\cup Q_i}$.
  Multiplicativity of encodings (\cref{encoding}\labelcref{encoding:mathematical:product}) yields that, for $i \neq j$ and all $A \in \cB(\HS_i)$, $B \in \cB(\HS_j)$,
  \begin{equation}
    [\cE(A_i \ox \1),\cE(B_j \ox \1)] = \cE([A_i \ox \1, B_j \ox \1]) = 0.
  \end{equation}
  Thus the algebras $\cA_i$ fulfil the hypothesis of \cref{Bravyi-Vyalyi} for the Hilbert spaces $\HS = \bigotimes_{i=1}^n \HS_{Q_i}$.
  Applying \cref{Bravyi-Vyalyi}, we obtain a decomposition
  \begin{gather}
    \HS_{Q_0} = \bigoplus_\alpha \left[\bigotimes_{i=0}^n \HS_{0.i}^{(\alpha)}\right]
  \end{gather}
  such that $\cA_i = \bigoplus_\alpha \cA_i^{(\alpha)}$ where $\cA_i^{(\alpha)}$ acts non-trivially only on the factors $\HS_{Q_i}\ox\HS_{0.i}^{(\alpha)}$.
  Let $U_{Q_0}^\dg:\HS_{Q_0}\to\bigoplus_\alpha \bigotimes_{i=0}^n \HS_{0.i}^{(\alpha)}$ be the unitary change of basis corresponding to this decomposition of $\HS_{Q_0}$.

  Now, from the general characterisation of encodings \cref{eq:characterisation}, we know $\cA_i$ has the form
  \begin{equation}
    \label{eq:aiform}
    \cA_i = \Bigl\langle W \left( (A\ox\1)^{\oplus p} \oplus
                   (\bar{A}\ox\1)^{\oplus q} \right) W^\dg \Bigr\rangle
          = \Bigl\langle W (A^{\oplus Dp} \oplus {\bar{A}}^{\oplus Dq}) W^\dg \Bigr\rangle
  \end{equation}
  for some unitary $W$ and $p,q\in\N$.
  ($D$ here is the dimension of the identity operator which acts on all but the $i$'th qudit of the original system.)
  Thus $\cA_i$ is unitarily equivalent to a direct sum of identity and conjugated identity representations of the full matrix algebra on $\HS_i$.
  Note that this decomposes $\cA_i$ into irreducible representations, as the full matrix algebra in any dimension is irreducible.

  Since $\cA_i$ is simultaneously equivalent to $\bigoplus_\alpha \cA_i^{(\alpha)}$, each $\cA_i^{(\alpha)}$ must itself be unitarily equivalent to a direct sum of copies of identity and conjugated identity representations.
  Thus, for arbitrary $A\in\cB(\HS_i)$,
  \begin{equation}
    \label{eq:Aialphadecomp}
    \cE(A_i\ox\1)
    = U_{Q_0}
      \left(
        \bigoplus_\alpha \left[ U_{Q_i^+}^{(\alpha)} \left(A^{\oplus n_i(\alpha)}\oplus \bar{A}^{\oplus m_i(\alpha)}\right) {U_{Q_i^+}^{(\alpha)}}^\dg \ox \1^{(\alpha)}_{\rest} \right]
      \right)
      U_{Q_0}^{\dg}
  \end{equation}
  for some $n_i(\alpha)$, $m_i(\alpha)\in\N$, where $U_{Q_i^+}^{(\alpha)}$ acts on $\HS_{Q_i^+}^{(\alpha)} := \HS_{Q_i}\ox\HS_{0.i}^{(\alpha)}$.

  We will show that for each $\alpha$, either $n_i(\alpha) = 0$ for all~$i$, or $m_i(\alpha) = 0$ for all~$i$.
  Note that $J=\cE(i\1) = \cE\left((i\1_j)\ox\1_k\ox\1_{\rest}\right) = \cE\left(\1_j\ox(i\1_k)\ox\1_{\rest}\right)$ for any qudits $j,k$ of the original system.
  From \cref{eq:Aialphadecomp},
  \begin{multline}
    \left(\bigoplus_\alpha U_{Q_j^+}^{(\alpha)}\ox U_{Q_k^+}^{(\alpha)}\right)^\dg U_{Q_0}^\dg \;
    \cE((i\1_j)\ox\1_k\ox\1_{\rest}) \;
    U_{Q_0} \left(\bigoplus_\alpha U_{Q_j^+}^{(\alpha)}\ox U_{Q_k^+}^{(\alpha)}\right) \\
    = i\bigoplus_{\alpha}\left[
        \left( \1^{\oplus n_j(\alpha)}\oplus(-\1)^{\oplus m_j(\alpha)} \right)
        \ox \1_{Q_k^+}^{(\alpha)}\ox\1_{\rest}^{(\alpha)}
      \right].
  \end{multline}
  Equating this with $U_{Q_0}^\dg \cE(\1_j\ox(i\1_k)\ox\1_{\rest}) U_{Q_0}$ and matching up factors in the direct sum over $\alpha$, we obtain
  \begin{equation}
      \left( \1^{\oplus n_j(\alpha)}\oplus(-\1)^{\oplus m_j(\alpha)} \right)
      \ox \1_{Q_k^+}^{(\alpha)}
    =\1_{Q_j^+}^{(\alpha)}\ox
      \left( \1^{\oplus n_k(\alpha)}\oplus(-\1)^{\oplus m_k(\alpha)} \right),
  \end{equation}
  which is only possible if either $n_j(\alpha) = n_k(\alpha) = 0$ or $m_j(\alpha) = m_k(\alpha) = 0$.
  Since this holds for any pair $j,k$, either $n_i(\alpha) = 0$ for all $i$, or $m_i(\alpha) = 0$ for all $i$, as claimed.
  We write ``$n(\alpha)=0$'', ``$m(\alpha)=0$'' as shorthand for each of these two cases.
  Then
  \begin{multline}\label{eq:Eiplusdecomp}
    U_{Q_0}^\dg \cE(A_i\ox\1) U_{Q_0}
    =\left(\bigoplus_{\alpha : m(\alpha)=0}
        U_{Q_i^+}^{(\alpha)} \bigl(A \ox \1_{E_i^+}^{(\alpha)}\bigr) {U_{Q_i^+}^{(\alpha)}}^\dg \ox \1^{(\alpha)}_{\rest}
     \right) \\
    \oplus
     \left(\bigoplus_{\alpha : n(\alpha)=0}
        U_{Q_i^+}^{(\alpha)} \bigl(\bar{A} \ox \1_{E_i^+}^{(\alpha)}\bigr) {U_{Q_i^+}^{(\alpha)}}^\dg \ox \1^{(\alpha)}_{\rest}
     \right)
  \end{multline}
  where $\dim(\1_{E_i^+}^{(\alpha)}) = n_i(\alpha) + m_i(\alpha)$, so that $\HS_i\ox\HS_{E_i^+}^{(\alpha)} \simeq  \HS_{Q_i^+}^{(\alpha)}$.

  At this point, since $U_{Q_i^+}^{(\alpha)}$ acts on the whole of $\HS_{Q_i^+}^{(\alpha)}$, how we choose to factor $\HS_{Q_i^+}^{(\alpha)}$ to obtain \cref{eq:local_encoding_characterisation} is arbitrary, as long as we choose the factorisation consistently across all $\alpha$.
  Recalling that $\HS_{Q_i^+}^{(\alpha)} \simeq \HS_{Q_i}\ox\HS_{0.i}^{(\alpha)} \simeq \HS_i\ox\HS_{E_i^+}^{(\alpha)}$, one possible choice is to take
  \begin{align}
    \dim\HS_i^{(\operatorname{out})} &= \gcd\Bigl\{\dim\HS_i,\dim\HS_{Q_i}\Bigr\}, \label{eq:dimHiout}\\
    \dim\HS_i^{(\operatorname{in})} &= \dim\HS_i/\dim\HS_i^{(\operatorname{out})}, \label{eq:dimHiin}\\
    \dim E_i &= \dim\HS_{Q_i}/\dim \HS_i^{(\operatorname{out})},\\
    \dim E_{0.i}^{(\alpha)} &= \dim\HS_{0.i}^{(\alpha)}/\dim H_i^{(\operatorname{in})}.
  \end{align}
  (Note that any of these spaces could turn out to be 1-dimensional.)
  This choice manifestly satisfies $\dim\HS_i = \dim(\HS_i^{(\operatorname{out})}\ox\HS_i^{(\operatorname{in})})$, $\dim\HS_{Q_i} = \dim(\HS_i^{(\operatorname{out})}\ox E_i)$, and $ \dim\HS_{0.i}^{(\alpha)}=\dim (\HS_i^{(\operatorname{in})}\ox E_{0.i}^{(\alpha)})$.

  To see that this choice is possible for all $\alpha$, it remains to show that $\dim\HS_{0.i}^{(\alpha)}$ is divisible by $\dim \HS_i^{(\operatorname{in})}$, so that $\dim E_{0.i}^{(\alpha)}$ is well-defined.
  First, note that $\dim\HS_i^{(\operatorname{in})}$ and $\dim\HS_{Q_i}$ are co-prime by \cref{eq:dimHiout,eq:dimHiin}.
  But $\dim\HS_i^{(\operatorname{in})}$ divides $\dim\HS_{Q_i^+}^{(\alpha)} = \dim\HS_{Q_i}\cdot\dim\HS_{0.i}^{(\alpha)}$, so $\dim\HS_i^{(\operatorname{in})}$ must divide $\dim\HS_{0.i}^{(\alpha)}$.

  Therefore, we can consistently factor
  \begin{align}
    \HS_{Q_i} &\simeq \HS_i^{(\operatorname{out})} \ox E_i,\\
    \HS_{0.i}^{(\alpha)} &\simeq \HS_i^{(\operatorname{in})} \ox E_{0.i}^{(\alpha)},\\
    \HS_{Q_i^+}^{(\alpha)}
      &\simeq \HS_i^{(\operatorname{out})} \ox \HS_i^{(\operatorname{in})} \ox E_i \ox E_{0.i}^{(\alpha)}.
  \end{align}
  Recalling that $H_{Q_i^+}^{(\alpha)} \simeq \HS_i\ox\HS_{E_i^+}^{(\alpha)}$, we can identify $H_{E_i^+}^{(\alpha)} \simeq E_i\ox E_{0.i}^{(\alpha)}$, allowing us to rewrite \cref{eq:Eiplusdecomp} in the form
  \begin{multline}
    \label{eq:singlequditchar}
    U_{Q_0}^\dg \cE(A_i\ox\1) U_{Q_0}
    =\left(\bigoplus_{\alpha : m(\alpha)=0}
        U_{Q_i^+}^{(\alpha)} \bigl(A \ox \1_{E_i} \ox \1_{E_{0.i}^{(\alpha)}}\bigr) {U_{Q_i^+}^{(\alpha)}}^\dg \ox \1^{(\alpha)}_{\rest}
     \right) \\
    \mspace{200mu}
    \oplus
     \left(\bigoplus_{\alpha : n(\alpha)=0}
        U_{Q_i^+}^{(\alpha)} \bigl(\bar{A} \ox \1_{E_i} \ox \1_{E_{0.i}^{(\alpha)}}\bigr) {U_{Q_i^+}^{(\alpha)}}^\dg \ox \1^{(\alpha)}_{\rest}
     \right)
  \end{multline}
  where $A$ and $\bar{A}$ act on $\HS_i^{(\operatorname{in})} \ox \HS_i^{(\operatorname{out})}$.

  Finally, note that \cref{eq:singlequditchar} holds for any single qudit operator $A_i\ox\1$ on any qudit~$i$.
  For an arbitrary operator $M$ on $n$ qudits, \cref{local_encoding} follows by expressing $M$ as a real-linear combination of products of single qudit terms, and using additivity, real-linearity and multiplicativity of encodings from \cref{encoding}\labelcref{encoding:mathematical}.
\end{proof}

An alternative statement of the characterisation in \cref{local_encoding} is given by the following corollary:
\begin{corollary}
  Let $\cE:\cB(\bigotimes_{i=1}^n\HS_i) \to \cB(\bigotimes_{i=1}^{n'}\HS'_i)$ be a local encoding with respect to $\{S'_i\}$.
  Denote $Q_0 = \bigcup_{i,j} S'_i\cap S'_j$ and $Q_i = S'_i \setminus Q_0$ (see \cref{fig:local-encodings}).
  Then there exist decompositions $\HS_{Q_0} \simeq {E_0} \ox (\bigotimes_i\HS_i^{(\operatorname{in})})$ and $\HS_{Q_i} \simeq {E_i}\ox\HS_i^{(\operatorname{out})}$, together with identifications $\HS_i \simeq \HS_i^{(\operatorname{in})}\ox\HS_i^{(\operatorname{out})}$, such that the encoding takes the form:
  \begin{multline}\label{eq:local_encoding_characterisation2}
    \cE(M)
    =U_{Q_0}
      \left(
   \prod_i U_{(i)}
   \right)
      \biggl(
        M \ox \1 \ox P_{E_0}
        + \bar{M}\ox \1 \ox P_{E_0}^\perp
      \biggr)
      \left(
         \prod_i U_{(i)}
     \right)^\dg
      U_{Q_0}^\dg
  \end{multline}
  where each unitary $U_{(i)}$ acts non-trivially only on $\cH_i \ox E_0 \ox E_i$, and the following commutators vanish for all $i,j$:
  \begin{equation}
    [U_{(i)},U_{(j)}] = 0 \quad \text{and} \quad [P_{E_0},U_{(i)}] = 0.
  \end{equation}
\end{corollary}

\begin{proof}
  This is immediate from \cref{local_encoding} and the following definitions of $P_{E_0}$ and $U_{(i)}$:
  \begin{equation}
    P_{E_0} = \bigoplus_{\alpha : m(\alpha)=0} \1_{E_0}^{(\alpha)},
    \quad \text{and} \quad
    U_{(i)} = \bigoplus_{\alpha}\left[ U_{Q_i^+}^{(\alpha)}\ox \1^{(\alpha)}_{\rest}\right],
  \end{equation}
  where $m(\alpha)$ is defined as in the proof of \cref{local_encoding}.
\end{proof}

\Cref{local_encoding} characterises what encodings must look like if they are to map local Hamiltonians to local Hamiltonians, and more generally local observables on the original system to local observables on the simulator.
We have seen that, because encodings preserve commutators, observables on different qudits of the original system are necessarily mapped to commuting observables on the simulator system, so remain simultaneously measurable.

However, if the subsets $S'_i$ overlap, these observables on the simulator will in general no longer be on disjoint subsets of qudits; tensor products of operators on the original system are not necessarily mapped to tensor products on the simulator.
If we impose the additional requirement that tensor products are mapped to tensor products, which is equivalent to requiring that all the subsets $S'_i$ are disjoint, then there is no $Q_0$ subsytem and the characterisation from \cref{local_encoding} simplifies substantially:

\begin{corollary}[Product-preserving encodings]\hfill\\
  \label{product-encodings}
  Let $\cE:\cB(\bigotimes_{i=1}^n\HS_i) \to \cB(\bigotimes_{i=1}^{n'}\HS'_i)$ be a local encoding with respect to $\{S'_i\}$, where $S'_i$ are disjoint subsets.
  Then the encoding must take one of the following forms, where $S'_i = \{i\}\cup E_i$:
  \begin{align}
    \cE(M) &= \Bigl(\bigotimes_i U_{i,E_i}\Bigr)
              \Bigl(M_{1,\dots,n} \ox \1_{E_1,E_2,...E_n} \Bigr)
              \Bigl(\bigotimes_i U_{i,E_i}^\dg\Bigr)\\
    \intertext{or}
    \cE(M) &= \Bigl(\bigotimes_i U_{i,E_i}\Bigr)
              \Bigl(\bar{M}_{1,\dots,n} \ox \1_{E_1,E_2,...E_n} \Bigr)
              \Bigl(\bigotimes_i U_{i,E_i}^\dg\Bigr).
  \end{align}
\end{corollary}

Thus for tensor products to be mapped to tensor products under encoding, the encoding must be rather trivial.
Up to local unitaries, it either consists solely of copies of $H$, \emph{or} solely of copies of $\bar{M}$; it cannot contain both $M$ and $\bar{M}$.
This rules out for example the complex-to-real encoding of \cref{complex-to-real enc}.

\Cref{product-encodings} applies to product-preserving encodings that map to the entire Hilbert space of the simulator system.
We will see shortly that things are more interesting if the local encoding maps into a subspace of the simulator's Hilbert space; non-trivial tensor-product-preserving encodings into a subspace \emph{are} possible.


\subsection{Encodings in a subspace}
\label{sec:subspaceenc}

It may be the case that an encoding $\cE(H)$ acts only within a subspace $S$ of the simulator system $\cH'$.
That is, we say a map $\cE :\cB(\cH)\to \cB(\cH')$ is an \emph{encoding into the subspace $S$} if $\cE(H)$ has support only on $S$ and the map $H\mapsto \cE(H)|_S$ is an encoding.
Later we may refer to a map of this form simply as a \emph{subspace encoding} or even just an encoding when the subspace is implicit.
We call the subspace $S_{\cE}$ onto which $\cE$ maps the \emph{encoded subspace}.

All the conclusions of the above section still hold, but now the target space $S_{\cE}$ is embedded in a larger space $\cH'$, so the unitary $U$ is replaced with an isometry $V$.
Any subspace encoding may therefore be written in the form
\begin{equation}
  \mathcal{E}(M) = V \left(M^{\oplus p} \oplus \bar{M}^{\oplus q}\right) V^\dg
  = V\left(M\ox P + \bar{M} \ox Q\right)V^\dg.
\end{equation}
We remark that $P$ and $Q$ may be chosen to be any orthogonal projectors on the ancilla system $E$ with rank$(P)=p$ and rank$(Q)=q$, provided that the isometry $V$ is changed accordingly.
Indeed, even the dimension of the ancilla system $E$ may be increased such that $P$ and $Q$ do not sum to the identity, as long as the map $V|_{\text{supp}(P+Q)}$ is an isometry onto the subspace $S_{\cE}$.
This will be useful in the simple characterisation of local subspace encodings given in the next section.
Note that $\cE(\1)$ is the projector onto the subspace $S_{\cE}$.


\subsection{Local encodings in a subspace}

We can now consider encodings into a subspace that are \emph{local}.
Since all the encodings we construct later will not only be local, but in fact will also satisfy the stronger condition of mapping tensor products of operators to tensor products on the simulator, we will restrict our attention here to tensor-product-preserving encodings into a subspace.
We therefore want to be able to decompose the simulator system $\cH'$ into $n$ subsystems $\cH'= \bigotimes_{i=1}^n \cH'_i$ such that $\cH'_i$ corresponds to $\cH_i$ operationally.
The encoding of a local observable should then be equivalent to a local observable, in terms of its action on the subspace $S_{\cE}$ into which the encoding maps:

\begin{definition}\label{def:localsubenc}
  Let $\cE:\cB\left(\bigotimes_{i=1}^n \cH_j\right) \to \cB\left(\bigotimes_{j=1}^n \cH'_j\right)$ be a subspace encoding.
  We say that the encoding is \emph{local} if for any $A_j\in \Herm(\cH_j)$, there exists $A_j'\in \Herm(\cH'_j)$ such that
  \begin{equation}
\label{eq:localsubencdef}
    \cE(A_j\ox \1)=(A_j'\ox \1)\cE(\1).
  \end{equation}
\end{definition}

Note that for a simulation of $n$ particles with $m$ particles, this does not mean we require $m=n$, but rather that the $m$ particles can be partitioned into $n$ groups, each of which is labelled by $\cH_j'$. First we show that local observables on the original system correspond to local observables on the simulator system:

\begin{proposition}\label{localmeas}
  Let $\cE$ be a local encoding into the subspace $S_{\cE}$.
  Let $\rho'$ be a state in the encoded subspace such that $\cE(\1)\rho'=\rho'$.
  Let $A_j$ be an observable on qudit $j$ of the original system.
  Then there exists an observable $A'_j$ on $\cH'_j$ such that
  \begin{equation}
    \tr[(A_j\ox \1)\rho]=\tr[(A'_j \ox \1)\rho']
  \end{equation}
  where $\rho=F(\rho')+B(\rho')$, for $F$ and $B$ defined as

  \begin{equation}
    F(\rho')=\tr_E[V^{\dag}\rho'V(\1 \ox P)] \text{ and } B(\rho')=\tr_E[V^{\dag}\rho'V(\1 \ox Q)]
  \end{equation}
\end{proposition}

\begin{proof}
  This is an immediate consequence of \cref{def:localsubenc} and \cref{eqn:Fmeas}.
\end{proof}

It turns out that \cref{def:localsubenc} is equivalent to saying that $\cE$ is a tensor product of encodings acting on the the encoded space $S_{\cE}$:

\begin{lemma}\label{lem:tensorenc}
  An encoding $\cE$ is local if and only if it can be written as a ``tensor product'' of encodings $\varphi_j:\Herm(\cH_j)\to \Herm(\cH'_j)$ in the following way:
  \begin{equation}\label{eqn:tensorenc}
    \cE\left(\bigotimes_{j=1}^n A_j\right)=\left[\bigotimes_{j=1}^n \varphi_j(A_j)\right]\cE(\1)
  \end{equation}
\end{lemma}

\begin{proof}
  If there exist encodings $\varphi_j$ such that \cref{eqn:tensorenc} holds, then $\cE$ is local as for any $A_j \in \cB(\cH_j)$ one can take $A_j'=\varphi_j(A_j) \in \cB(\cH'_j)$, and
  \begin{align}
    \cE(A_j \ox \1)
    &=\left[\varphi_j(A_j) \ox \left(\bigotimes_{k\neq j} \varphi_k(\1)\right)\right]\cE(\1)\\
    &=\left[\varphi_j(A_j) \varphi_j(\1) \ox \left(\bigotimes_{k\neq j} \varphi_k(\1)\right)\right]\cE(\1)\\
    &=\left[(\varphi_j(A_j) \ox \1) \left(\bigotimes_{k=1}^n \varphi_k(\1)\right)\right]\cE(\1)\\
    &=(A_j' \ox \1)\cE(\1).
  \end{align}
  For the converse, we will first show that the map $A_j \mapsto A_j'$ can be taken to be a subspace encoding.
  Since $A_j' \in \Herm(\cH'_j) $ is Hermitian, we have
  \begin{equation}
    (A_j'\ox \1)\cE(\1)=\cE(A_j\ox \1)=\cE(A_j\ox \1)^\dag=\cE(\1)(A_j'\ox \1),
  \end{equation}
  so $A_j'\ox \1$ commutes with $\cE(\1)$.

  For a given $j$, consider the subspace $T_j$ of $\cH'_j$ which is entirely annihilated by $\cE(\1)$, defined by $T_j=\{\ket{\psi}\in \cH_j\: : \:(\proj{\psi} \ox \1)\cE(\1)=0\}$.
  We will choose to take $\varphi_j(A_j)=\Pi_j A_j' \Pi_j$ where $\Pi_j$ is the projector onto $T_j^{\perp}$.
  We will show that $\varphi_j$ is a subspace encoding, by showing the requirements of \cref{encoding}\ref{encoding:operational_subset} hold in the subspace $T_j^{\perp}$: Hermiticity preservation, spectrum preservation and real-linearity.
  First note that $\varphi_j(A_j)$ is Hermitian and has support only on $T_j^{\perp}$.

  The projector $(\1-\Pi_j)\ox \1$ annihilates $\cE(\1)$ by definition of $T_j$, so $(\Pi_j\ox\1)\cE(\1)=\cE(\1)$.
  Therefore
  \begin{equation}
    [\varphi_j(A_j)\ox \1]\cE(\1)=[\Pi_j A_j' \Pi_j\ox \1]\cE(\1)=\cE(A_j \ox \1),
  \end{equation}
  where we have used the fact that $\cE(\1)$ commutes with $A_j' \ox \1$.
  Thus $\varphi_j(A_j)$ can be used as a replacement for $A_j'$ in \cref{eq:localsubencdef} which has support only on $T_j^\perp$.

  We know that $\varphi_j(A_j)\ox\1$ commutes with $\cE(\1)$ and is therefore block diagonal with respect to the $\cE(\1), \1-\cE(\1)$ split.
  Furthermore since $\varphi_j(A_j)$ has no support on $T_j$, no eigenvalues of $\varphi_j(A_j) \ox \1$ are completely annihilated when multiplied by $\cE(\1)$.
  Therefore
  \begin{equation}
    \spec(\varphi_j(A_j)|_{T_j^{\perp}})=\spec(\cE(A_j \ox\1)|_{S_{\cE}})=\spec(A_j).
  \end{equation}

  Next we show that $\varphi_j$ is real-linear, using the real-linearity of $\cE$.
  For any $\lambda,\mu \in \R$, and $A_j,B_j \in \Herm(\cH)$,
 \begin{align}
[\varphi_j(\lambda A_j+\mu B_j) \ox \1]\cE(\1)  &=\cE((\lambda A_j+\mu B_j)\ox\1)\\
&=\lambda\cE( A_j\ox \1)+\mu\cE( B_j\ox \1)\\
 &= [(\lambda\varphi_j(A_j)+\mu\varphi_j( B_j))\ox \1 ]\cE(\1)\\
 \Leftrightarrow [(\lambda\varphi_j( A_j)+\mu\varphi_j( B_j) &-\varphi_j(\lambda A_j+\mu B_j))\ox \1 ]\cE(\1) =0.
\end{align}
For real-linearity of $\varphi_j$ we need to show that $M=\lambda\varphi_j( A_j)+\mu\varphi_j( B_j) -\varphi_j(\lambda A_j+\mu B_j)$ vanishes. This follows because $M\ox \1$ commutes with and is annihilated by $\cE(\1)$, but $M$ has no support on $T_j$.
  Therefore $\varphi_j$ is an encoding into the subspace $T_j^{\bot}$.

  It remains to show that $\cE$ can be written in the form of \cref{eqn:tensorenc}.
  This follows from the fact that $\cE$ and $\varphi_j$ are Jordan homomorphisms, and $(A_j\ox \1)(\1\ox B_k)=(\1\ox B_k)(A_j\ox \1)$.
  So for example for a bipartite system with two subsystems labelled $a$ and $b$:
  \begin{align}
    \cE(A_a\ox B_b)
    &=\cE(A_a\ox \1)\cE(\1\ox B_b)\\
    &=\left[\varphi_a(A_a)\ox \1\right]\cE(\1)\left[\1\ox \varphi_b(B_b)\right]\cE(\1)\\
    &=\left[\varphi_a(A_a)\ox \varphi_b(B_b)\right]\cE(\1).
  \end{align}
\end{proof}

We remark that if $\cE$ and $\varphi_j$ are extended to homomorphisms on all matrices as described in \cref{local_encoding}, then \cref{eqn:tensorenc} holds for all matrices, not just Hermitian ones. This is because the enveloping algebra for the Hermitian matrices includes all matrices, so any matrix can be written as a product of Hermitian matrices.

This extension to all matrices may seem problematic: for example, when calculating $\cE(i\1)$ one could put the factor of $i$ on any one of the subsystems $\cH_j$ before appplying \cref{eqn:tensorenc}.
This just implies that the encodings $\varphi_j$ must satisfy some extra constraints, in order for the overall map to be an encoding.

In fact, we are able to use this condition to derive the following general form of a local encoding (see \cref{fig:localenc}):

\begin{figure}
\begin{center}
\begin{tikzpicture}[scale=0.8]
\tikzstyle{targ}=[fill=yellow!50]
\tikzstyle{sim}=[fill=green!50]
\tikzstyle{anc}=[fill=blue!50]
\def\r{0.5}
\foreach \y in {1,2}
{\def\h{-2.2*\y*\r}
\draw (-0.1,\h) circle (\r cm);
\node  at (-0.1,\h) {$\mathcal{H}_{\y}$};
\draw (3*\r+1.75,\h) ellipse (2*\r cm and \r cm);
\node  at (3*\r+1.75,\h) {$\mathcal{H}'_{\y}$};
\draw (1.5*\r,\h) circle (0.6*\r cm);
\node at (1.5*\r,\h) {\footnotesize $E_{\y}$};
\draw [->] (\r +0.75,\h) -- (\r +1.5,\h);
\node at (\r +1,\h+0.3) {$V_{\y}$};}
\node at (0,-6*\r) {$\vdots$};
\node at (3*\r+1.75,-6*\r) {$\vdots$};
\def\h{-8*\r};
\draw (-0.1,\h) circle (\r cm);
\node  at (-0.1,\h) {$\mathcal{H}_{n}$};
\draw (3*\r+1.75,\h) ellipse (2*\r cm and \r cm);
\node  at (3*\r+1.75,\h) {$\mathcal{H}'_{n}$};
\draw (1.5*\r,\h) circle (0.6*\r cm);
\node at (1.5*\r,\h) {\footnotesize $E_{n}$};
\draw [->] (\r +0.75,\h) -- (\r +1.5,\h);
\node at (\r +1,\h+0.3) {$V_{n}$};
\end{tikzpicture}
\end{center}
\caption[Any local encoding within a subspace can be represented as a tensor product of isometries.]{%
  Any local encoding within a subspace can be represented as a tensor product of isometries, as illustrated here.}
\label{fig:localenc}
\end{figure}
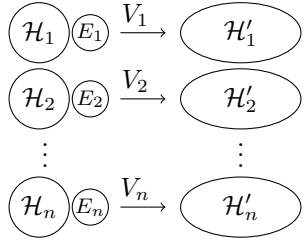

\begin{theorem}
\label{thm:localenc}
A map $\cE:\cB(\bigotimes_{j=1}^n \cH_j) \to \cB(\bigotimes_{j=1}^n \cH'_j)$ is a local encoding if and only if there exist $n$ ancilla systems $E_j$ such that $\cE$ is of the form \begin{equation}
  \cE(M)=V(M\ox P+\bar{M} \ox Q)V^\dag
\end{equation}
 where \begin{itemize}
 \item $V$ is a local isometry: $V=\bigotimes_j V_{j}$ for isometries $V_j:\cH_j\ox E_j \to \cH'_j$.
 \item $P$ and $Q$ are orthogonal projectors on $E=\bigotimes_j E_j$, and are locally distinguishable: for all $j$, there exist orthogonal projectors $P_{E_j}$ and $Q_{E_j}$ acting on $E_j$ such that $(P_{E_j}\ox \1)P=P$ and $(Q_{E_j}\ox \1)Q=Q$.
\end{itemize}
\end{theorem}

\begin{proof}
  If $\cE$ is of the form given above then by \cref{encoding} it is an encoding into the subspace $\cE(\1)=V(\1\ox(P+Q))V^\dg$.
  It is easy to check that $\cE$ is local: for $A_j \in \Herm(\cH_j)$, just take $A_j'=V_j(A_j\ox P_{E_j}+\bar{A_j}\ox Q_{E_j})V_j^\dg \in \Herm(\cH'_j)$ and use the conditions of the theorem.

  For the converse, note that since $\cE$ is an encoding, it must be of the form $\cE(M)=W(M\ox \widetilde{P}+\bar{M} \ox \widetilde{Q})W^\dg$, where $\widetilde{P}$ and $\widetilde{Q}$ are projectors on an ancilla system $\widetilde{E}$ and $W:\cH \ox \widetilde{E} \to \cH'$ is an isometry.
  By \cref{lem:tensorenc}, there exist $n$ encodings $\varphi_j$ such that $\cE(A_j\ox \1)=(\varphi_j(A_j)\ox \1)\cE(\1)$ for any $A_j \in \Herm(\cH'_j)$.
  Since $\varphi_j$ is an encoding, it must be of the form $\varphi_j(A_j)= V_{j}(A_j\ox P_{E_j}+\bar{A_j}\ox Q_{E_j})V_{j}^\dg$ where $P_{E_j}$ and $Q_{E_j}$ are projectors on an ancilla system $E_j$ and $V_j:\cH_j \ox E_j \to \cH'_j$ is an isometry.

Let $E=\bigotimes_j E_j$ and define an isometry $V=\bigotimes_j V_j: \cH \ox E \to \cH'$. Then by \cref{lem:tensorenc}, for any $j$ and $A_j \in \cB(\cH_j)$:
  \begin{align}\label{eqn:localV}
    \cE(A_j \ox \1)
    &= W(A_j\ox \1 \ox \widetilde{P}+\bar{A_j} \ox \1\ox \widetilde{Q})W^\dg\\
    &= V(A_j\ox P_{E_j}\ox \1+\bar{A_j} \ox Q_{E_j}\ox \1)V^\dg W(\1\ox(\widetilde{P}+\widetilde{Q}))W^\dg
  \end{align}
  Substituting in $A_j=i\1$ in the above expression and matching up the $+i$ and $-i$ eigenspaces implies that
  \begin{align}
    \label{eqn:localdistinguish}
    V(P_{E_j}\ox \1)V^\dg W(\1 \ox \widetilde{P})W^\dg=W(\1 \ox \widetilde{P})W^\dg \\
    V(Q_{E_j}\ox \1)V^\dg W(\1 \ox \widetilde{Q})W^\dg=W(\1 \ox \widetilde{Q})W^\dg.
  \end{align}
  We can therefore multiply \cref{eqn:localV} by $W(\1\ox \widetilde{P})W^\dg$ to obtain:
  \begin{equation}
    W(A_j\ox \1\ox \widetilde{P})W^\dg=V(A_j\ox \1)V^\dg W(\1\ox \widetilde{P})W^\dg
 \end{equation}
implying
\begin{equation}
    V^\dg W(\1\ox \widetilde{P})(A_j\ox \1)=(A_j\ox \1)V^\dg W(\1\ox \widetilde{P})
\label{eq:VWPcommutation}
  \end{equation}
  Let $\sum_l B_l \ox C_l$ be the operator Schmidt decomposition of $V^{\dag}W(\1 \ox \widetilde{P})$, where $B_l \in \cB(\cH_j)$ and $C_l: \left(\bigotimes_{k\neq j} \cH_k \right)\ox \widetilde{E} \to \left(\bigotimes_{k\neq j} \cH_k \right) \ox E$.  Then from \cref{eq:VWPcommutation} we have
\begin{equation}
\sum_l[B_l,A_j] \ox C_l=0
\end{equation}
which implies $[B_l,A_j]=0$ for all $l$ by linear independence of the $C_l$. This holds for all matrices on $A_j \in \cB(\cH_j)$. So by Schur's lemma, each $B_l$, and hence also $V^\dg W(\1\ox \widetilde{P})$, must act trivially (i.e.\ as a multiple of the identity) on $\cH_j$ for all $j$, and hence on $\cH$.

  By the same argument $V^\dg W(\1\ox \widetilde{Q})$ acts trivially on all of $\cH$ and so we can conclude there must exist an isometry $U:  \widetilde{E} \to E$ such that
  \begin{equation}
    V^\dg W(\1\ox \widetilde{P})=(\1\ox U \widetilde{P}) \text{ and }     V^\dg W(\1\ox \widetilde{Q})=(\1\ox U \widetilde{Q})
  \end{equation}
  Define $P=U\widetilde{P}U^\dg$ and $Q=U\widetilde{Q}U^\dg$, and remember that $\cE(M)$ must be in the range of the isometry $V$ by \cref{lem:tensorenc}, so we have
  \begin{align}
    \cE(M)&=VV^{\dag}\cE(M)VV^{\dag} =VV^{\dag}W(M\ox \widetilde{P}+\bar{M} \ox \widetilde{Q})W^\dg VV^{\dag}\\
&= V(M\ox P+\bar{M} \ox Q)V^\dg
  \end{align}
  and note that \cref{eqn:localdistinguish} implies that $(P_{E_j}\ox \1)P=P$ and $(Q_{E_j}\ox \1)Q=Q$ as required.
\end{proof}

When $\cE$ is a local encoding from $n$ qudits to $m$ qudits of the same local dimension $d$, the space $\cH_j'$ is a group of $k_j$ qudits. As described at the end of \cref{sec:subspaceenc}, the dimension of the ancilla $E_j$ can be increased until it is of size $d^{k_j-1}$ so that the dimensions of $\cH_j \ox E_j$ and $\cH'$ match. If this is done for all $j$, then  all the $V_j$ (and hence also $V=\bigotimes V_j$) are unitaries.


\subsection{Composition and approximation of encodings}

In this section, we collect some straightforward technical lemmas about encodings which we will need later: that encodings compose properly, and that approximations to encodings behave as one would expect.

\begin{lemma}\label{lem:chaining}
  If $\cE_1$ and $\cE_2$ are encodings, then their composition $\cE_1\circ\cE_2$ is also an encoding.
  Furthermore, if $\cE_1$ and $\cE_2$ are both local, then their composition $\cE_1\circ\cE_2$ is local.
\end{lemma}

\begin{proof}
  By the definition of encodings, we can write
  \begin{align}
  \cE_1(M)=V(M\ox P^{(1)}+\bar{M}\ox Q^{(1)})V^\dg \\
  \cE_2(M)=W(M\ox P^{(2)}+\bar{M}\ox Q^{(2)})W^\dg
  \end{align}
  for isometries $V$ and $W$, and orthogonal pairs of projectors $P^{(1)}, Q^{(1)}$ and $P^{(2)}, Q^{(2)}$ . Then
  \begin{align}
    (\cE_1\circ\cE_2)(M)
    &= V\left[W(M\ox P^{(2)}+\bar{M}\ox Q^{(2)})W^\dg\ox P^{(1)}\right.\\
    &\qquad \left.+\overline{W(M\ox P^{(2)}+\bar{M}\ox Q^{(2)})W^\dg}\ox Q^{(1)}\right]V^\dg\\
    &= U\left[M\ox\left(P^{(2)}\ox P^{(1)}+\bar{Q}^{(2)}\ox Q^{(1)}\right) \right.\\
    & \quad \left.+\bar{M}\ox \left(Q^{(2)}\ox P^{(1)} +\bar{P}^{(2)}\ox Q^{(1)}\right)\right]U^\dg
  \end{align}
  where $U=V\left(W\ox P^{(1)} +\bar{W}\ox Q^{(1)}+\1\ox (\1-P^{(1)}-Q^{(1)})\right)V^\dg$ is an isometry.
  Then observing that $P=P^{(2)}\ox P^{(1)}+\bar{Q}^{(2)}\ox Q^{(1)}$ and $Q=Q^{(2)}\ox P^{(1)} +\bar{P}^{(2)}\ox Q^{(1)}$ are orthogonal projectors, we conclude that $\cE_1\circ\cE_2$ is an encoding.

  If $\cE_1$ and $\cE_2$ are both local then the projectors are locally distinguishable, which means there exist projectors $P^{(a)}_{E_i^{(a)}}$ and $Q^{(a)}_{E_i^{(a)}}$ for $a \in \{1,2\}$ such that
  \begin{equation}
    \left(P^{(a)}_{E_i^{(a)}}\ox \1\right)P^{(a)}=P^{(a)} \quad \text{ and } \quad \left(Q^{(a)}_{E_i^{(a)}}\ox \1\right)Q^{(a)}=Q^{(a)}.
  \end{equation}
  We can show that $P$ and $Q$ are locally distinguishable by defining orthogonal projectors on the systems $E_i=E_i^{(2)}\ox E_i^{(1)}$ as follows:
  \begin{equation}
    P_{E_i}=P^{(2)}_{E_i^{(2)}}\ox P^{(1)}_{E_i^{(1)}} +\bar{Q}^{(2)}_{E_i^{(1)}}\ox Q^{(1)}_{E_i^{(1)}}
    \text{ and }
    Q_{E_i}=Q^{(2)}_{E_i^{(2)}}\ox P^{(1)}_{E_i^{(1)}} +\bar{P}^{(2)}_{E_i^{(1)}}\ox Q^{(1)}_{E_i^{(1)}}
  \end{equation}
  such that $(P_{E_i}\ox \1)P=P$ and $(Q_{E_i}\ox \1)Q=Q$.

  Furthermore, since $\cE_1$ and $\cE_2$ are local, the isometries $V$ and $W$ are tensor products $V=\bigotimes_iV_i$ and $W=\bigotimes_i W_i$, and we can define a local isometry
  \begin{equation}
    U'=\bigotimes_iV_i\left(W_i\ox P^{(1)}_{E_i^{(1)}}+\bar{W}_i\ox Q^{(1)}_{E_i^{(1)}}+\1\ox (\1-P^{(1)}_{E_i^{(1)}}-Q^{(1)}_{E_i^{(1)}})\right)V_i^\dg
  \end{equation}
  such that $(\cE_1\circ\cE_2)(M)=U'(M\ox P+\bar{M}\ox Q)U'^\dg$.
\end{proof}

Next we show that, unsurprisingly, if two encodings are close, the results of applying the encodings to the same operator are also close; and similarly that if two operators are close, the results of applying the same encoding to the operators are close.
We first prove a small technical lemma, which will be useful both here and throughout the paper.

\begin{lemma}\label{lem:conjdiff}
  Let $A,B:\cH\to \cH'$ and $C:\cH \to \cH$ be linear maps.
  Let $\| \cdot \|_a$ be the trace norm or operator norm.
  Then
  \begin{equation}
    \|ACA^{\dg}-BCB^{\dg}\|_a \le (\|A\|+\|B\|)\|A-B\|\|C\|_a.
  \end{equation}
\end{lemma}

\begin{proof}
  The proof is a simple application of the triangle inequality followed by submultiplicativity:
  \begin{align}
    \|ACA^{\dg}-BCB^{\dg}\|_a
    &\le \|ACA^{\dg}-BCA^{\dg}\|_a +\|BCA^{\dg}-BCB^{\dg}\|_a\\
    &\le \|A-B\|\|C\|_a \|A^{\dg}\|+\|B\|\|C\|_a \|A^{\dg}-B^{\dg}\| \\
    &=(\|A\|+\|B\|)\|A-B\|\|C\|_a
  \end{align}
  where we have also used $\|A\|=\|A^{\dg}\|$.
\end{proof}

\begin{lemma}\label{lem:closeencs}
  Consider two encodings $\mathcal{E}$ and $\widetilde{\mathcal{E}}$ defined by $\mathcal{E}(M) = V(M^{\oplus p} \oplus \bar{M}^{\oplus q})V^\dg$, $\widetilde{\mathcal{E}}(M) = \widetilde{V}(M^{\oplus p} \oplus \bar{M}^{\oplus q})\widetilde{V}^\dg$, for some isometries $V$, $\widetilde{V}$.
  Then, for any operators $M$ and $\widetilde{M}$:
  \begin{enumerate}
  \item $\|\mathcal{E}(M) - \widetilde{\mathcal{E}}(M)\| \le 2\|V - \widetilde{V}\| \|M\|$;
  \item $\|\mathcal{E}_{\operatorname{state}}(M) - \widetilde{\mathcal{E}}_{\operatorname{state}}(M)\|_1 \le 2\|V - \widetilde{V}\| \|M\|_1$;
  \item $\|\mathcal{E}(M) - \mathcal{E}(\widetilde{M})\| = \|M - \widetilde{M}\|$.
  \end{enumerate}
\end{lemma}

\begin{proof}
  Write $M' = M^{\oplus p} \oplus \bar{M}^{\oplus q}$.
  Then, for the first part,
  \begin{equation}
    \|\mathcal{E}(M) - \widetilde{\mathcal{E}}(M)\| = \| V M' V^\dg - \widetilde{V} M' \widetilde{V}^\dg \| \le 2\|V - \widetilde{V}\| \|M\|
  \end{equation}
  by \cref{lem:conjdiff}, using $\|M'\| = \|M\|$.
  For the second part, recall that $\mathcal{E}_{\operatorname{state}}(\rho)$ is either defined as $V(\rho \ox \sigma)V^\dg$ or $V(\bar{\rho} \ox \sigma)V^\dg$, dependent on whether $p \ge 1$, for some fixed state $\sigma$.
  Then, writing $M' = M \ox \sigma$ or $M' = \bar{M} \ox \sigma$ and observing that $\|M'\|_1 = \|M\|_1$, the argument is the same as the first part (replacing the operator norm with the trace norm appropriately).

  The third part is essentially immediate:
  \begin{equation}
    \|\mathcal{E}(M) - \mathcal{E}(\widetilde{M})\| = \|V((M-\widetilde{M})^{\oplus p} \oplus (\bar{M}- \bar{\widetilde{M}})^{\oplus q})V^\dg\| = \|M - \widetilde{M}\|.
  \end{equation}
\end{proof}


\section{Hamiltonian simulation}
\label{sec:hamsim}

\subsection{Perfect simulation}

We have seen that encodings capture the notion of one Hamiltonian exactly reproducing all the physics of another.
We will be interested in a less restrictive notion, where this holds only for the low-energy part of the first Hamiltonian.
This concept can be captured by generalising the idea of encodings to \emph{simulations}.
Let $H \in \mathcal{B}((\C^d)^{\ox n})$ and $H' \in \mathcal{B}((\C^{d'})^{\ox m})$ for some $m \ge n$.
We usually think of the local dimensions $d$, $d'$ as fixed, but the number of qudits $n$, $m$ as growing.
Recall that $S_{\leq \Delta (H')}=\linspan \{ \ket{\psi} : H\ket{\psi}=\lambda\ket{\psi}, \lambda \le \Delta \}$ denotes the low energy space of $H'$ and  $P_{\leq \Delta(H')}$ denotes the projector onto this space.
\begin{definition}
  \label{dfn:perfectsim}
  We say that $H'$ perfectly simulates $H$ below energy $\Delta$ if there is a local encoding $\mathcal{E}$ into the subspace $S_{\cE}$ such that:
  \begin{enumerate}
  \item \label[condition]{dfn:perfectsim:subspace}
    $S_{\cE}=S_{\leq \Delta (H')}$ (or equivalently $\cE(\1)=P_{\leq \Delta(H')}$);
  \item \label[condition]{dfn:perfectsim:Hamiltonian}
    $H'|_{\le \Delta} = \mathcal{E}(H)|_{S_{\cE}}$ .
  \end{enumerate}
\end{definition}
Note that \cref{dfn:perfectsim:subspace} is crucial in order for it to make sense to compare $H'|_{\leq \Delta}$ and $\cE (H)|_{S_{\cE}}$.
When \cref{dfn:perfectsim:subspace} holds, \cref{dfn:perfectsim:Hamiltonian} is equivalent to $H'_{\le \Delta} = \mathcal{E}(H)$, where $H'_{\le \Delta}=H'P_{\leq \Delta(H')}$ is the low energy part of $H'$.

To gain some intuition for the above definition, taking $\cE$ to be the identity map, we see that $H$ perfectly simulates itself.
Further, for any $U \in U(d)$, we see that $U^{\ox n} H (U^\dg)^{\ox n}$ is a perfect simulation of $H$.
This freedom to apply local unitaries allows us, for example, to relabel Pauli matrices in the Pauli expansion of $H$.
It also allows us to bring 2\nbd-qubit interactions into a canonical form~\cite{Cubitt-Montanaro}.
Imagine we have a Hamiltonian on $n$ qubits which can be written as a sum of 2\nbd-local terms, each proportional to some 2\nbd-qubit interaction $H$ which is symmetric under interchange of the qubits.
Then it is not hard to show~\cite{Cubitt-Montanaro} that there exists $U \in SU(2)$ such that
\begin{equation}
   U^{\ox 2} H (U^\dg)^{\ox 2} = \sum_{s \in \{x,y,z\}} \alpha_s \sigma_s \ox \sigma_s + \sum_{t \in \{x,y,z\}} \beta_t(\sigma_t \ox \1 + \1 \ox \sigma_t)
\end{equation}
for some weights $\alpha_s,\beta_t \in \R$.
Applying $U^{\ox n}$ to the whole Hamiltonian simulates the $H$ interactions with interactions of this potentially simpler form.

Both of these examples of perfect simulations are actually also encodings. As an example of a perfect simulation which is not an encoding, we observe that qubit Hamiltonians can simulate qudit Hamiltonians.

\begin{lemma}
\label{prop:qudits}
  Let $H$ be a $k$\nbd-local qudit Hamiltonian on $n$ qudits with local dimension $d$.
  Then, for any $\Delta \ge \|H\|$, there is a $k\lceil \log_2 d\rceil$\nbd-local qubit Hamiltonian $H'$ which perfectly simulates $H$ below energy $\Delta$.
\end{lemma}

\begin{proof}
  We use a local encoding $\mathcal{E}(M) = VMV^\dg$, where $V = W^{\ox n}$, and $W:\C^d \to (\C^2)^{\ox \lceil \log_2 d \rceil}$ is an arbitrary isometry.
  Write $P = \1 - WW^\dg$ for the projector onto the subspace orthogonal to the image of $W$ (if $d$ is a power of 2, $P=0$).
  Then we define the Hamiltonian
  \begin{equation}
    H' = \mathcal{E}(H) + \Delta' \sum_{i=1}^n P_i,
  \end{equation}
  for some $\Delta'>\Delta$.
  The nullspace of the positive semidefinite operator $P := \sum_{i=1}^n P_i$ is precisely the image of $V$, and the smallest nonzero eigenvalue of $P$ is $\Delta'$.
  So, as $\Delta < \Delta'$, $\cE$ is an encoding into the subspace $S_{\cE}=S_{\le \Delta(H')}$; and as $\Delta \ge \|H\|$, $H'|_{\le \Delta} = \mathcal{E}(H)|_{S_{\cE}}$.
  Thus $H'$ meets the requirements of \cref{dfn:perfectsim} and perfectly simulates $H$ below energy $\Delta$.
\end{proof}

Another case where we can achieve perfect simulation is the simulation of complex Hamiltonians with real Hamiltonians, using an alternative to the complex-to-real encoding of \cref{complex-to-real enc} where no single qubit corresponds to the ancilla qubit of \cref{complex-to-real enc}. This enables us to make the subspace encoding in the simulation local.

\begin{lemma} \label{complex-to-real sim}
For any integer $k$, let $H$ be a $k$-local qubit Hamiltonian. Then for any $\Delta \ge 2\|H\|$ there is a real $2k$-local qubit Hamiltonian $H'$ which simulates $H$ perfectly below energy $\Delta$.
\end{lemma}

\begin{proof}
  Let $H$ be a $k$-local qubit Hamiltonian, and let $h = \bigotimes_{i=1}^k\sigma_{s_i}$ with $s_i \in \{x,y,z\}$ be a $k$-local term in the Pauli decomposition of $H$.
  The complex-to-real encoding $\varphi$ from \cref{complex-to-real enc} maps individual Paulis as follows:
  \begin{align}
    \varphi(\1) &= \1\oplus\1\\
    \varphi(\sigma_{x,z}) &= \sigma_{x,z}\oplus\sigma_{x,z} = \1 \ox\sigma_{x,z} \\
    \varphi(\sigma_y) &= J(\sigma_y\oplus\sigma_y) = \sigma_y\ox\sigma_y.
  \end{align}
  For each qubit $j$ in the original Hamiltonian $H$, add an additional qubit labelled $j'$ and apply the map $\varphi$ separately to these pairs of qubits.
  This results in a term $h'$ on $2k$ qubits of the following form:
  \begin{equation}
   h'=\bigotimes_{j=1}^k \left(\ket{+_y}\bra{+_y}_{j'}\ox \sigma_{s_j}+\ket{-_y}\bra{-_y}_{j'} \ox \bar{\sigma}_{s_j}\right)
 \end{equation}
 Restricted to the space $S$ spanned by $\ket{+_y}^{\ox n}$ and $\ket{-_y}^{\ox n}$ on the ancilla qubits, $h$ is of the desired form.
 (Indeed, the restriction recovers the complex-to-real encoding of \cref{complex-to-real enc}.)
 Let $\widetilde{H}$ be the total Hamiltonian formed by the sum of the $h'$ terms.
 Then $\widetilde{H}$ is real and
  \begin{equation}
   \widetilde{H}|_S= \ket{+_y}\bra{+_y}^{\ox n}\ox H +\ket{-_y}\bra{-_y}^{\ox n}\ox \bar{H}
\end{equation}
We can add a term $\Delta'H_0$ where $\Delta'>\Delta$ and $H_0 = \sum_i (Y_{i'} Y_{(i+1)'}+\1)$ is zero on $S$ and is $\ge \1$ on $S^{\bot}$. The overall Hamiltonian $H'=\widetilde{H} +\Delta'H_0$ is therefore real and, since $\Delta' > \Delta\ge2\| H\|$, $S_{\le \Delta(H')}=S$ and $H'|_{\le \Delta}=\widetilde{H}|_S$.
\end{proof}


\subsection{Approximate simulation}

In general we may not be able to achieve perfect simulation, so it is natural to generalise this concept to allow approximate simulations.
If \cref{dfn:perfectsim:subspace} in \cref{dfn:perfectsim} no longer holds exactly for a map $\cE(M)=V(M\ox P + \bar{M}\ox Q)V^\dg$, it is not immediately clear how to generalise \cref{dfn:perfectsim:Hamiltonian}, as $H'_{\leq \Delta}$ and $\cE (H)$ now have support on different spaces.
However, if \cref{dfn:perfectsim:subspace} holds approximately such that $\|\cE(\1)-P_{\leq \Delta(H')}\|\le\eta$, then there exists an alternative encoding $\widetilde{\cE}(M)=\widetilde{V}(M\ox P + \bar{M}\ox Q)\widetilde{V}^\dg$ such that $\|\widetilde{V}-V\|\le\sqrt{2} \eta$ and $\widetilde{\cE}(\1)=P_{\leq \Delta (H')}$ (see \cref{tildeV} below); so we can compare $H'_{\leq \Delta}$ and $\widetilde{\cE} (H)$.

\begin{definition}
  \label{dfn:sim}
  We say that $H'$ is a $(\Delta,\eta,\epsilon)$-simulation of $H$ if there exists a local encoding $\cE(M)=V(M\ox P + \bar{M}\ox Q)V^\dg$  such that:
  \begin{enumerate}
  \item \label[condition]{dfn:sim:encoding}
    There exists an encoding $\widetilde{\cE}(M)=\widetilde{V}(M\ox P + \bar{M}\ox Q)\widetilde{V}^\dg$  such that $S_{\widetilde{\cE}}= S_{\le \Delta(H')}$ and $\|\widetilde{V} - V\| \le \eta$;
  \item \label[condition]{dfn:sim:Hamiltonian}
    $\| H'_{\le \Delta} - \widetilde{\mathcal{E}}(H)\| \le \epsilon$.
  \end{enumerate}
  We say that a family $\mathcal{F}'$ of Hamiltonians can simulate a family $\mathcal{F}$ of Hamiltonians if, for any $H \in \mathcal{F}$ and any $\eta,\epsilon >0$ and $\Delta \ge \Delta_0$ (for some $\Delta_0 > 0$), there exists $H' \in \mathcal{F}'$ such that $H'$ is a $(\Delta,\eta,\epsilon)$-simulation of $H$.

  We say that the simulation is efficient if, in addition, for $H$ acting on $n$ qudits and $H'$ acting on $m$ qudits, $\|H'\| = \poly(n,1/\eta,1/\epsilon,\Delta)$ and $m= \poly(n,1/\eta,1/\epsilon,\Delta)$; $H'$ is efficiently computable given $H$, $\Delta$, $\eta$ and $\epsilon$; each local isometry $V_i$ in the decomposition of $V$ from \cref{thm:localenc} is itself a tensor product of isometries which map to $O(1)$ qudits; and there is an efficiently constructable state $\ket{\psi}$ such that $P\ket{\psi}=\ket{\psi}$.
\end{definition}

Note that different notions of computational efficiency could be used in \Cref{dfn:sim}. For all the simulations considered in this paper, $H'$ is computed using a polynomial-time classical algorithm, and $\ket{\psi}=\ket{00\dots 0}$.

We usually think of $\Delta$ as satisfying $\Delta \gg \|H\|$.
We may interpret \cref{dfn:sim} as stating that $H'_{\leq \Delta}$ is close to an encoding $\widetilde{\cE}(H)$ of $H$, and that the encoding map $\widetilde{\cE}$ is close to a local encoding $\cE$.
However, we assume that $\cE$ is the map that we understand and have access to, whereas all we know about $\widetilde{\cE}$ is that it exists.

A perfect simulation of $H$ by $H'$ below energy $\Delta$ is a $(\Delta,\eta,\epsilon)$-simulation of $H$ with $\eta=\epsilon=0$.
Observe that every local encoding is a perfect simulation with $\Delta=\infty$.
Reducing the inaccuracy $\eta$, $\epsilon$ of the simulation will typically require expending more ``effort'', e.g.\ by increasing the strength of the local interactions.

An alternative definition might try to compare $H'_{\leq \Delta}$ and $\cE(H)$ even though they have different support.
This would be essentially equivalent to our definition because, from \cref{lem:closeencs} and the reverse triangle inequality,
\begin{equation}
   \left| \| H'_{\leq \Delta} -\cE(H) \| - \| H'_{\le \Delta} -\widetilde{\cE}(H) \|\right| \le 2\|V-\widetilde{V}\| \| H^{\oplus p} \oplus \bar{H}^{\oplus q}\| \le 2\eta \|H\|.
\end{equation}
Thus the two definitions are equivalent up to a $O(\eta \|H\|)$ term.
Our simulations will in general assume that $\eta = O(1/\poly(\|H\|))$, making the difference negligible.
It is also worth noting that this alternative definition appears to result in worse bounds in \cref{lem:compsim} and \cref{time-evolution} below.

We remark that our physically motivated definition of simulation is very similar to one previously introduced by Bravyi and Hastings~\cite{Bravyi-Hastings}.
The main differences are:
\begin{enumerate}
\item
  The second part of the definition in \cite{Bravyi-Hastings} is stated as
  \begin{equation}
    \| H - \widetilde{V}^\dg H' \widetilde{V} \| \le \epsilon.
  \end{equation}
  But we have $\| H - \widetilde{V}^\dg H' \widetilde{V} \| = \| \widetilde{V} H \widetilde{V}^\dg - \widetilde{V}\widetilde{V}^\dg H' \widetilde{V}\widetilde{V}^\dg \| = \| \widetilde{V} H \widetilde{V}^\dg - H'_{\le \Delta} \| $, which matches the term $\|H'_{\le \Delta} -\widetilde{\cE}(H) \|$ in our definition, except that our encoding $\widetilde{\cE} (H)$ may be of the more general form $\widetilde{V}( H^{\oplus p} \oplus \bar{H}^{\oplus q})\widetilde{V}^\dg$.
  As discussed above, this is essential to enable e.g.\ complex Hamiltonians to be encoded as real Hamiltonians.

\item
  We insist that $\mathcal{E}$ is local, whereas~\cite{Bravyi-Hastings} deliberately does not impose any restriction on the isometry $V$, other than to say it should be sufficiently simple in practice.
  This enables us to find stronger implications of our notion of simulation for error-tolerance and computational complexity.
\end{enumerate}

We also remark that, although \cref{dfn:sim} requires simulation in the low-energy subspace, this can readily be generalised to other types of subspace, by replacing $P_{\le \Delta(H')}$ by a projector onto the subspace of interest.
However, some of the physical consequences of \cref{dfn:sim} later in this section do depend on the simulation being in the low-energy subspace.
All the simulations we construct will achieve this.

We now prove the previously promised claim that if the isometry $V$ used in an encoding approximately maps to the ground space of $H'$, there exists an isometry $\widetilde{V}$ close to $V$ which maps exactly to this ground space.
See~\cite{Bravyi-Hastings} for a similar result.

\begin{lemma} \label{tildeV}%
  Let $\cE:\cB(\cH)\to\cB(\cH')$ be a subspace encoding of the form $\cE(M)=V(M\ox P+\bar{M}\ox Q)V^{\dg}$, and let $\Pi$ be the projector onto a subspace $S\subseteq \mathcal{H}'$.
  If $\|\Pi-\cE(\1)\|<1$, then there exists an isometry $\widetilde{V}:\mathcal{H}\to\mathcal{H'}$ such that $\|\widetilde{V}-V\|\leq \sqrt{2} \|\Pi-\cE(\1)\|$ and the corresponding encoding $\widetilde{\cE}(M)=\widetilde{V}(M\ox P+\bar{M}\ox Q)\widetilde{V}^{\dg}$ satisfies $\widetilde{\cE}(\1)=\Pi$.
\end{lemma}

\begin{proof}
Recall that $\cE(\1)$ is a projector.
  If $\|\Pi-\cE(\1)\|<1$, then $\rank(\Pi)=\rank(\cE(\1))$ and hence there exists a unitary $U$ on $\mathcal{H}'$ such that $\Pi=U\cE(\1)U^\dg$.
  One can show using Jordan's lemma that $U$ can be chosen to obey the bound $\|U-\1\|\leq \sqrt{2}\|\Pi-\cE(\1)\|$; the short argument is contained in the proof of Lemma~3 in~\cite{Bravyi-Hastings}.

  Defining $\widetilde{V}=UV$, we have $\widetilde{\cE}(\1)=U\cE(\1)U^\dg=\Pi$ and
  \begin{equation}
    \|\widetilde{V}-V\|\leq\|U-\1\|\|V\|\leq \sqrt{2}\|\Pi-\cE(\1)\|
  \end{equation}
  as desired.
\end{proof}

Importantly, the notion of simulation we use is transitive: if $A$ simulates $B$, and $B$ simulates $C$, then $A$ simulates $C$.
We now formalise this as a lemma; a very similar result to this was shown by Bravyi and Hastings~\cite{Bravyi-Hastings}, but as our encodings are somewhat more general to those they consider we include a proof.

\begin{lemma}
  \label{lem:compsim}
  Let $A$, $B$, $C$ be Hamiltonians such that $A$ is a $(\Delta_A,\eta_A,\epsilon_A)$-simulation of $B$ and $B$ is a $(\Delta_B,\eta_B,\epsilon_B)$-simulation of $C$.
  Suppose $\epsilon_A,\epsilon_B \leq \|C\|$ and $\Delta_B\geq\|C\|+2\epsilon_A+\epsilon_B$.
  Then $A$ is a $(\Delta,\eta,\epsilon)$-simulation of $C$, where $\Delta \geq \Delta_B-\epsilon_A,$
  \begin{equation}
  \eta=\eta_A+\eta_B+O\left(\frac{\epsilon_A}{\Delta_B-\|C\|+\epsilon_B}\right)\quad \text{ and }\quad \epsilon =\epsilon_A + \epsilon_B+O\left(\frac{\epsilon_A\|C\|}{\Delta_B-\|C\|+\epsilon_B}\right).
\end{equation}
\end{lemma}

Note that any good simulation should satisfy $\Delta_B \gg \|C\|$ (see \cref{prop:partition} below for one reason why) in which case the condition on $\Delta_B$ is easily satisfied and we have $\eta=\eta_A+\eta_B+o(1)$ and $\epsilon \approx \epsilon_A+\epsilon_B$.

\begin{proof}
  We closely follow the argument of~\cite[Lemma 3]{Bravyi-Hastings}.
  Let $\mathcal{E}_A$ be the local encoding corresponding to the simulation of $B$ with $A$, and let $\mathcal{E}_B$ be the local encoding corresponding to the simulation of $C$ with $B$.
  We will use the composed map $\mathcal{E} = \mathcal{E}_A \circ \mathcal{E}_B$ to simulate $C$ with $A$.
  By \cref{lem:chaining}, this map is indeed a local encoding.

  Let $V_A$ and $V_B$ be the isometries in the definition of $\mathcal{E}_A$ and $\mathcal{E}_B$.
  Recall from the definition of simulation that there exist isometries $\widetilde{V}_A$, $\widetilde{V}_B$ such that $\|\widetilde{V}_A - V_A\| \le \eta_A$, $\|\widetilde{V}_B - V_B\| \le \eta_B$, $\widetilde{V}_A\widetilde{V}_A^\dg = P_{\le\Delta_A(A)}$, $\widetilde{V}_B\widetilde{V}_B^\dg = P_{\le\Delta_B(B)}$.
  We define the encodings $\widetilde{\mathcal{E}}_A$, $\widetilde{\mathcal{E}}_B$ to be the encodings obtained by replacing $V_A$ with $\widetilde{V}_A$ and $V_B$ with $\widetilde{V}_B$.
  Note that composing these maps to obtain $\widetilde{\mathcal{E}}_A \circ \widetilde{\mathcal{E}}_B$ makes sense ($\widetilde{\mathcal{E}}_B$ maps $C$ to the low-energy part of $B$, and $\widetilde{\mathcal{E}}_A$ maps all of $B$ to the low-energy part of $A$).

  Let $N$ be the dimension of $S_{\leq \Delta_B(B)}$.
  By \cref{lem:eigenvalues}, the $N$th smallest eigenvalue of $B$ is bounded by $\lambda_N(B)\leq \|C\| +\epsilon_B$.
  Therefore the condition $\Delta_B\geq\|C\|+2\epsilon_A+\epsilon_B$ allows us to put a lower bound on $\Delta_G$, the spectral gap between the $N$th and $(N+1)$th eigenvalues of $B$:
  \begin{equation}
    \Delta_G=\lambda_{N+1}(B)-\lambda_{N}(B)>\Delta_B-\|C\|-\epsilon_B\geq 2\epsilon_A.
  \end{equation}
  Let $\widetilde{\mathcal{E}}_A(B) = \widetilde{V}_A (B^{\oplus p}\oplus \bar{B}^{\oplus q})\widetilde{V}_A^\dg$.
  By \cref{lem:eigenvalues}, $\lambda_{N(p+q)}(A) \le \lambda_{N}(B)+\epsilon_A$ and $\lambda_{N(p+q)+1}(A) \geq \lambda_{N+1}(B)-\epsilon_A$, so the condition $\Delta_G>2\epsilon_A$ implies that there exists $\Delta$ such that $\lambda_{N(p+q)}(A)< \Delta < \lambda_{N(p+q)+1}(A)$.
  Furthermore, since $\lambda_{N(p+q)+1}(A)\geq \lambda_{N+1}(B)-\epsilon_A>\Delta_B-\epsilon_A$, we can choose $\Delta$ to be at least as big as $\Delta_B-\epsilon_A$.

  Let $B'=B^{\oplus p}\oplus \bar{B}^{\oplus q}$, so we can write $\widetilde{\mathcal{E}}_A(B) = \widetilde{V}_A B' \widetilde{V}_A^\dg$.
  It is shown in the proof of~\cite[Lemma 3]{Bravyi-Hastings} that there exists a unitary operator $U$ such that
  \begin{equation}\label{eq:exactu}
    S_{\le \Delta(A)} = U \widetilde{V}_A S_{\le \Delta_B(B')}
  \end{equation}
  and $\|U - \1\| \le 2 \sqrt{2} \epsilon_A/\Delta_G$.
  That is, $U\widetilde{V}_A$ maps the low-energy subspace of $B'$ precisely onto the low-energy subspace of $A$.
  Note that the existence of such a $U$ is nontrivial, as all we know in advance from the fact that $A$ simulates $B$ is that $\widetilde{V}_A$ maps all of $B'$ into the less low-energy subspace $S_{\le \Delta_A(A)}$.

  The composed approximate encoding in the simulation of $C$ by $A$ will be $\widetilde{\mathcal{E}}(M) = U\widetilde{\mathcal{E}}_A(\widetilde{\mathcal{E}}_B(M))U^\dg$.
  By \cref{eq:exactu}, $\widetilde{\mathcal{E}}$ maps the Hilbert space of $C$ onto $S_{\le \Delta(A)}$.
  The overall isometry $\widetilde{V}$ in the encoding $\widetilde{\mathcal{E}}_A \circ \widetilde{\mathcal{E}}_B$ is obtained from the isometry $V$ in the encoding $\mathcal{E}$ by replacing $V_A$ with $\widetilde{V}_A$ and $V_B$ with $\widetilde{V}_B$.
  By the triangle inequality and \cref{lem:chaining}, $\|V - \widetilde{V}\| \le \eta_A + \eta_B$, so
  \begin{equation}
    \eta=\|V - U\widetilde{V}\| \le \eta_A + \eta_B + O(\epsilon_A\Delta_G^{-1}).
  \end{equation}
  Therefore, $\mathcal{E}$ meets \cref{dfn:sim:encoding} from \cref{dfn:sim} for simulation of $C$ with $A$.

  It remains to show \cref{dfn:sim:Hamiltonian}.
  We aim to bound $\|A_{\le \Delta} - U \widetilde{\mathcal{E}}_A(\widetilde{\mathcal{E}}_B(C))U^\dg \|$, which, by the triangle inequality, is upper-bounded by
  \begin{equation} \label{eq:terms}
    \begin{split}
    &\|A_{\le \Delta} - U \widetilde{\mathcal{E}}_A(\widetilde{\mathcal{E}}_B(C))U^\dg \| \\
    &\qquad \leq \|A_{\le \Delta} - U \widetilde{\mathcal{E}}_A(B_{\le \Delta_B})U^\dg \|
         + \|U \widetilde{\mathcal{E}}_A(B_{\le \Delta_B})U^\dg
             - U \widetilde{\mathcal{E}}_A(\widetilde{\mathcal{E}}_B(C))U^\dg \|.
    \end{split}
  \end{equation}
  The second term in \cref{eq:terms} is precisely equal to $\|B_{\le \Delta_B} - \widetilde{\mathcal{E}}_B(C) \|$.
  By the assumption of the present lemma that $B$ is a $(\Delta_B,\eta_B,\epsilon_B)$-simulation of $C$, this term is upper-bounded by $\epsilon_B$.
  In order to deal with the first term in \cref{eq:terms}, we rewrite it as
  \begin{equation}
    \|  A_{\le \Delta} U \widetilde{V}_A - U \widetilde{V}_A B'_{\le \Delta_B}  \|.
  \end{equation}
  We write $U= \1 + M$, so
  \begin{align}
    A_{\le \Delta} U \widetilde{V}_A - U \widetilde{V}_A B'_{\le \Delta_B}
    &=  P_{\le \Delta(A)} ( A_{\le \Delta} U \widetilde{V}_A - U \widetilde{V}_A B'_{\le \Delta_B})  P_{\le \Delta_B(B')} \\
    &= P_{\le \Delta(A)} ( A\widetilde{V}_A - \widetilde{V}_A B')  P_{\le \Delta_B(B')}\\
    &\qquad+ A_{\le \Delta} M \widetilde{V}_A P_{\le \Delta_B(B')} - P_{\le \Delta(A)} M \widetilde{V}_A B'_{\le \Delta_B}.
  \end{align}
  For the first part,
  \begin{equation}
    \|P_{\le \Delta(A)} ( A\widetilde{V}_A - \widetilde{V}_A B')  P_{\le \Delta_B(B')}\|
    \leq \| A\widetilde{V}_A - \widetilde{V}_A B'\| = \| A_{\leq \Delta_A}- \widetilde{V}_A B'\widetilde{V}_A^\dg\|
    \leq \epsilon_A
  \end{equation}
  by simulation of $B$ with $A$.
  The second part is bounded by $\|M\| \|A_{\le \Delta}\|$ and the third by $\|M\| \|B'_{\le \Delta_B}\|$.
  We have $\|M\| = O(\epsilon_A \Delta_G^{-1})$ by \cref{eq:exactu}.
  By simulation of $B$ with $A$ and \cref{eq:exactu}, $\|A_{\le \Delta}\| \le \|B'_{\le \Delta_B}\| + \epsilon_A$; by simulation of $C$ with $B$, $\|B'_{\le \Delta_B}\|=\|B_{\le \Delta_B}\| \le \|C\|+\epsilon_B$.
  Combining all the terms, we get the overall bound that
  \begin{equation}
    \|A_{\le \Delta} - U \widetilde{\mathcal{E}}_A(\widetilde{\mathcal{E}}_B(C))U^\dg  \|
    \leq \epsilon_A+\epsilon_B +2\sqrt{2}\epsilon_A \Delta_G^{-1}(\|C\| +\epsilon_A+2 \epsilon_B ).
  \end{equation}
  Since $\epsilon_A,\epsilon_B \leq \|C\|$ and $\Delta_B\leq \Delta_G+\|C\|+\epsilon_B$, we have that the overall error $\epsilon$ is
  \begin{equation}
    \epsilon =\epsilon_A+\epsilon_B+O\left(\frac{\epsilon_A\|C\|}{\Delta_B-\|C\|+\epsilon_B}\right)
  \end{equation}
  as claimed.
\end{proof}

Later we will see that certain families of Hamiltonians are extremely powerful simulators: they can simulate any other Hamiltonian.

\begin{definition}\label{dfn:universal}
  We say that a family of Hamiltonians is a \emph{universal simulator}, or is \emph{universal}, if \emph{any} (finite-dimensional) Hamiltonian can be simulated by a Hamiltonian from the family.
  We say that the universal simulator is \emph{efficient} if the simulation is efficient for all local Hamiltonians.
\end{definition}
Although we restrict to finite-dimensional Hamiltonians in this definition, infinite-dimensional cases can be treated via standard discretisation techniques.
Indeed, we will see one such example later.
We restrict our notion of efficiency to local Hamiltonians, as this is a natural class of Hamiltonians which have efficient descriptions themselves.

First, however, we will show that the definition of simulation we have arrived at has some interesting consequences.


\subsection{Simulation and static properties}

First we show that Hamiltonian simulation does indeed approximately preserve important physical properties of the simulated Hamiltonian.
Although this is effectively immediate for perfect simulations from the definition of encodings, for approximate simulations we need to check how the level of inaccuracy in the simulation translates into a level of inaccuracy in the property under consideration.
We first do this for eigenvalues; essentially the same result was shown in~\cite{Bravyi-Hastings} but we include a proof for completeness.

\begin{lemma}\label{lem:eigenvalues}
  Let $H$ act on $(\C^d)^{\ox n}$, let $H'$ be a ($\Delta,\eta,\epsilon$)-simulation of $H$, and let $\lambda_i(H)$ (resp.\ $\lambda_i(H')$) be the $i$'th smallest eigenvalue of $H$ (resp.\ $H'$).
  Then for all $1 \le i \le d^n$ and all $j$ such that $(i-1)(p+q)+1\le j \le i(p+q)$, $|\lambda_i(H) - \lambda_{j}(H')| \le \epsilon$ (where the integers $p,q$ are those appearing in simulation's encoding).
\end{lemma}

\begin{proof}
  For any $i,j$ satisfying the above conditions, $\lambda_i(H)=\lambda_j(\cE(H))$ by the definition of an encoding.
  By \cref{dfn:sim:encoding} of \cref{dfn:sim}, the spectrum of $\widetilde{\cE}(H)$ is the same as the spectrum of $\cE(H)$.
  By \cref{dfn:sim:Hamiltonian} of \cref{dfn:sim} and Weyl's inequality $|\lambda_j(\widetilde{\cE}(H)) - \lambda_j(H')| \le \|\widetilde{\cE}(H)-H'_{\le \Delta}\|$, each eigenvalue differs from its counterpart by at most $\epsilon$.
\end{proof}

Next we verify that simulation approximately preserves partition functions.

\begin{proposition}
  \label{prop:partition}
  Let $H'$ on $m$ $d'$-dimensional qudits be a ($\Delta,\eta,\epsilon$)-simulation of $H$ on $n$ $d$-dimensional qudits, with $\|H'_{\leq \Delta}-\widetilde{\cE}(H)\|\leq \epsilon$ for some encoding $\widetilde{\cE}(H)=\widetilde{V}(H^{\oplus p }\oplus \bar{H}^{\oplus q})\widetilde{V}^\dg$.
  Then the relative error in the simulated partition function evaluated at $\beta$ satisfies
  \begin{equation}
    \frac{|\mathcal{Z}_{H'}(\beta)-(p+q)\mathcal{Z}_H(\beta)|}{(p+q)\mathcal{Z}_H(\beta)}
    \le \frac{(d')^m e^{-\beta\Delta}}{(p+q)d^n e^{-\beta \|H\|}}+(e^{\epsilon\beta}-1) .
  \end{equation}
\end{proposition}

\begin{proof}
  Let $S$ be the low-energy subspace of $H'$, $S=\operatorname{Im}(\widetilde{V})$.
  We have
  \begin{equation}
  (p+q)\mathcal{Z}_{H}(\beta) = (p+q)\tr(e^{-\beta H}) = \tr(e^{-\beta \widetilde{\cE}(H)}|_S)
  \end{equation}
   and hence
  \begin{align}
   & \frac{|\mathcal{Z}_{H'}(\beta)-(p+q)\mathcal{Z}_{H}(\beta) |}{(p+q)\mathcal{Z}_{H}(\beta) }\\
   &  = \frac{|\tr(e^{-\beta H'})-\tr(e^{-\beta \widetilde{\cE}(H)}|_S)|}{\tr(e^{-\beta \widetilde{\cE}(H)}|_S)} \\
    & \leq \frac{|\tr(e^{-\beta H'})-\tr(e^{-\beta H'|_{\leq \Delta}})|}{(p+q)\tr(e^{-\beta H})}
      + \frac{|\tr(e^{-\beta H'|_{\leq \Delta}})
      - \tr(e^{-\beta \widetilde{\cE}(H)|_S})|}{\tr(e^{-\beta \widetilde{\cE}(H)}|_S)}.
  \end{align}
  For the first term, the numerator is upper-bounded by $(d')^m e^{-\beta\Delta}$, whereas in the denominator $\tr(e^{-\beta H})$ is lower-bounded by $d^n e^{-\beta \|H\|}$.
  For the second term, we write $\lambda_k$ for the $k$'th eigenvalue of $H$ (in nonincreasing order), and $\lambda_k+\epsilon_k$ for the $k$'th eigenvalue of $H'|_{\leq \Delta}$ (in the same order), and have
  \begin{equation}
    |\tr(e^{-\beta H'|_{\leq \Delta}})-\tr(e^{-\beta \widetilde{\cE}(H)|_S})|
    \leq \sum_k |e^{-\beta(\lambda_k + \epsilon_k)} - e^{-\beta \lambda_k}|
    = \sum_k e^{-\beta \lambda_k} |e^{-\beta \epsilon_k}-1|.
  \end{equation}
  By \cref{lem:eigenvalues}, $|\epsilon_k|\leq \epsilon$ for all $k$, so we have $|e^{-\beta \epsilon_k}-1| \leq e^{\beta \epsilon}-1$, and thus the relative error is upper-bounded by
  \begin{equation}
    \frac{(d')^m e^{-\beta\Delta}}{(p+q)d^n e^{-\beta \|H\|}}+(e^{\epsilon\beta}-1)
  \end{equation}
  as claimed.
\end{proof}
We remark that if we choose $\Delta \gg \|H\| + (m \log d' - n \log d - \log(p+q))/\beta$ and $\epsilon \ll 1/\beta$ then this relative error tends to zero.
All the simulations we construct allow us to choose $\Delta \gg m-n$, so these scalings are possible.


\subsection{Simulation and time-evolution}

We showed in \cref{prop:encodingswork} that encodings allow perfect simulation of time-evolution.
We now confirm that this holds for simulations too, up to a small approximation error.

\begin{proposition}
\label{prop:time-evolution}
  Let $H'$ be a ($\Delta,\eta,\epsilon)$-simulation of $H$ with corresponding encoding $\cE=V(M\ox P +\bar{M} \ox Q)V^\dg$.
  Then for any density matrix $\rho'$ in the encoded subspace, so that $\cE(\1)\rho'=\rho'$,
  \begin{equation}
    \|e^{-iH't}\rho'e^{iH't}-e^{-i\cE(H)t}\rho'e^{i\cE(H)t}\|_1\leq 2\epsilon t +4 \eta
  \end{equation}
\end{proposition}

We will need the following simple \namecref{Lipschitz} establishing the Lipschitz constant of the exponential map for Hermitian operators:
\begin{lemma}\label{Lipschitz}
  For $H$, $E$ Hermitian and $\norm{\cdot}$ any unitarily invariant norm,\linebreak $\norm{e^{i(H+E)} - e^{iH}} \le \norm{E}$.
\end{lemma}
\begin{proof}
  \begin{align}
    \norm{e^{i(H+E)} - e^{iH}}
    &= \Norm{ \int_0^1 \frac{\dd}{\dd s} e^{i(H+sE)} \dd s }
    = \Norm{ \int_0^1 iE e^{i(H+sE)} \dd s }\\
    &\le \int_0^1 \Norm{ E e^{i(H+sE)} } \dd s
    = \int_0^1 \norm{E} \dd s
    = \norm{E}.
  \end{align}
\end{proof}

\begin{proof}[of \cref{prop:time-evolution}]
Recall that by the definition of simulation there exists an alternative encoding $\widetilde{\cE}(M)=\widetilde{V}(M\ox P+\bar{M}\ox Q)\widetilde{V}^{\dg}$ such that\linebreak $\widetilde{\cE}(\1)=P_{\le \Delta(H')}$ and $\|\widetilde{V}-V\|\le \eta$.
 Let $\widetilde{\rho}=\widetilde{V}V^\dg\rho' V\widetilde{V}^\dg$. Then
  \begin{align}
    \| e^{-iH't}&\rho'e^{iH't} - e^{-i\cE(H)t}\rho'e^{i\cE(H)t}\|_1 \\
    \begin{split}
      \leq\,
      &\|e^{-iH't}\rho'e^{iH't}-e^{-iH't}\widetilde{\rho}e^{iH't}\|_1
      + \| e^{-iH't}\widetilde{\rho}e^{iH't}
           - e^{-i\widetilde{\cE}(H)t}\widetilde{\rho}e^{i\widetilde{\cE}(H)t}\|_1\\
      &+ \| e^{-i\widetilde{\cE}(H)t}\widetilde{\rho}e^{i\widetilde{\cE}(H)t}-e^{-i\cE(H)t}\rho'e^{i\cE(H)t}\|_1
    \end{split}
  \end{align}
  by the triangle inequality.
  Since $\rho'$ is in the encoded subspace, we know that $VV^{\dg}\rho'VV^{\dg}=\rho'$.
  Therefore \cref{lem:conjdiff} lets us bound the first term by $\|\rho'-\widetilde{\rho}\|_1 \le 2\|\widetilde{V}V^\dg-VV^\dg\| \le 2\eta$.
  Similarly, noting that
  \begin{equation}
    e^{-i\widetilde{\cE}(H)t}\widetilde{\rho}e^{i\widetilde{\cE}(H)t}
    = \widetilde{V}V^\dg e^{-i\cE(H)t}\rho'e^{i\cE(H)t}V\widetilde{V}^\dg,
  \end{equation}
  we use \cref{lem:conjdiff} to bound the third term by $2\|\widetilde{V}V^\dg-VV^\dg\| \le 2\eta$.
  Finally, for the second term, we note that $P_{\le \Delta(H')} \widetilde{\rho}=\widetilde{\rho}$, so $e^{-iH't}\widetilde{\rho}e^{iH't} = e^{-iH'_{\le \Delta}t}\widetilde{\rho}e^{iH'_{\le \Delta}t}$, and by \cref{lem:conjdiff} again this term is bounded by
  \begin{equation}
    2\|e^{iH'_{\le \Delta}t}-e^{i\widetilde{\cE}(H)t}\|
    \le 2t\|H'_{\le \Delta} -\widetilde{\cE}(H)\| \le 2\epsilon t,
  \end{equation}
  where we have used \cref{Lipschitz}.
\end{proof}

\begin{corollary}
\label{time-evolution}
  Suppose in addition to the conditions of \cref{time-evolution} that $\cE$ is a standard encoding. Let $\cEs(\rho)=V(\rho\ox \sigma)V^\dg$ for some state $\sigma$ satisfying $P\sigma=\sigma$, and let $F(\rho')=\tr_E[(\1\ox P)V^\dg\rho'V]$ as defined in \cref{eq:fandb}.
  Then
  \begin{gather}
    \|e^{-iH't}\cEs(\rho)e^{iH't}-\cEs(e^{-iHt}\rho e^{iHt})\|_1\leq 2\epsilon t +4 \eta,\\
    \|F(e^{-iH't}\rho'e^{iH't})-e^{-iHt}F(\rho')e^{iHt}\|_1\leq 2\epsilon t +4 \eta.
  \end{gather}
\end{corollary}

\begin{proof}
  The first statement follows from setting $\rho'=\cEs(\rho)$ in \cref{prop:time-evolution} and noting that $e^{-i\cE(H)t}\cEs(\rho)e^{i\cE(H)t}=\cEs(e^{-iHt}\rho e^{iHt})$.
  The second statement follows from $F(e^{-i\cE(H)t}\rho'e^{i\cE(H)t})=e^{-iHt}F(\rho')e^{iHt}$ and the fact that $F$ is trace-nonincreasing.
\end{proof}


\subsection{Errors and noise}

An important question for any simulation technique is how errors affecting the simulator relate to errors on the simulated system.
Understanding this in full detail will depend strongly on the physical noise model being considered and the implementation details of the simulation.
However, our notion of simulation via local encodings enables us to make some general statements about errors.

First, we show that a local error on the simulator does not map between the forward-evolving and backward-evolving parts of the simulator.
This implies the existence of a corresponding local error on the original system by using the $F$ map to extract the forward-evolving part.
Second, we show that for the types of encoding used in this paper, a stronger result holds: any local error on an encoded state is equal to the encoding of a local error on the original system.
Finally, we show that, under a reasonable physical assumption, any error on the simulator is close to an error that acts only within the encoded subspace.
This allows us to continue to simulate time-evolution and measurement following an error.

\begin{theorem}\label{errors}
  Let $\cE(M)=V(M\ox P+\bar{M}\ox Q)V^\dg$ be a local encoding, where $M$ acts on $n$ qudits, and let $\rho'$ be a state on the encoded subspace such that $\cE(\1)\rho'=\rho'$.
  Let $\cN'$ be a CP-map whose Kraus operators each act on at most $l<n$ qudits of the simulator system.

  \begin{enumerate}[label=\arabic{*}.,ref=\arabic{*}]
  \item \label[part]{errors:1}
    Let $P'=V(\1\ox P)V^{\dg}$ and $Q'=V(\1 \ox Q)V^{\dg}$.
    Then
    \begin{equation}
      P'\cN'(\rho')=P'\cN'(P'\rho') \quad \text{and} \quad Q'\cN'(\rho')=Q'\cN'(Q'\rho').
    \end{equation}

  \item \label[part]{errors:2}
    Let $\cEs(\rho)=V(\rho \ox \sigma)V^{\dg}$ for a density matrix $\sigma$ satisfying $P\sigma=\sigma$.
    Then the map defined by $\cN(\rho)=F(\cN'(\cEs(\rho)))$ is a CP-map whose Kraus operators act on at most $l$ qudits of the original system.
  \end{enumerate}
\end{theorem}

\begin{proof}
  Let $\cN'(\rho')=\sum_k N'_k \rho' N_{k}^{\prime\dg}$.
  For a given $k$, the Kraus operator $N'_k$ acts on only $l$ qudits of the simulator system.
  Therefore $N'_k$ must act trivially on at least one subsystem $\cH'_j$.
  Recall from \cref{thm:localenc} that there exists a projector $P_{E_j}$ which acts only on the ancilla $E_j$ such that $(\1 \ox P_{E_j})P=P$ and $(\1 \ox P_{E_j})Q=0$.
  Defining $P'_j=\1\ox V_j P_{E_j} V_j^{\dg}$, we have $P'_j P'=P'$ and $P'_j Q'=0$.
  Note that $P'_j$ acts non-trivially only on $\cH'_j$ and so commutes with $N'_k$.
  Therefore
  \begin{equation}
    P'N'_k\cE(\1)=P'P'_j N'_k\cE(\1)=P' N'_k P'_j\cE(\1)=P' N'_k P'_j (P'+Q')=P'N'_k P'.
  \end{equation}
  So, remembering that $\rho'$ is in the encoded subspace and satisfies $\rho'=\cE(\1)\rho'$, we have
  \begin{equation}
    P'\cN(\rho')=\sum_k P' N'_k \cE(\1)\rho'N_k^{\prime\dg}
    = \sum_k P' N'_k P'\rho' N_k^{\prime\dg}=P'\cN'(P' \rho').
  \end{equation}
  The statement for $Q$ follows analogously.

  We now prove the second part of the theorem.
  $\cN(\rho)$ is clearly CP, since it is defined by a composition of CP maps.
  Let the spectral decomposition of $\sigma$ be given by $\sigma=\sum_j \lambda_j \proj{\psi_j}$.
  Extend $\{\ket{\psi_j}\}_j$ to a basis for the subspace of the ancilla $E$ given by the support of $P$.
  Then
  \begin{align}
    \cN(\rho)
    &=F(\cN'(\cEs(\rho)))\\
    &=\tr_E[(\1 \ox P)\sum_{k} V^{\dg}N'_k V(\rho \ox \sigma)V^{\dg}N_k^{\prime \dg}V(\1 \ox P)]\\
    &=\sum_{i,j,k} (\1 \ox \bra{\psi_i}) V^{\dg} N'_k V(\rho \ox \lambda_j \proj{\psi_j})V^{\dg}N_k^{\prime \dg}V(\1 \ox \ket{\psi_i}) \\
    &=\sum_{k,i,j} N_{i,j,k} \rho N_{i,j,k}^{\dg}
  \end{align}
  where $N_{i,j,k}=\sqrt{\lambda_j} (\1 \ox \bra{\psi_i})V^{\dg}N'_k V(\1 \ox \ket{\psi_j})$ are the Kraus operators of $\cN$.
  Since $\cE$ is a local encoding, the isometry $V$ may be chosen to be local by \cref{thm:localenc}, so $V^{\dg}N'_kV$ acts non-trivially on at most $l$ qudits of the original system.
  Therefore the Kraus operators $N_{i,j,k}$ act non-trivially on at most $l$ qudits, as claimed.
\end{proof}

For a general encoding with a corresponding map on states $\cEs(\rho)=V(\rho\ox \sigma)V^\dg$, the error $\cN'$ may entangle $\rho$ and $\sigma$, so it is not possible in general to show that $\cN'(\cEs(\rho))\approx \cEs(\cN(\rho))$.
However, if $\operatorname{rank}(P)=1$ (as is the case in all our simulations) then we are able to get a stronger result, which composes more straightforwardly with our other results.

\begin{corollary}
\label{errorsEstate}
  Let $\cE(M)=V(M\ox P+\bar{M}\ox Q)V^\dg$ be a local encoding with $\operatorname{rank}(P)=1$ and let $\cEs(\rho)=V(\rho\ox P)V^\dg$.
 Let $\mathcal{N'}$ and $\cN$ be the CP-maps given in \cref{errors}.
  Then
  \begin{equation}
    \cE(\1)\cN'(\cEs(\rho))\cE(\1)=\cEs(\cN(\rho)).
  \end{equation}
\end{corollary}

\begin{proof}
  Let the Kraus operators of $\cN'$ be given by $N'_k$.
  Since rank$(P)=1$, we must have $P=\proj{\psi}$ for some state $\ket{\psi}$ on the ancilla system $E$.
  Since $Q'\cEs(\rho)=0=\cEs(\rho)Q'$, where $Q'$ is defined as in \cref{errors}, \cref{errors:1} of that \lcnamecref{errors} shows that $Q'\cN'(\cEs(\rho))=0=\cN'(\cEs(\rho))Q'$. Then writing $\cE(\1)=P'+Q'$, we have
  \begin{align}
    &\cE(\1)\cN'( \cEs(\rho))\cE(\1)=P'\cN'(\cEs(\rho))P' \\
    &=V(\1 \ox  \proj{\psi} )V^{\dg}\left(\sum_k N'_k V(\rho\ox \proj{\psi} )V^{\dg} N^{\prime \dg}_k \right) V(\1 \ox \proj{\psi} )V^{\dg} \\
    &= V \left(\sum_k N_k \rho N_k^{\dg} \ox \proj{\psi} \right)V^{\dg}=\cEs(\cN(\rho)),
  \end{align}
  where we recall from the proof of \cref{errors} that the Kraus operators of $\cN(\rho)$ are given by $N_k=(\1 \ox \bra{\psi})V^{\dg}N'_k V(\1\ox \ket{\psi})$ (the sum over $i$ and $j$ is not necessary when rank$(P)=1$).
\end{proof}

\Cref{errorsEstate} is the strongest general result relating errors on the simulator and simulated systems that one could hope for: it states that any error (CP-map) on the simulator system corresponds naturally to simulating an error (CP-map) on the simulated system.

Even in the more general setting of \cref{errors}, we interpret the map $\cN(\rho)=F(\cN'\cEs(\rho))$ as the error on the original system corresponding to $\cN'$. This is because by \cref{errors:1} of \cref{errors}  we have $B(\cEs(\rho))=0$, and therefore by \cref{eqn:Fmeas}, for any observable $A$,
\begin{equation}
  \tr[A\cN(\rho)]=\tr[\cE(A)\cN'(\cEs(\rho))].
\end{equation}

Although $\cN'$ may not map between the forwards and backwards parts of the encoded space, it may take a state out of the encoded subspace. But in order to implement a local measurement with \cref{localmeas} and time-evolve with \cref{time-evolution}, we need $\rho'=\cN'(\cEs(\rho))$ to be in the encoded subspace.

The map $\rho' \mapsto \cE(\1)\cN'(\rho')\cE(\1)$ does map within the encoded subspace, and has the same corresponding error $\cN$ on the original system. Indeed, it is the map that appears in \cref{errorsEstate}.
For this error map we can therefore apply \cref{localmeas,time-evolution} as desired.
We will make an extra physically-motivated assumption on the form of the error map $\mathcal{N'}$, which guarantees that the difference between this map and $\cN'$ is negligible.

Let $H'$ be a $(\Delta,\eta,\epsilon)$-simulation of $H$ with corresponding local encoding~$\mathcal{E}$.
We might reasonably assume that errors that take the state out of the low-energy space of $H'$ are unlikely due to the high energy required for such an error.
We can formalise this by considering only noise operations $\mathcal{N}'$ such that $\tr[P_{\le \Delta(H')} \mathcal{N}'(\sigma)] \ge 1-\delta$ for any state $\sigma$ supported only on $S_{\le \Delta(H')}$, and some~$\delta$.

\begin{proposition}
  \label{error_resistance_cor}
  Let $H'$ be a $(\Delta,\eta,\epsilon)$-simulation of $H$ with corresponding local encoding $\mathcal{E}$.
  Let $\mathcal{N'}$ be a quantum channel acting on the simulator system and let $\rho'$ be a state in the encoded subspace, so that $\cE(\1)\rho'=\rho'$.

  Then, if $\tr[P_{\le \Delta(H')} \mathcal{N}'(\sigma)] \ge 1-\delta$ for all states $\sigma$ supported only on $S_{\le \Delta(H')}$,
  \begin{equation}
    \norm{\cN'(\rho')-\cE(\1)\cN'(\rho')\cE(\1)}_1\leq \sqrt{\delta(4-3\delta)}+8\eta.
  \end{equation}
\end{proposition}

\begin{proof}
  For readability, write $P_{\le \Delta} := P_{\le \Delta(H')}$. Then three applications of the triangle inequality give
  \begin{align}
    \| &\cN'(\rho')-\cE(\1)\cN'(\rho')\cE(\1) \|_1 \\
    &\leq \|\cE(\1)\cN'(\rho')\cE(\1)-\cE(\1)\cN'(\sigma)\cE(\1)\|_1 \\
    & \mspace{20mu} +\|\cE(\1)\cN'(\sigma)\cE(\1)-P_{\le \Delta}\cN'(\sigma)P_{\le \Delta}\|_1 \\
    & \mspace{20mu} +\|P_{\le \Delta}\cN'(\sigma)P_{\le  \Delta}-\cN'(\sigma)\|_1+\|\cN'(\sigma)-\cN'(\rho')\|_1 \label{eq:error_resistance_bound}
  \end{align}
  where $\sigma=\widetilde{V}V^\dg\rho' V\widetilde{V}^\dg$.
  Since $\cN'$ is a quantum channel and $\cE(\1)$ is a projector, the first and fourth terms are both bounded by
  \begin{equation}
    \|\rho'-\sigma\|_1=\|\rho'-\widetilde{V}V^\dg\rho' V\widetilde{V}^\dg\|_1
    \leq 2\|VV^\dg-\widetilde{V}V^\dg\|
    \leq 2 \eta,
  \end{equation}
  where we have used \cref{lem:conjdiff}.
  Similarly, we can bound the second term using \cref{lem:conjdiff} twice:
  \begin{equation}
    \|\cE(\1)\cN'(\sigma)\cE(\1)-P_{\le \Delta}\cN'(\sigma) P_{\le \Delta}\|_1
    \leq 2\|\cE(\1)-P_{\le \Delta}\|
    \leq 4 \eta.
  \end{equation}

  It remains to bound the third term $\|P_{\le \Delta}\cN'(\sigma) P_{\le \Delta}-\cN'(\sigma)\|$ in terms of $\delta$ using the condition assumed in the \lcnamecref{error_resistance_cor}.
  Given any state $\ket{\psi}$ such that $P_{\le\Delta}\ket{\psi} \neq \ket{\psi}$, define the orthonormal states $\ket{\phi_0}=P_{\le \Delta}\ket{\psi}/\sqrt{1-x}$ and $\ket{\phi_1}=(\1-P_{\le \Delta})\ket{\psi}/\sqrt{x}$ where $x=1-\bra{\psi}P_{\le \Delta}\ket{\psi}$.
  The operator $\proj{\psi}-P_{\le \Delta}\proj{\psi}P_{\le \Delta}$ is a rank 2 operator which acts non-trivially only on the space spanned by $\{\ket{\phi_0},\ket{\phi_1}\}$ as the following matrix:
  \begin{equation}
    \begin{pmatrix} 0 & \sqrt{x(1-x)} \\ \sqrt{x(1-x)} & x \end{pmatrix} \quad \text{ with eigenvalues } \lambda_{\pm}=\frac{x}{2} \pm \sqrt{x(1-x)+\frac{x^2}{4}}.
  \end{equation}
  Therefore $\norm{\proj{\psi}-P_{\le \Delta}\proj{\psi}P_{\le \Delta}}_1=|\lambda_+|+|\lambda_-|=\sqrt{x(4-3x)}$.
  This equality also holds trivially in the case $P_{\le\Delta}\ket{\psi} = \ket{\psi}$.

  Using the spectral decomposition, we can write $\cN'(\sigma)=\sum_j \lambda_j\proj{\psi_j}$ and use the triangle inequality to show that the third term in \cref{eq:error_resistance_bound} is bounded by
  \begin{align}
    \|&\cN'(\sigma) - P_{\le \Delta}\cN'(\sigma) P_{\le \Delta}\|_1 \\
    &\leq \sum_j \lambda_j \norm{\proj{\psi_j}-P_{\le \Delta}\proj{\psi_j}P_{\le \Delta}}_1 \\
    &= \sum_j \lambda_j\sqrt{x_j(4-3x_j)}=\sum_j \sqrt{\lambda_j x_j}\sqrt{\lambda_j(4-3x_j)}\\
    &\leq \sqrt{\left(\sum_j \lambda_j x_j\right)\left(4-3\sum_k \lambda_k x_k)\right)}
  \end{align}
  where $x_j=1-\bra{\psi_j}P_{\le \Delta}\ket{\psi_j}$ and we have used the Cauchy-Schwarz inequality in the last step.
  The result follows from $\sum_j\lambda_j x_j=1-\tr(P_{\le \Delta} \cN'(\sigma))=\delta$.
\end{proof}

By setting $\rho'=\cEs(\rho)$ in  \cref{error_resistance_cor}, and using \cref{errorsEstate}, we have
\begin{corollary}
  Let $H'$ be a $(\Delta,\eta,\epsilon)$-simulation of $H$ with corresponding local encoding $\cE(M)=V(M\ox P+\bar{M}\ox Q)V^\dg$ such that $\operatorname{rank}(P)=1$.
 Let $\cEs(\rho)=V(\rho\ox P)V^\dg$ and let $\mathcal{N'}$ be a quantum channel whose Kraus operators act on at most $l<n$ qudits of the simulator system.

  Then, if $\tr[P_{\le \Delta(H')} \mathcal{N}'(\cEs(\rho))] \ge 1-\delta$, there exists a CP-map $\cN$ whose Kraus operators act on at most $l$ qudits of the original system such that
  \begin{equation}
    \norm{\cN'(\cEs(\rho))-\cEs(\cN(\rho)) }_1\leq \sqrt{\delta(4-3\delta)}+8\eta.
  \end{equation}
\end{corollary}

\section{Universal Hamiltonian simulation}

Having drawn some consequences from the notion of simulation, we will now move on to prove that certain types of Hamiltonians are universal simulators, first introducing the key technique we use: perturbative reductions~\cite{Kempe-Kitaev-Regev,Bravyi-Hastings, Oliveira-Terhal,Bravyi-DiVincenzo-Loss}.

\subsection{Techniques}
\label{sec:perturbative}

Let $\mathcal{H}_{\text{sim}}$ be a Hilbert space decomposed as $\mathcal{H}_{\text{sim}} = \mathcal{H}_+ \oplus \mathcal{H}_-$, and let $\Pi_{\pm}$ denote the projector onto $\mathcal{H}_{\pm}$.
For any linear operator $O$ on $\mathcal{H}_{\text{sim}}$, write
\begin{equation}
  O_{--} = \Pi_- O \Pi_-,\;\;\;\; O_{-+}
  = \Pi_- O \Pi_+,\;\;\;\; O_{+-}
  = \Pi_+ O \Pi_-,\;\;\;\; O_{++} = \Pi_+ O \Pi_+.
\end{equation}
Let $H_0$ be a Hamiltonian such that $H_0$ is block-diagonal with respect to the split $\mathcal{H}_+ \oplus \mathcal{H}_-$, $(H_0)_{--} = 0$, and $\lambda_{\min}((H_0)_{++}) \ge 1$.

Slight variants of the following lemmas were shown in~\cite{Bravyi-Hastings}, building on previous work~\cite{Oliveira-Terhal,Bravyi-DiVincenzo-Loss}:

\begin{lemma}[First-order simulation~\cite{Bravyi-Hastings}]
  \label{lem:firstorder}
  Let $H_0$ and $H_1$ be Hamiltonians acting on the same space. Suppose there exists a local isometry $V$ such that $\operatorname{Im}(V)=\cH_-$ and
  \begin{equation}
    \| V H_{\operatorname{target}} V^\dg - (H_1)_{--} \| \le \epsilon/2.
  \end{equation}
  Then $H_{\operatorname{sim}} = \Delta H_0 + H_1$ $(\Delta/2,\eta,\epsilon)$-simulates $H_{\operatorname{target}}$, provided that the bound $\Delta \ge O(\|H_1\|^2/\epsilon + \|H_1\| / \eta)$ holds.
\end{lemma}

\begin{lemma}[Second-order simulation~\cite{Bravyi-Hastings}]
  \label{lem:secondorder}
  Let $H_0$, $H_1$, $H_2$ be Hamiltonians acting on the same space, such that: $\max\{\|H_1\|,\|H_2\|\} \le \Lambda$; $H_1$ is block-diagonal with respect to the split $\mathcal{H}_+ \oplus \mathcal{H}_-$; and $(H_2)_{--} =0$.
  Suppose there exists a local isometry $V$ such that $\operatorname{Im}(V)=\cH_-$ and
  \begin{equation}
    \| V H_{\operatorname{target}} V^\dg - (H_1)_{--} + (H_2)_{-+} H_0^{-1} (H_2)_{+-} \| \le \epsilon/2.
  \end{equation}
  Then $H_{\operatorname{sim}} = \Delta H_0 + \Delta^{1/2} H_2 + H_1$ $(\Delta/2,\eta,\epsilon)$-simulates $H_{\operatorname{target}}$, provided that $\Delta \ge O(\Lambda^6/\epsilon^2 + \Lambda^2/\eta^2)$.
\end{lemma}

\begin{lemma}[Third-order simulation~\cite{Bravyi-Hastings}]
  \label{lem:thirdorder}
  Let $H_0$, $H_1$, $H_1'$, $H_2$ be Hamiltonians acting on the same space, such that: $\max\{\|H_1\|,\|H_1'\|,\|H_2\|\} \le \Lambda$; $H_1$ and $H_1'$ are block-diagonal with respect to the split $\mathcal{H}_+ \oplus \mathcal{H}_-$; $(H_2)_{--}=0$.
  Suppose there exists a local isometry $V$ such that $\operatorname{Im}(V)=\cH_-$ and
  \begin{equation}
    \| V H_{\operatorname{target}} V^\dg - (H_1)_{--} - (H_2)_{-+} H_0^{-1} (H_2)_{++} H_0^{-1} (H_2)_{+-} \| \le \epsilon/2
  \end{equation}
  and also that
  \begin{equation}
    (H_1')_{--} = (H_2)_{-+} H_0^{-1} (H_2)_{+-}.
  \end{equation}
  Then $H_{\operatorname{sim}} = \Delta H_0 + \Delta^{2/3} H_2 + \Delta^{1/3} H_1' + H_1$ $(\Delta/2, \eta, \epsilon)$-simulates $H_{\operatorname{target}}$, provided that $\Delta \ge O(\Lambda^{12}/\epsilon^3 + \Lambda^3/\eta^3)$.
\end{lemma}

In fact, whenever we use \cref{lem:firstorder,lem:secondorder} we will be able to replace the approximate equalities up to $\epsilon/2$ with exact equalities.
We do not invoke \cref{lem:thirdorder} explicitly in this work; however, we state it for completeness because it can be used to show that a QMA-completeness result of~\cite{Oliveira-Terhal} (\cref{thm:pauliterms} below) actually implies a simulation result.
The scaling of $\Delta$ assumed in these lemmas is sufficient to ensure that $\Delta/2$ separates the high- and low-energy parts of the simulator Hamiltonian $H_{\operatorname{sim}}$.
The main difference between these lemmas and their equivalents in~\cite{Bravyi-Hastings} is that here we insist on locality of the isometry $V$, corresponding to our local notion of simulation.
The correctness proofs of~\cite{Bravyi-Hastings} go through without change.

We remark that, in order to use the above lemmas, it will often be convenient to add a multiple of the identity to the simulator or target Hamiltonians, corresponding to an overall energy shift.
The families of Hamiltonians which we consider will always contain the identity, so we are free to do this with impunity.
For readability, we often omit this implicit freely added identity term when we state the form of restricted types of Hamiltonians below.

In the Hamiltonian complexity literature, many constructions, known as ``gadgets'', have been developed to prove that special cases of the {\sc Local Hamiltonian} problem\footnote{The problem of computing the ground-state energy of a $k$-local Hamiltonian on $n$ qubits, up to $1/\poly(n)$ precision~\cite{Kitaev-Shen-Vyalyi,Kempe-Kitaev-Regev}.} are QMA-complete, by reducing more complex cases to the more specialised cases (e.g.~\cite{Kempe-Kitaev-Regev,Oliveira-Terhal,Cubitt-Montanaro,Piddock-Montanaro}).
These reductions often use perturbation theory and can be interpreted as instances of \cref{lem:firstorder} or \cref{lem:secondorder}.
Thus, rather than being merely \emph{reductions}, they are \emph{simulations} in our terminology.
Two types of gadget are commonly used:

\begin{itemize}
\item \textbf{Mediator qubits.}
  Imagine we have two qubits $a$ and $b$ and would like to implement some effective interaction across them.
  One way to achieve this is to attach an ancilla, ``mediator'' qubit $c$, and apply a heavily-weighted local term $H_0$ to $c$, and a less heavily-weighted term $H_2 = H_{ac} + H_{bc}$.
  If we insist that qubit $c$ is in the ground state of $H_0$, this produces an effective interaction across qubits $a$ and $b$, together with some additional local terms on $a$ and $b$ which we can cancel out by adding an extra term $H_1$.
  This puts us in the setting of \cref{lem:secondorder}.
  The isometry $V$ is the map which acts as the identity on qubits $a$ and $b$, and attaches a qubit $c$ in the ground state of $H_0$.
  This type of gadget is used in~\cite{Oliveira-Terhal,Cubitt-Montanaro,Schuch-Verstraete} and elsewhere in the literature.
  Whenever such gadgets are used and analysed using second-order perturbation theory, the preconditions of \cref{lem:secondorder} hold, so we obtain that the physical Hamiltonian constructed simulates the desired logical Hamiltonian.

\item \textbf{Subspace encodings.}
  This type of gadget encodes a logical qubit within $\ell = O(1)$ physical qubits.
  A Hamiltonian $H$ on $\ell$ qubits is chosen whose ground space is 2\nbd-dimensional.
  Then an overall Hamiltonian is produced using a sum of heavily-weighted $H$ terms, one on each $\ell$-tuple of physical qubits.
  Within the ground space of the whole Hamiltonian, each $\ell$-tuple corresponds to a qubit.
  Less heavily-weighted interactions across $\ell$-tuples produce interactions across logical qubits.
  \Cref{lem:firstorder} and \cref{lem:secondorder} can be used to show that the simulator Hamiltonian does indeed simulate the target Hamiltonian.
  Now the isometry $V$ is a tensor product of $n$ isometries, each of which maps a qubit to the ground space of $H$ within the space of $\ell$ qubits.
  By choosing the right isometry, corresponding to a choice of basis for this ground space, we obtain desired new interactions across logical qubits.

  This type of gadget is used in~\cite{Cubitt-Montanaro}.
  However, note that two of the reductions in that work (simulating an arbitrary 2\nbd-local qubit Hamiltonian with a Hamiltonian made up of interactions of Heisenberg or XY type) were more complicated.
  In these reductions $H$ acts on 3 qubits and has a 4-dimensional ground space, corresponding to two logical qubits.
  Then additional heavily weighted terms are used to effectively project one qubit in each logical pair into a fixed, and highly entangled, state of $n$ qubits.
  This technique would not comply with our notion of simulation, as the state attached by the corresponding isometry $V$ would be far from a product state.
  Here we no longer need to use this type of reduction as we have a genuinely local simulation (\cref{thm:heisenberg} below).
\end{itemize}

In this work we will use both of these kinds of simulation.
For readability, we will not fully repeat the correctness proofs of the simulations from previous work, instead sketching the arguments and deferring to the original papers for technical details.
However, we stress that replacing the analysis of these gadgets in previous work with the use of \cref{lem:firstorder,lem:secondorder} is sufficient to obtain fully rigorous proofs of correctness.

In addition, to gain some intuition, we now describe more formally how one of the simpler gadgets from~\cite{Oliveira-Terhal} can be analysed using \cref{lem:secondorder}, and verify that it fits the constraints.
The gadget, which is called the subdivision gadget and is an example of a mediator qubit gadget, allows a $k$\nbd-local Hamiltonian to be simulated by a $(\lceil k/2 \rceil + 1)$\nbd-local Hamiltonian.
Consider an interaction of the form $H_{\operatorname{target}} = A_a B_b$, where $A$ acts on a subset of qubits $a$, and $B$ acts on a disjoint subset of qubits $b$.
A mediator qubit $c$ is introduced and we take Hamiltonians
\begin{equation}
   H_0 = \proj{1}_c,\;\;\;\; H_2 = \frac{1}{\sqrt{2}}(A_a X_c - X_c B_b).
\end{equation}
Then $(H_2)_{-+} = (H_2)_{+-}^\dg = \frac{1}{\sqrt{2}}\ket{0}\!\bra{1}_c(A_a - B_b)$, so
\begin{equation}
  (H_2)_{-+} H_0^{-1} (H_2)_{+-}
  = \frac{1}{2} \proj{0}_c (A_a - B_b)^2 = \proj{0}_c (\frac{1}{2} A_a^2 - A_a B_b + \frac{1}{2} B_b^2).
\end{equation}
In addition, $(H_2)_{--} = 0$.
We choose $H_1 = \frac{1}{2}(A_a^2 + B_b^2)$, so $(H_1)_{--} = \frac{1}{2}\proj{0}_c (A_a^2 + B_b^2)$.
Consider the isometry defined by $V\ket{\psi}_{ab} = \ket{\psi}_{ab} \ket{0}_c$.
Then it is easy to verify that
\begin{equation}
   V H_{\operatorname{target}} V^\dg = (H_1)_{--} - (H_2)_{-+} H_0^{-1} (H_2)_{+-}.
\end{equation}
It follows from \cref{lem:secondorder} that, for sufficiently high $\Delta$, $H_{\operatorname{sim}} = \Delta H_0 + \sqrt{\Delta} H_2 + H_1$ $(\Delta,\eta,\epsilon)$-simulates $H_{\operatorname{target}}$.
Observe that $H_{\operatorname{sim}}$ contains interactions on only at most $\max \{|a|+1,|b|+1\}$ qubits.
This idea can be used to reduce the locality of the whole Hamiltonian simultaneously, by writing each $k$-local interaction term in the original Hamiltonian as a sum of tensor product interactions, and adding a new mediator qubit for each such interaction to simulate it with a $(\lceil k/2 \rceil + 1)$\nbd-local interaction.
The corresponding isometry simply attaches a state of $\poly(n)$ qubits, each in the state $\ket{0}$, so is local.

Since each term of $H_2$ acts on at most one mediator qubit, there is no interference between gadgets and the total effective Hamiltonian is simply the sum of the effective interactions of each gadget.
We say that the gadgets are applied \emph{in parallel}.
For a detailed discussion of the parallel application of mediator qubit gadgets, see~\cite{Piddock-Montanaro}.
We formalise this discussion in the following lemma, with the addition of a corresponding result for subspace encoding gadgets.

\begin{lemma}
  \label{lem:interference}
  Let the Hamiltonian $H_0=\sum_i H_0^{(i)}$ be a sum of terms $H_0^{(i)}$ each with ground space energy $0$ and acting non-trivially only on disjoint subsets of qudits $S_i$.
  Let the ground space projection operator for $H_0^{(i)}$ be given by $P_-^{(i)}$ so the overall ground space projection operator for $H_0$ is given by $P_-=\prod_{i} P_-^{(i)}$.

  \begin{itemize}
  \item
    If $H_1$ can be expressed as a sum of terms, $H_1=\sum_\alpha H_1^{(\alpha)}$, the first order perturbation satisfies
    \begin{equation}
      (H_1)_{--}=\sum_{\alpha}(H_1^{(\alpha)})_{--}.
    \end{equation}

  \item {\textbf{Mediator gadgets.}}
    Let $H_2=\sum_{i} H_2^{(i)}$, where $H_2^{(i)}$ acts trivially on all qudits in $\cup_{j\neq i}S_j$.
    Suppose that all first order terms vanish, i.e.
    $P_-^{(i)}H_2^{(i)}P_-^{(i)}=0$ for all $i$.
    Then the second and third order terms are given by
    \begin{align}
      - (H_2 H_0^{-1}H_2)_{--} &= -\sum_{i}P_-H_2^{(i)}(H_0^{(i)})^{-1}H_2^{(i)}P_-\\
      - (H_2 H_0^{-1}H_2 H_0^{-1}H_2)_{--} &= -\sum_{i}P_-H_2^{(i)}(H_0^{(i)})^{-1}H_2^{(i)}(H_0^{(i)})^{-1}H_2^{(i)}P_-.
    \end{align}

  \item \textbf{Subspace gadgets.}
    Let $H_2=\sum_{(i,j)} H_2^{(i,j)}$ for ordered pairs $(i,j)$, where $H_2^{(i,j)}$ acts non-trivially only on $S_i$ and $S_j$ and raises both sets of qudits completely out of their ground spaces such that $P_-^{(i)}H_2^{(i,j)}P_-=0$ and $ P_-^{(j)}H_2^{(i,j)} P_-=0$.
    Then the second order perturbation is given by
    \begin{equation}
      - (H_2 H_0^{-1}H_2)_{--}= -\sum_{ (i,j)}P_-H_2^{(i,j)}\left(H_0^{(i)}+H_0^{(j)}\right)^{-1}H_2^{(i,j)}P_-.
    \end{equation}
  \end{itemize}
\end{lemma}

Before providing a proof, we remark why different results are needed for mediator and subspace gadgets.
In the mediator gadget case, the qudits of $S_i$ are in a one dimensional ground space of $H_0^{(i)}$ for all $i$, and the effective Hamiltonian acts non-trivially on the remaining qudits in $\overline{\cup_j S_j}$.
Therefore interesting interactions can be effected, even when each perturbative term acts on only one of the sets $S_i$.
Whereas for subspace gadgets, the $i$th logical qudit lives in the groundspace of $H_0^{(i)}$ on the physical qudits $S_i$, so we need perturbative terms to act between different $S_i$ in order to make 2\nbd-local interactions.

\begin{proof}
  The first claim is trivial.
  For mediator qudit gadgets, define a projection operator $P_-^{\overline{(i)}}=\prod_{k\neq i}P_-^{(k)}$ and note that it acts trivially on $S_i$, so commutes with $H_2^{(i)}$.
  Since the ground state energy for each $H_0^{(k)}$ is zero, $H_0^{(k)}P_-^{(k)}=0$ and so $H_0^{(k)}P_-^{\overline{(i)}}=0$ for all $k\neq i$.
  Therefore,
  \begin{equation}
    (H_0)^{-1}P_-^{\overline{(i)}}=\left(\sum_k H_0^{(k)}\right)^{-1}P_-^{\overline{(i)}}=(H_0^{(i)})^{-1}P_-^{\overline{(i)}}.
  \end{equation}
  Since $P_-= P_-^{\overline{(i)}}P_-$, the second order term is given by
  \begin{align}
    - (H_2 H_0^{-1}H_2)_{--}
    &= -\sum_{i}P_-H_2 H_0^{-1}H_2^{(i)}P_- \\
    &= -\sum_{ i}P_-H_2 H_0^{-1}P_-^{\overline{(i)}}H_2^{(i)}P_- \\
    &= -\sum_{ i}P_-H_2 P_-^{\overline{(i)}}(H_0^{(i)})^{-1}H_2^{(i)}P_- \\
    &= -\sum_{ i}P_-H_2^{(i)}(H_0^{(i)})^{-1}H_2^{(i)}P_-
  \end{align}
  where the final equality holds because $P_-^{(j)}H_2^{(k)}P_-^{(j)}=0$ for all $j\neq k$, and so $P_-H_2P_-^{\overline{(i)}}=P_-H_2^{(i)}$.

  Using the same techniques, we can show that the third order term is equal to
  \begin{equation}
    -\sum_{i,j,k}P_-H_2^{(j)}(H_0^{(j)})^{-1} P_-^{\overline{(j)}}H_2^{(k)} P_-^{\overline{(i)}}(H_0^{(i)})^{-1}H_2^{(i)}P_-.
  \end{equation}
  If $k\neq i,j$, then $P_-^{(k)}$ appears in the product expression for both $P_-^{\overline{(i)}}$ and $P_-^{\overline{(j)}}$ and so $ P_-^{\overline{(j)}} H_2^{(k)} P_-^{\overline{(i)}}=0$.
  We may therefore assume $k=j$ (the proof for $k=i$ proceeds analogously), in which case we have
  \begin{equation}
    -\sum_{i,j}P_-H_2^{(j)}(H_0^{(j)})^{-1} P_-^{\overline{(j)}}H_2^{(j)} P_-^{\overline{(i)}}(H_0^{(i)})^{-1}H_2^{(i)}P_-.
  \end{equation}
  The operator $P_-^{\overline{(j)}}$ commutes with $H_2^{(j)}$ and $P_-^{\overline{(i)}}$, and so, remembering that $P_-^{\overline{(j)}}(H_0^{(i)})^{-1}=0$ for $i\neq j$, we must have $i=j$, giving the desired result.

  The proof is very similar for subspace gadgets, but we instead define a projection operator $P_-^{\overline{(i,j)}}=\prod_{k\neq i,j}P_-^{(k)}$ for ordered pairs $(i,j)$, noting that it acts trivially on $S_i$ and $S_j$, so commutes with $H_2^{(i,j)}$.
  As before, we have $H_0^{(k)}P_-^{\overline{(i,j)}}=0$ for all $k\neq i,j$, so $(H_0)^{-1}P_-^{\overline{(i,j)}}=(H_0^{(i)}+H_0^{(j)})^{-1}P_-^{\overline{(i,j)}}$.
  Therefore the second order term is given by
  \begin{align}
    - (H_2 H_0^{-1}H_2)_{--}
    &= -\sum_{ (i,j)}P_-H_2 H_0^{-1}H_2^{(i,j)}P_- \\
    &= -\sum_{ (i,j)}P_-H_2 H_0^{-1}P_-^{\overline{(i,j)}}H_2^{(i,j)}P_- \\
    &= -\sum_{ (i,j)}P_-H_2 P_-^{\overline{(i,j)}}\left(H_0^{(i)}+H_0^{(j)}\right)^{-1}H_2^{(i,j)}P_- \\
    &= -\sum_{ (i,j)}P_-H_2^{(i,j)}\left(H_0^{(i)}+H_0^{(j)}\right)^{-1}H_2^{(i,j)}P_-
  \end{align}
  where the final equality holds since by the form of $H_2^{(i,j)}$ assumed in the lemma, $P_-H_2^{(i',j')}P_-^{\overline{(i,j)}}=0$ unless $(i',j')=(i,j)$, so $P_-H_2 P_-^{\overline{(i,j)}}=P_-H_2^{(i,j)} P_-^{\overline{(i,j)}}$.
\end{proof}

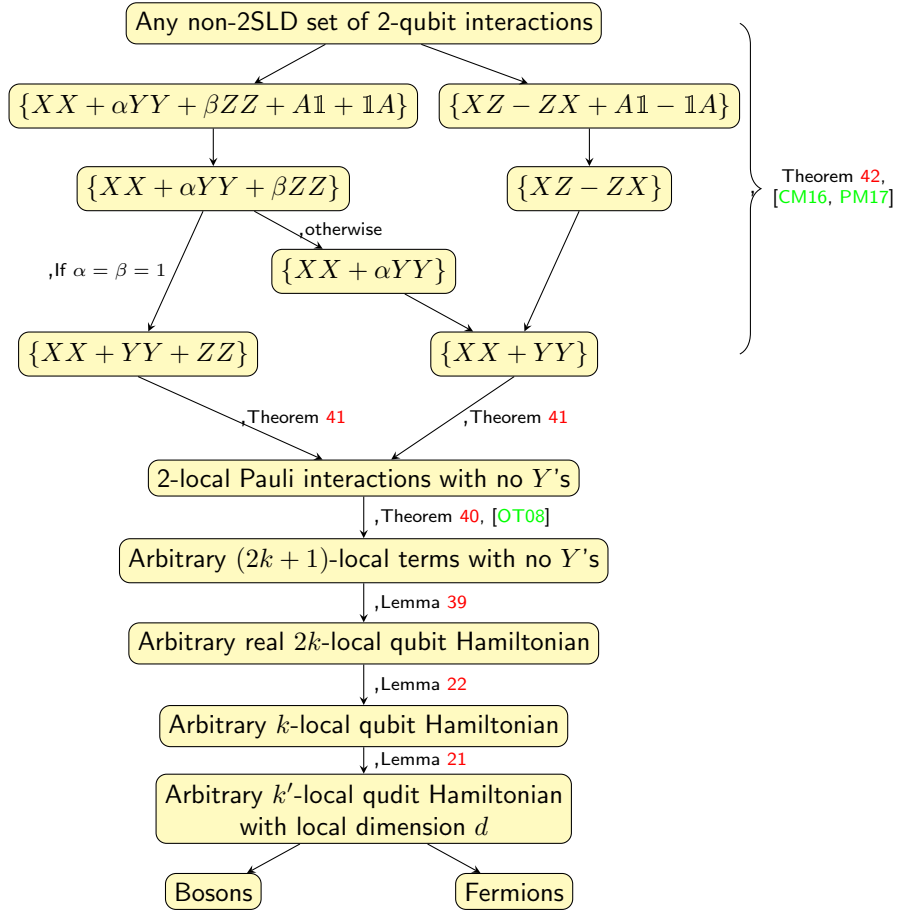
\begin{figure}[p]
  \begin{center}
    \begin{tikzpicture}[yscale=1.1,font=\sffamily,every node/.style={draw,fill=yellow!30,rounded corners},label/.style={draw=none,fill=none,font={\sffamily,\scriptsize}},>=stealth]
      \node (arb) at (0,4) {Any non-2SLD set of 2\nbd-qubit interactions};
      \node (arbsym) at (-2,3) {$\{XX + \alpha YY + \beta ZZ + A\1 + \1 A\}$};
      \node (arbanti) at (3,3) {$\{XZ - ZX + A\1 - \1 A\}$};
      \node (xxayybzz) at (-2,2) {$\{XX + \alpha YY + \beta ZZ\}$};
      \node (xxayy) at (0,1) {$\{XX + \alpha YY\}$};
      \node (heisenberg) at (-3,0) {$\{XX+YY+ZZ\}$};
      \node (xy) at (2,0) {$\{XX+YY\}$};
      \node (xxzz) at (0,-1.5) {2\nbd-local Pauli interactions with no $Y$'s};
      \node (antisym) at (3,2) {$\{XZ-ZX\}$};
      \node (4local) at (0,-2.5) {Arbitrary $(2k+1)$\nbd-local terms with no $Y$'s};
      \node (real) at (0,-3.5) {Arbitrary real $2k$\nbd-local qubit Hamiltonian};
      \node (klocal) at (0,-4.5) {Arbitrary $k$\nbd-local qubit Hamiltonian};
      \node[align=center] (klocald) at (0,-5.5) {Arbitrary $k'$\nbd-local qudit Hamiltonian\\ with local dimension $d$};
      \node (bosons) at (-2,-6.5) {Bosons};
      \node (fermions) at (2,-6.5) {Fermions};
      \draw[->] (arb) -- (arbsym);
      \draw[->] (arb) --  (arbanti);
      \draw[->] (arbsym) --  (xxayybzz);
      \draw[->] (arbanti) --  (antisym);
      \draw[->] (xxayybzz) -- node[label,right] {otherwise} (xxayy);
      \draw[->] (xxayy) --   (xy);
      \draw[->] (heisenberg.south) -- node[label,right] {\cref{thm:heisenberg}} (xxzz);
      \draw[->] (xy.south) -- node[label,right] {\cref{thm:heisenberg}} (xxzz);
      \draw[->] (xxzz) -- node[label,right] {\cref{thm:pauliterms}, \cite{Oliveira-Terhal}} (4local);
      \draw[->] (xxayybzz) -- node[label,left] {If $\alpha=\beta=1$} (heisenberg);
      \draw[->] (antisym) -- (xy);
      \draw[->] (4local) -- node[label,right] {\cref{lem:YYfromXZ}} (real);
      \draw[->] (real) -- node[label,right] {\cref{complex-to-real sim}} (klocal);
      \draw [decorate,decoration={brace,amplitude=8pt}] (5,4) -- node[label,right]
        {\begin{tabular}{c}
         \cref{thm:reductions},\\ \cite{Cubitt-Montanaro,Piddock-Montanaro}\\\end{tabular}} (5,0);
     \draw[->] (klocal) -- node[label,right] {\cref{prop:qudits}} (klocald);
     \draw[->] (klocald) -- (fermions);
     \draw[->] (klocald) -- (bosons);
   \end{tikzpicture}
 \end{center}
 \caption[Sequence of simulations used in this work.]{%
   Sequence of simulations used in this work.
   An arrow from one box to another indicates that a Hamiltonian of the first type can simulate a Hamiltonian of the second type.
   Where two arrows leave a box, this indicates that a Hamiltonian of this type can simulate one of the two target Hamiltonians, but not necessarily both.
   ``2SLD'' is short for ``the 2\nbd-local parts of all interactions in the set are simultaneously locally diagonalisable'', and $k,k' \ge 2$ are arbitrary integers such that $k \ge \lceil k' \log_2 d\rceil$.}
 \label{fig:reductions}
\end{figure}


\subsection{Universal simulators}

We are now ready to prove universality of a variety of classes of Hamiltonians.
The overall structure of the argument is illustrated in \cref{fig:reductions}; the eventual result is that each of the classes of qudit Hamiltonians illustrated in the diagram is universal.
For brevity, when we state and prove simulation results, rather than writing ``The family of A-Hamiltonians can simulate the family of B-Hamiltonians'' for some A and B, we simply write ``A-Hamiltonians can simulate B-Hamiltonians''.
We stress that such a statement is nevertheless rigorous and should be understood in the sense of \cref{dfn:sim}.

We have already proven some of the simulation results required (Lemmas \ref{prop:qudits} and \ref{complex-to-real sim}). We now complete the programme of \cref{fig:reductions} by showing that every remaining type of qudit Hamiltonian in the diagram is universal.
The simulation of \cref{complex-to-real sim} may produce terms which include even numbers of $Y$ components.
First we show that such terms are not necessary.
Note that it was already known that Hamiltonians without any $Y$ components can be QMA-complete~\cite{Biamonte-Love}; what we show here is that such Hamiltonians can in fact be universal simulators.

\begin{lemma}
  \label{lem:YYfromXZ}
  Real $k$-local qubit Hamiltonians can be simulated by real $(k+1)$-local qubit Hamiltonians whose Pauli decomposition does not contain any $Y$ terms.
\end{lemma}

\begin{proof}
  Let $H$ be a real $k$-local qubit Hamiltonian.
  For each $k'$-local interaction $h$ in the Pauli decomposition of $H$ ($k' \le k$), add an additional mediator qubit $a$.
  Since $H$ is real, there must be an even number of $Y$ terms in $h$.
  We may assume, by reordering qubits if necessary, that $h=Y^{\ox 2m} \ox A$ where $A$ is a tensor product of $X$ and $Z$ terms on $k'-2m$ qubits.

  We use second-order perturbation theory (\cref{lem:firstorder}) to effectively generate $h$ from an interaction containing only $X$ and $Z$ terms.
  Consider a heavy interaction term $H_0$ acting only on the mediator qubit, $H_0=(\1+Z_a)/2=\ket{0}\bra{0}_a$, with groundstate $\ket{1}_a$, and a perturbative term $H_2=X_a(X^{\ox 2m} \ox \1+ (-1)^{m+1} Z^{\ox 2m} \ox A)$.
  $H_2$ acts as a switch between the ground space and the excited space.
  It is clear that the first-order term $\Pi_{-}H_2\Pi_-$ vanishes.
  The second-order term is, up to a multiple of the identity, of the desired form:
  \begin{align}
    -\Pi_{-}H_2(H_0^{-1})_{++}H_2\Pi_-
    &= -\ket{0}\bra{0}_a (X^{\ox 2m} \ox \1+ (-1)^{m+1}Z^{\ox 2m} \ox A)^2\\
    &=2\ket{0}\bra{0}_a \left(Y^{\ox 2m} \ox A+\1\right).
  \end{align}
  It follows from \cref{lem:firstorder} that, for sufficiently large $\Delta$, $H' = \Delta H_0 + \Delta^{1/2} H_2$ is a $(\Delta,\eta,\epsilon)$-simulation of the interaction $h$.
  This can be used to simulate the whole Hamiltonian $H$ by applying separate mediator qubit gadgets for each term $h$ in parallel; by \cref{lem:interference}, different terms do not interfere with each other.
\end{proof}

It may be tempting to think that a similar second-order mediator qubit gadget could be used to simulate a 1-local $Y$ interaction, since $ZX=iY$.
However the same trick would not work if we took $H_2=X_a(X_1+Z_1)$, for example, because the anticommutator $\{X,Z\}$ vanishes and so $(X+Z)^2=2\1$.
Of course, this should not be surprising, as the perturbative expansion of any real Hamiltonian can only result in real Hamiltonian terms.

Next we use a result of Oliveira and Terhal~\cite{Oliveira-Terhal} to further specialise the class of Hamiltonians proven universal in \cref{lem:YYfromXZ}.

\begin{theorem}[essentially~\cite{Oliveira-Terhal}]
  \label{thm:pauliterms}
  $k$-local qubit Hamiltonians whose Pauli decomposition does not contain any $Y$ terms can be simulated by 2\nbd-local Hamiltonians of the form $\sum_{i>j}\alpha_{ij}A_{ij} +\sum_k (\beta_k X_k+\gamma_k Z_k)$, where $A_{ij}$ is one of the interactions $X_iX_j, X_iZ_j, Z_iX_j$ or $Z_iZ_j$ and $\alpha_{ij}, \beta_k, \gamma_k \in \R$.
\end{theorem}

We sketch the proof of \cref{thm:pauliterms}; see~\cite{Oliveira-Terhal} for more technical details.

\begin{proof}[sketch]
  The claim is trivial for $k \le 2$, so assume $k \ge 3$.
  We first note that, for each $k$-tuple of qubits, one can decompose any interaction across that $k$-tuple as a weighted sum of interactions which are each tensor products of Pauli matrices.
  These can be thought of as separate hyperedges in the hypergraph of interactions in $H$, and henceforth treated separately.

  Then, to effectively produce each of these Pauli interactions, the subdivision gadgets described in~\cite{Oliveira-Terhal} can be used.
  There are two of these gadgets.
  One gadget simulates an arbitrary $k$-wise interaction of the form $A \ox B$ across sets of qubits $a$ and $b$ by using a mediator qubit $c$, and $\lceil k/2 \rceil$-wise interactions of the form $A_a X_c + X_c B_b$.
  This gadget was discussed in detail near the start of \cref{sec:perturbative}.
  Repeated use of this procedure enables $k$-local interactions to be simulated using 3-local interactions.
  The second gadget simulates a 3-local Hamiltonian with a 2\nbd-local Hamiltonian.
  The gadget generates interactions of the form $A_a B_b C_c$ by introducing a mediator qubit $d$ and a Hamiltonian whose terms are proportional to $A_a X_d$, $B_b X_d$ and $C_c \proj{1}_d$, and using third-order perturbation theory to generate effective 3-local terms from these~\cite{Bravyi-Hastings,Oliveira-Terhal}.
  This leads to unwanted 2\nbd-local and 1-local terms being generated too, which can be effectively deleted using compensating terms of the form $XZ$, $X$, $Z$.
  By \cref{lem:interference}, these third order mediator qubit gadgets do not interfere.
  Note that the analysis of~\cite{Oliveira-Terhal} can be replaced with the use of \cref{lem:thirdorder} to show that this gadget indeed gives a simulation in our terminology.

  Finally, observe that these gadgets do not introduce any $Y$ terms if they were not present already.
\end{proof}

Next we show that the Heisenberg and XY interactions are sufficient to simulate any Hamiltonian of the form of \cref{thm:pauliterms}.
This is the most technically involved simulation in this paper.
A similar simulation result could be achieved using a reduction presented in \cite{Cubitt-Montanaro}, but this would require attaching a highly entangled state (the projector onto which corresponds to $P$ in \Cref{dfn:sim}). In the new simulation $P$ is trivial, implying that it is easy to implement the map $\cEs(\rho)$ defined in \cref{eq:estate}.

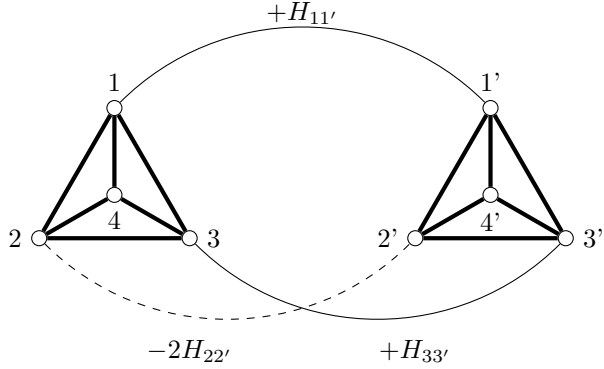
\begin{figure}
  \centering
  \begin{tikzpicture}
    \node[mqubit] (2) at (0,0) [label=left:2] {};
    \node[mqubit] (3) at (2,0) [label=right:3] {};
    \node[mqubit] (1) at (1,1.732) [label=above:1] {};
    \node[mqubit] (4) at (1,0.577) [label=below:4] {};

    \draw [heavy] (1) to (2);
    \draw [heavy] (2) to (3);
    \draw [heavy] (3) to (4);
    \draw [heavy] (4) to (1);
    \draw [heavy] (3) to (1);
    \draw [heavy] (2) to (4);

    \begin{scope}[shift={(5,0)}]
      \node[mqubit] (2a) at (0,0) [label=left:2'] {};
      \node[mqubit] (3a) at (2,0) [label=right:3'] {};
      \node[mqubit] (1a) at (1,1.732) [label=above:1'] {};
      \node[mqubit] (4a) at (1,0.577) [label=below:4'] {};

      \draw [heavy] (1a) to (2a);
      \draw [heavy] (2a) to (3a);
      \draw [heavy] (3a) to (4a);
      \draw [heavy] (4a) to (1a);
      \draw [heavy] (3a) to (1a);
      \draw [heavy] (2a) to (4a);
    \end{scope}

    \draw (1) to[out=45, in=135] (1a) [label=below:+1];
    \draw (2)[dashed] to[out=-45, in=-135] (2a);
    \draw (3) to[out=-45, in=-135] (3a);

    \node at (3.5,3) {$+H_{11'}$};
    \node at (5,-1.5) {$+H_{33'}$};
    \node at (2,-1.5) {$-2H_{22'}$};
  \end{tikzpicture}
  \caption[One logical qubit is encoded within a quadruple of physical qubits.]{%
    One logical qubit is encoded within a quadruple of physical qubits (1--4 and $1'$--$4'$).
    2-local interactions are implemented using interactions across the quadruples.
    The figure illustrates the Hamiltonian for simulating $X_LX_L$, up to 1-local terms.}
  \label{fig:comp4}
\end{figure}

\begin{theorem}
  \label{thm:heisenberg}
  Let $\mathcal{F}$ be the family of qubit Hamiltonians of the form $H=\sum_{i>j}\alpha_{ij}A_{ij} +\sum_k (\beta_k X_k+\gamma_k Z_k)$, where $A_{ij}$ is one of the interactions $X_iX_j$, $X_iZ_j$, $Z_iX_j$ or $Z_iZ_j$ and $\alpha_{ij}, \beta_k, \gamma_k \in \R$.
  Then $\{XX+YY+ZZ\}$-Hamiltonians and $\{XX+YY\}$-Hamiltonians can simulate $\mathcal{F}$.
\end{theorem}

\begin{proof}
  We prove the claim for the Heisenberg interaction $XX+YY+ZZ$; the argument is completely analogous for the XY interaction $XX+YY$.
  We use a subspace encoding gadget to encode a logical qubit in the ground space of the Hamiltonian of the complete graph on 4 qubits, as illustrated in \cref{fig:comp4}.

  The overall heavy interaction used is
  \begin{equation}
    H_0 := H_{12}+H_{23}+H_{34}+H_{14}+H_{24}+H_{13}+6\1,
  \end{equation}
  where we write $H_{ij}=X_iX_j+Y_iY_j+Z_iZ_j$.
  The identity term is present to ensure that the ground space of $H_0$ corresponds to eigenvalue zero.
  $H_0$ has a two dimensional ground space $S$ given in terms of singlet states $\ket{\Psi^-}$ by
  \begin{equation}
    S= \linspan\left\{\ket{\Psi^{-}}_{12}\ket{\Psi^{-}}_{34},\ket{\Psi^{-}}_{13}\ket{\Psi^{-}}_{24}\right\}
    \quad \text{where } \ket{\Psi^{-}}=\frac{\ket{01}-\ket{10}}{\sqrt{2}}.
  \end{equation}
  We choose the following orthonormal basis for our logical qubit:
  \begin{equation}
    \ket{0_L} = \ket{\Psi^-}_{13}\ket{\Psi^-}_{24} \qquad
    \ket{1_L} = \tfrac{2}{\sqrt{3}}\ket{\Psi^-}_{12}\ket{\Psi^-}_{34}
                - \tfrac{1}{\sqrt{3}}\ket{\Psi^-}_{13}\ket{\Psi^-}_{24}
  \end{equation}

  \paragraph{First-order perturbations}
  We can simulate 1-local interactions $X_L$ and $Z_L$ using first-order perturbation theory.
  By \cref{lem:firstorder}, given a perturbation term $H_1$, the first-order perturbation is given by $\Pi_{-}H_1\Pi_{-}$, where $\Pi_{-}$ is the projector into the ground space.
  Note that the ground space is defined in terms of singlet states which have the same form in any local basis, and so
  \begin{equation}
    \Pi_{-}X_iX_j\Pi_{-}=\Pi_{-}Y_iY_j\Pi_{-}=\Pi_{-}Z_iZ_j\Pi_{-}
  \end{equation}
  which we can also check explicitly.
  Although the heavy Hamiltonian $H_0$ is invariant under permutations of the physical qubits, this symmetry is lost when we fix the basis, and so $\Pi_{-}X_iX_j\Pi_{-}$ does depend on $(i,j)$ -- the values are given in \cref{tab:1local}.

  Therefore, we can simulate any real 1-local interaction up to an irrelevant identity term; by \cref{lem:firstorder}, choosing $H_1=\frac{\alpha}{\sqrt{3}} H_{14} + \frac{1}{2}(\frac{\alpha}{\sqrt{3}}-\beta) H_{13}$ will simulate the interaction $\Pi_-H_1\Pi_-=\alpha X_L +\beta Z_L+\frac{1}{2}(\beta-\sqrt{3}\alpha)\1$.

  \begin{table}
    \begin{equation}
      \begin{array}{|c|c|}
        \hline
        (i,j) & \Pi_{-}X_iX_j\Pi_{-} \\
        \hline
        \begin{array}{c}(1,3)\\
          (2,4)\end{array} & -\frac{2}{3}Z_L-\frac{1}{3}\1\\

        \hline
        \begin{array}{c}(1,2)\\
          (3,4)\end{array} & -\frac{1}{\sqrt{3}}X_L+\frac{1}{3}Z_L-\frac{1}{3}\1\\

        \hline
        \begin{array}{c}(1,4)\\
          (2,3)\end{array} & \frac{1}{\sqrt{3}}X_L+\frac{1}{3}Z_L-\frac{1}{3}\1\\
        \hline
      \end{array}
    \end{equation}
    \caption{Effective interactions produced by physical interaction acting on different choices of qubits.}
    \label{tab:1local}
  \end{table}

  \paragraph{Second-order perturbations}
  In order to make an effective interaction between two logical qubits we need to use physical interactions that act between two of these 4-qubit gadgets.
  We label the four physical qubits of one logical qubit as $1,2,3,4$, and the qubits of a second logical qubit with a dash $1',2',3',4'$ and consider a perturbation term of the form $H_2=\sum \alpha_{ij}H_{ij'}$.
  All first-order perturbation terms vanish as it is easy to show that $\Pi_{-}X_i\Pi_{-}=\Pi_{-}Y_i\Pi_{-}=\Pi_{-}Z_i\Pi_{-}=0$ for all $i \in \{1,2,3,4\}$.

  Let $H_0^{\text{tot}}=H_0\ox \1+\1\ox H_0$ be the total heavy Hamiltonian on these 8 qubits, and let $\Pi_{-}^{\text{tot}}$ project onto the groundspace of $H_0^{\text{tot}}$.

  We note that $Z_1\ket{\Psi^{-}}_{12}\ket{\Psi^{-}}_{34}=\ket{\Psi^{+}}_{12}\ket{\Psi^{-}}_{34}$ is an eigenvector of $H_{0}$ with eigenvalue 4, where $\ket{\Psi^{+}}=(\ket{01}+\ket{10})/\sqrt{2}$.
  Since the other eigenvector spanning the ground space of $H_0$, $\ket{\Psi^{-}}_{12}\ket{\Psi^{-}}_{34}$, is of a similar form, it is clear that $Z_1$ maps the ground space of $H_0$ into the eigenspace of eigenvalue 4.
  By unitary invariance of the Heisenberg interaction, and the symmetry between qubits 1, 2, 3, 4, we can say the same for any $X_i$, $Y_i$ or $Z_i$.
  This allows us to simplify the calculation of the second-order perturbation term,
  \begin{align}
    & -\Pi^{\text{tot}}_{-}H_2\Pi_{+}(H_0^{\text{tot}})^{-1}\Pi_{+}H_2\Pi^{\text{tot}}_{-}\\
    & = -\Pi_{-}^{\text{tot}}H_2\Pi_{+}\tfrac{1}{4+4}\Pi_{+}H_2\Pi_{-}^{\text{tot}}
    = -\tfrac{1}{8}\Pi_{-}^{\text{tot}}H_2^2\Pi_{-}^{\text{tot}} \\
    &=-\tfrac{1}{8}(\Pi_{-}\ox\Pi_{-})
      \left(\sum_{i,j,k,l=1}^{4}\alpha_{ij}\alpha_{kl}H_{ij'}H_{kl'}\right)
      (\Pi_{-}\ox\Pi_{-}) \\
    &= -\tfrac{1}{8}\sum_{i,j,k,l=1}^{4}\alpha_{ij}\alpha_{kl}
    \big[    (\Pi_{-}X_iX_k\Pi_{-})\ox(\Pi_{-}X_{j'}X_{l'}\Pi_{-})\\
        & \hspace{1cm} + (\Pi_{-}X_iY_k\Pi_{-})\ox(\Pi_{-}X_{j'}Y_{l'}\Pi_{-})
        +\dots
      \big].
  \end{align}
  Next, one can check that $\Pi_{-}X_iY_k\Pi_{-}=\Pi_{-}X_iZ_k\Pi_{-}=\Pi_{-}Y_iZ_k\Pi_{-}=0$ for any pair $(i,k)$, so many of these terms vanish.
  Remembering also that $\Pi_{-}X_iX_k\Pi_{-}=\Pi_{-}Y_iY_k\Pi_{-}=\Pi_{-}Z_iZ_k\Pi_{-}$, this expression simplifies to
  \begin{equation}
    -\tfrac{1}{8}\Pi_{-}H_2^2\Pi_{-}=-\tfrac{1}{8}\sum_{i,j,k,l=1}^{4}3 \alpha_{ij}\alpha_{kl}(\Pi_{-}X_iX_k\Pi_{-})\ox(\Pi_{-}X_{j'}X_{l'}\Pi_{-}),
  \end{equation}
  where the effective interactions produced by $\Pi_{-}X_iX_k\Pi_{-}$ can be read off again from \cref{tab:1local}.

  \begin{table}
    \[\begin{array}{|c|c|}
        \hline H_2 & \text{2\nbd-local part of effective interaction}\\
        \hline
        H_{11'}\mp H_{33'} & \pm Z_L Z_L\\
        H_{13'}-H_{11'} \pm H_{32'} & \pm Z_L X_L\\
        H_{11'}-2H_{22'}+H_{33'} & X_L X_L\\
        35 H_{11'}+5 H_{22'}-3H_{33'}+5H_{44'} & -X_L X_L\\
        \hline
      \end{array}\]
    \caption{Effective 2\nbd-local interactions produced from different choices of $H_2$ terms, up to a non-negative scaling factor.}
    \label{tab:2local}
  \end{table}

  By \cref{lem:secondorder}, for any $\epsilon > 0$ and sufficiently large $\Delta = \poly(\|H\|,1/\eta,1/\epsilon)$, $\Delta H_0+\Delta^{\frac{1}{2}}H_2+H_1$ $(\Delta,\eta,\epsilon)$-simulates the interaction $\Pi_-H_1\Pi_--\tfrac{1}{8}\Pi_{-}^{\text{tot}}H_2^2\Pi_{-}^{\text{tot}}$.
  Choosing $H_1$ as above we can cancel out any 1-local part of $\tfrac{1}{8}\Pi_{-}^{\text{tot}}H_2^2\Pi_{-}^{\text{tot}}$, so we are interested only in the 2\nbd-local part.
  \cref{tab:2local} shows some choices of $H_2$ with integer coefficients that generate effective interactions whose 2\nbd-local part is proportional to $\pm ZZ$, $\pm ZX$, $\pm XX$.

  By \cref{lem:interference}, we can apply this simulation to each interaction in $H$ in parallel.
  Letting $H'$ denote the overall simulator Hamiltonian, we finally obtain that, for any $\epsilon > 0$ and sufficiently large $\Delta = \poly(\|H\|,1/\eta,1/\epsilon)$, $H'$ is a $(\Delta,\eta,\epsilon)$-simulation of $H$.

  Everything follows through in exactly the same way for the XY interaction.
  If we set $H_{ij} = X_iX_j+Y_iY_j$ and use the same gadget, the ground space is left unchanged.
  So the only thing to check is that $X_i,Y_i,Z_i$ all map the ground space into an eigenspace of $H_0$ again (which they do!).
  Then the simulated interactions will be the same up to a constant factor of $2/3$.
\end{proof}

Finally, we show that every remaining class of qudit Hamiltonians in \cref{fig:reductions} can simulate either XY interactions or Heisenberg interactions, implying that they are all universal too.

\begin{theorem}
  \label{thm:reductions}
  Let $\mathcal{S}$ be a set of interactions on at most 2 qubits.
  Assume that there does not exist $U \in SU(2)$ such that, for each 2\nbd-qubit matrix $H_i \in \mathcal{S}$, $U^{\ox 2} H_i (U^\dg)^{\ox 2} = \alpha_i Z^{\ox 2} + A_i \ox \1 + \1 \ox B_i$,
  where $\alpha_i \in \R$ and $A_i$, $B_i$ are arbitrary single-qubit Hamiltonians.
  Then $\mathcal{S}$-Hamiltonians can simulate either $\{XX+YY+ZZ\}$-Hamiltonians or $\{XX+YY\}$-Hamiltonians.
  Furthermore, if the interaction graph of the target Hamiltonian is a 2D square lattice, then the simulator Hamiltonian may also be chosen to be on a 2D square lattice.
\end{theorem}

Observe that the assumption in the theorem is equivalent to assuming that the set formed by extracting the 2\nbd-local parts of each interaction in $\mathcal{S}$ is not simultaneously locally diagonalisable.
\Cref{thm:reductions} was first proven in~\cite{Cubitt-Montanaro}, with the restriction to 2D square lattices shown in~\cite{Piddock-Montanaro}.
These proofs use different terminology (e.g.\ they prove ``reductions'' rather than ``simulations'').
However, all the gadgets used are examples of mediator qubit gadgets or first order subspace encoding gadgets which, as described in \cref{sec:perturbative}, give simulations in our terminology.
We therefore restrict ourselves here to sketching the arguments of~\cite{Cubitt-Montanaro, Piddock-Montanaro}.
See~\cite{Cubitt-Montanaro, Piddock-Montanaro} for a full proof of correctness and technical details.

\begin{table}
  \begin{center}
    \begin{tabular}{|l|l|l|}
      \hline Simulator interaction $H$ & Simulated interaction $H'$ & Gadget\\
      \hline $XX + \alpha YY$ & $XX+YY$ & $H_{ab} + H_{bc}$\\
      $XX + \alpha YY + \beta ZZ$ & $XX + \alpha' YY$ & $H_{ab} - H_{bc}$\\
      $XZ-ZX$ & $XX+YY$ & $H_{ab} + H_{bc} + H_{ca}$\\
      \hline
    \end{tabular}
    \caption[Subspace encodings used in \cref{thm:reductions}.]{%
      Subspace encodings used in \cref{thm:reductions}.
      In each case a qubit is encoded within the ground space of $H$ acting on three qubits labelled $a$--$c$.
      Here $\alpha$, $\beta$, $\alpha'$ are fixed nonzero real numbers.}
    \label{tab:subspaces}
  \end{center}
\end{table}

\begin{proof}[sketch]
  The claim follows by chaining together various simulations from~\cite{Cubitt-Montanaro} in the same order as used in that work; the sequence of simulations used is illustrated in \cref{fig:reductions}.
  To prove the final part of the theorem, each of the gadgets used in~\cite{Cubitt-Montanaro} can be replaced with a gadget from~\cite{Piddock-Montanaro} which fits onto a square lattice.
  Most of the steps of the argument show that, given access to one interaction $H$, we can effectively produce another interaction $H'$.
  Three of these are listed in \cref{tab:subspaces}.
  These are all simulations of the subspace encoding type, where we encode one logical qubit within the ground space of a Hamiltonian on 3 physical qubits.
  The simulations can be analysed using \cref{lem:firstorder} and, as they are subspace encodings, satisfy the definition of simulation.
  By applying the right interactions across qubit triples, we obtain new effective interactions between logical qubits.
  The effective interactions produced are calculated in~\cite{Cubitt-Montanaro}.
  Alternatively the mediator qubit gadgets of Figure 8 and Figure 11 of~\cite{Piddock-Montanaro} may be used to perform the same simulations on a square lattice.

  A somewhat different case is the interaction $H = XX + \alpha YY + \beta ZZ + A\1 + \1 A$, where at least one of $\alpha$ and $\beta$ is nonzero.
  Here the available interaction corresponds to one which was considered in \cref{tab:subspaces}, but with an additional 1\nbd-local term of some form.
  The simulation deletes these 1\nbd-local terms by introducing 4 ancilla qubits for each logical qubit $a$.
  Labelling these qubits $a$--$d$, it turns out that the ground state of $H_0 = H_{ab} + H_{cd} - H_{ac} - H_{bd}$ is unique and maximally-entangled across the $(a-c:d)$ split.
  If these four qubits are forced to be in this state, applying a $-H$ interaction between 4 and $a$ corresponds to a $-A$ term applied to $a$.
  This allows the local $A$ terms to be effectively deleted for each $H$ interaction used.
  The corresponding isometry $V$ attaches 4 ancilla qubits for each of the original qubits, in the ground state of $H_0$.
  The interaction $H = XZ - ZX + A\1 - \1 A$ is similar; here the local part of $H$ can be deleted using $H_0 = H_{ab} + H_{bc} + H_{cd} + H_{da}$.
  Section 4.6 of~\cite{Piddock-Montanaro} shows how these gadget constructions may be adjusted slightly such that they fit onto a 2D square lattice.

  Now that these special cases have been dealt with, to complete the argument we need to consider an arbitrary set $\mathcal{S}$ of 2\nbd-qubit interactions where there is no $U \in SU(2)$ such that, for each 2\nbd-qubit matrix $H_i \in \mathcal{S}$, $U^{\ox 2} H_i (U^\dg)^{\ox 2} = \alpha_i Z^{\ox 2} + A_i \1 + \1 B_i$.
  We sketch the argument and defer to~\cite{Cubitt-Montanaro} for details.

  Any 2\nbd-qubit interaction $H_i$ can be decomposed in terms of parts which are symmetric and antisymmetric under interchange of the qubits on which it acts, and each of these parts can be extracted by taking linear combinations of $H_i$ and the interaction obtained by swapping the two qubits; so we can assume that all the interactions in $\mathcal{S}$ are either symmetric or antisymmetric.
  The 2\nbd-local part of any symmetric interaction $H_i$ can be written as $\sum_{s,t \in \{x,y,z\}} M^{(i)}_{st} \sigma_s \ox \sigma_t$ for some symmetric $3 \times 3$ matrix $M^{(i)}$.
  Define the Pauli rank of $H_i$ to be the rank of $M^{(i)}$.
  If there exists $H_i \in \mathcal{S}$ with Pauli rank 2, we consider Hamiltonians produced only using $H_i$ interactions.
  As discussed in \cref{sec:hamsim}, by applying local unitaries and up to rescaling and relabelling Pauli matrices, we can replace $H_i$ with $XX + \alpha YY + \beta ZZ + A\1 + \1 A$ for some $A$, and some $\alpha,\beta \in \R$ such that at least one of them is nonzero.
  This is the special case we just considered.

  Otherwise, all $H_i \in \mathcal{S}$ have Pauli rank 1; we also know that there must exist $H_i, H_j \in \mathcal{S}$ such that the 2\nbd-local parts of $H_i$ and $H_j$ do not commute, by the assumptions of the theorem.
  This implies that there must exist some linear combination of $H_i$ and $H_j$ which has Pauli rank at least 2.
  Considering this linear combination, we are back in the same special case as before.
  Finally, the case where $\mathcal{S}$ contains an antisymmetric interaction can be dealt with in a similar way, by using local unitaries to put that interaction into the previously considered canonical form $XZ - ZX + A\1 - \1 A$.
\end{proof}

We finally observe that it was shown in~\cite{Piddock-Montanaro} that certain interactions remain universal even if they are only permitted to occur with non-negative weights.
Indeed, that work showed that the class of qubit Hamiltonians whose interactions are of the form $\alpha XX + \beta YY + \gamma ZZ$, where $\{\alpha + \beta, \alpha + \gamma, \beta + \gamma\} > 0$, can simulate qubit Hamiltonians with arbitrarily positively or negatively weighted interactions of the form $\alpha' XX + \beta' YY + \gamma' ZZ$, for some $\alpha'$, $\beta'$, $\gamma'$ such that at least two of $\alpha'$, $\beta'$, $\gamma'$ are nonzero.
This implies, for example, that the antiferromagnetic Heisenberg interaction is universal.


\subsection{Indistinguishable particles}
Throughout this work so far, we have only considered Hamiltonians on distinguishable particles with finite-dimensional Hilbert spaces.
As stated, our results -- and even the definitions of Hamiltonian encoding and simulation -- do not apply to indistinguishable particles or infinite-dimensional Hilbert spaces.
Extending these definitions to arbitrary self-adjoint operators on infinite-dimensional Hilbert spaces is beyond the scope of the present article.\footnote{As the definition and characterisation of encodings in particular is very $C^*$-algebraic in character, it does not seem too difficult to generalise.}

However, as bosonic and fermionic systems are ubiquitous in many-body physics, and our main focus is to show that there exist simple, universal quantum models that are able to simulate the physics of any other physical system, we will address the question of whether universal spin models such as the Heisenberg- and XY-models can simulate indistinguishable particles.
In fact, the required simulations follow from standard techniques for mapping fermionic and bosonic operators to spin operators, so we only sketch the arguments here.

\subsubsection{Fermions}
The canonical anti-commutation relation (CAR) algebra describing fermions is generated by fermionic creation and annihilation operators $c_i$, $c_i^\dg$ satisfying $\{c_i,c_j\} = 0$ and $\{c_i,c^\dg_j\} = \delta_{ij}$ (where the subscript indexes different fermionic modes).
This algebra is finite-dimensional (as long as the single-particle Hilbert space is).
It is well known that this algebra can be embedded into an operator algebra acting on a many-qubit system, e.g.\ by the well-known Jordan-Wigner transformation: $c_i = -\bigotimes_{j\leq i}Z_i \ox \frac{X_i + iY_i}{2}$, where we define some arbitrary total-ordering on the qubits.
However, this is not sufficient for our purposes.
It transforms individual fermionic creation or annihilation operators into operators that act non-trivially on all qubits in the system, so does not give a local encoding.

The mapping introduced by Bravyi and Kitaev~\cite{Bravyi-Kitaev} improves this to $\log n$-local operators (where $n$ is the total number of qubits)\footnote{Very recent independent work has given an analysis and comparison of different fermion-to-qubit mappings~\cite{havlicek17}.}.
However, simulating these $\log n$-local interactions using a universal model with two-body interactions, such as the Heisenberg- or XY-model, will require local interactions whose norms scale super-polynomially in $n$.
Whilst this gives a simulation with polynomial overhead in terms of the system size, it is not strictly speaking efficient according to our definition due to this super-polynomial scaling of the local interaction strengths.

Both of these mappings produce qubit Hamiltonians with the same number of qubits as fermionic modes.
This is much stronger than required for an efficient simulation in the the spirit of \cref{dfn:sim}, which allows a polynomial overhead in the simulator system size.
The fermion-to-spin mappings studied in~\cite{Verstraete-Cirac,Ball_fermions,Farrelly-Short,Troyer_fermions} preserve locality by adding additional auxiliary fermionic modes before mapping to qubits, at the expense of a polynomial system-size overhead.
The auxiliary fermions must be restricted to the appropriate subspace, which can be done by adding strong local terms to the Hamiltonian (see \cite{Ball_fermions,Verstraete-Cirac}).
(These strong local terms mutually commute, and when transformed to spin operators become products of Paulis.
So these terms in fact form a stabilizer Hamiltonian.)
Together with these strong local terms, this mapping gives a spin Hamiltonian that exactly reproduces the original fermionic Hamiltonian in its low-energy subspace.
The resulting spin Hamiltonian is local if the simulated fermionic system is a regular lattice Hamiltonian containing only even products of fermionic creation and annihilation operators \cite{Verstraete-Cirac}.
Simulating the resulting spin Hamiltonian using any universal model then gives an efficient simulation for this important class of fermionic Hamiltonians.

\subsubsection{Bosons}
In the case of bosons, the canonical commutation relation (CCR) algebra, generated by bosonic creation and annihilation operators $a_i$, $a_i^\dg$ satisfying $[a_i,a_j] = 0$ and $[a_i,a^\dg_j] = \delta_{ij}$, is infinite-dimensional.
To simulate bosons with spins, one must necessarily restrict to some finite-dimensional subspace of the full Hilbert space, and only simulate the system within that subspace.
The appropriate choice of subspace will depend on the particular bosonic system, and which physics one wishes to simulate, so one cannot give a completely general result here.

However, a natural choice will often be to limit the maximum number of bosons to some finite value $N$, i.e.\ to restrict to the finite-dimensional subspace spanned by eigenstates of the total number operator $\sum_i a_i^\dg a_i$ with eigenvalue $\leq N$.
For systems containing multiple bosonic modes, we can alternatively limit the maximum number of bosons in each mode separately, i.e.\ restrict to the subspace spanned by eigenvectors with eigenvalue $\leq N$ for each $a_i^\dg a_i$ individually.
(Since $[a_i^\dg a_i,a_j^\dg a_j] = 0$, this subspace also contains the subspace with maximum \emph{total} number of bosons $N$.)

In this way, each bosonic mode is restricted individually to a finite-dimensional subspace that can be represented by the Hilbert space of a qudit.
The original bosonic Hamiltonian restricted to this subspace is clearly equivalent to some Hamiltonian on these qudits.
Furthermore, since $[a_i,a_j] = [a_i,a_j^\dg] = [a_i^\dg,a_j^\dg] = 0$ for $i\neq j$, $k$-particle bosonic interactions become $k$-local interactions on the qudits.
The resulting $k$-local qudit Hamiltonian can then be simulated by the universal model, as shown in previous sections.

In fact, restricting the bosonic creation and annihilation operators to the finite-particle-number subspace in this way is a well-known procedure.
The equivalent qudit operators $S_i^\pm$ are given by the (exact) Holstein-Primakov transformation~\cite{Holstein-Primakoff}:
\begin{equation}
  S_i^+ =\sqrt{d-1}\sqrt{1-\frac{a_i^\dg a_i}{d-1}}\, a_i,\qquad
  S_i^- =\sqrt{d-1}a_i^\dg\sqrt{1-\frac{a_i^\dg a_i}{d-1}}.
\end{equation}


\subsection{Universal stoquastic simulators}
\label{sec:stoquastic}

It was previously shown by Bravyi and Hastings~\cite{Bravyi-Hastings} that the Ising model with transverse fields acts as a universal simulator for the class of stoquastic 2-local Hamiltonians.
The transverse Ising model (TIM) corresponds to Hamiltonians which can be written as a weighted sum of terms picked from the set $\mathcal{S} = \{XX,Z\}$.
A Hamiltonian is said to be stoquastic if its off-diagonal matrix entries are all nonpositive in the computational basis~\cite{bravyi06}.
Bravyi and Hastings used a slightly different notion of simulation to the one we define here; as discussed in \cref{sec:hamsim}, the most important difference is that in our notion of simulation, the encoding operation must be local.

In~\cite{Bravyi-Hastings}, a sequence of 5 encodings is used to map 2\nbd-local stoquastic Hamiltonians to the transverse Ising model.
We check each of the encodings in turn to see that the encodings are indeed local, so the overall result goes through with our definitions.
The encodings proceed through a succession of other physical models, which we avoid defining here; see~\cite{Bravyi-Hastings} for the details.

The encodings used are:
\begin{itemize}
\item TIM simulates HCD on a triangle-free graph: the encoding is the identity map.
\item HCD on a triangle-free graph simulates HCB$_2$: the encoding attaches one additional qubit $v'$ to each vertex $v$, and a qubit for each edge in the interaction graph.
  Each of the edge qubits is in the state $\ket{0}$, and for each vertex $v$, $\ket{0}_v$ is encoded as $\ket{0}_v \ket{0}_{v'}$, $\ket{1}_v$ is encoded as $\ket{1}_v \ket{1}_{v'}$.
  This is clearly a local encoding.
\item HCB$_2$ simulates HCB$_1$: the encoding attaches $\poly(n)$ additional qubits, each in the state $\ket{0}$.
\item HCB$_1$ simulates HCB$^*_1$: the encoding is the identity map.
\item HCB$^*_1$ simulates 2\nbd-local stoquastic Hamiltonians: the encoding maps each qubit to a subspace of two qubits in a ``dual rail'' encoding, and attaches some additional ``mediator'' qubits in a state which is a product of states of $O(1)$ qubits.
\end{itemize}
As these encodings are all local, we obtain that the transverse Ising model is a universal simulator for the class of 2-local stoquastic Hamiltonians.

To extend this simulation to $k$-local stoquastic Hamiltonians for $k > 2$, one can use a result from~\cite{bravyi06}.
This work gave (in our terminology) a simulation of $k$-local termwise-stoquastic Hamiltonians with 2-local stoquastic Hamiltonians.
The simulation is efficient for $k=O(1)$.
A termwise-stoquastic $k$-local Hamiltonian $H$ is one for which the matrices $H_S$ occurring in the decomposition $H = \sum_S H_S$, where each subset $S$ of subsystems on which $H_S$ acts is of size at most $k$, can be taken to be stoquastic.
Although all stoquastic Hamiltonians on $n$ qubits are clearly termwise-stoquastic when viewed as $n$-local Hamiltonians, not all stoquastic $k$-local Hamiltonians are termwise-stoquastic when viewed as $k$-local~\cite{bravyi06}.
Thus, using the simulation of~\cite{bravyi06}, we obtain that the transverse Ising model is a universal simulator for stoquastic Hamiltonians, but the simulation is only efficient for termwise-stoquastic Hamiltonians.

It is shown in~\cite{Cubitt-Montanaro}, using similar techniques to the proof of \cref{thm:reductions}, that any family of Hamiltonians built from interactions of the form $H = \alpha Z^{\ox 2} + A \ox \1 + \1 \ox B$, where $A$ or $B$ is not diagonal, can simulate TIM Hamiltonians. Thus any family of Hamiltonians of this form is also a universal stoquastic Hamiltonian simulator.


\subsection{Classification of two-qubit interactions}

We can complete the universality picture for two-qubit interactions by classifying the interactions into universality families.
Combining the result of the previous section with \cref{thm:reductions} and a previous classification of universal classical Hamiltonians~\cite{Science}, we obtain a full classification of universality classes:

\begin{theorem}
  Let $\mathcal{S}$ be any fixed set of two-qubit and one-qubit interactions such that $\mathcal{S}$ contains at least one interaction which is not 1-local. Then:
  \begin{itemize}
  \item If there exists $U \in SU(2)$ such that $U$ locally diagonalises $\mathcal{S}$, then $\mathcal{S}$\nbd-Hamiltonians are universal classical Hamiltonian simulators~\cite{Science};
  \item Otherwise, if there exists $U \in SU(2)$ such that, for each 2\nbd-qubit matrix $H_i \in \mathcal{S}$, $U^{\ox 2} H_i (U^\dg)^{\ox 2} = \alpha_i Z^{\ox 2} + A_i \ox \1 + \1 \ox B_i$, where $\alpha_i \in \R$ and $A_i$, $B_i$ are arbitrary single-qubit Hamiltonians, then $\mathcal{S}$\nbd-Hamiltonians are universal stoquastic Hamiltonian simulators~\cite{Bravyi-Hastings,Cubitt-Montanaro};
  \item Otherwise, $\mathcal{S}$\nbd-Hamiltonians are universal quantum Hamiltonian simulators.
  \end{itemize}
\end{theorem}

We remark that the definition of universal classical simulation used in~\cite{Science} does not quite match up with our notion of universal quantum simulation.
Similarly to ours, that work associates a small number of physical qubits with each logical qubit in the simulation.
However, in~\cite{Science} the sets of physical qubits associated with distinct logical qubits are allowed to overlap.
Also note that, as discussed in \Cref{sec:stoquastic}, the second (stoquastic) class of universal simulators is only efficient for termwise-stoquastic Hamiltonians.

Since the two-qubit interactions that are \emph{not} universal must satisfy a non-trivial set of algebraic constraints, this classification immediately implies that generic two-qubit interactions are universal, an implication that can be formalised as follows:

\begin{corollary}
  Given any measure on the set of two-qubit Hamiltonians with full support, the subset of universal Hamiltonians has full measure.
\end{corollary}


\subsection{Spatial sparsity and simulation on a square lattice}

Up to this point, we have not assumed anything about the spatial locality of the Hamiltonians we are simulating, nor the simulator Hamiltonians.
Indeed, even if the target Hamiltonian has a rather simple spatial structure -- for example, is a lattice Hamiltonian -- this structure need not be preserved in the simulator Hamiltonian.
We now show that in certain cases we can find universal simulators where all interactions take place on a square lattice.
The price paid for simulating general Hamiltonians in this way (for example, those with long-range interactions) is an exponential increase in the weights required in the simulator.
However, when the target Hamiltonian is spatially sparse (a class which encompasses all 2D lattice Hamiltonians), this exponential increase can be avoided.

\begin{definition}[Spatial sparsity~\cite{Oliveira-Terhal}]
  A spatially sparse interaction graph $G$ on $n$ vertices is defined as a graph in which \begin{inparaenum}
  \item every vertex participates in $O(1)$ edges,
  \item there is a straight-line drawing in the plane such that every edge overlaps with $O(1)$ other edges and the length of every edge is $O(1)$.
  \end{inparaenum}
\end{definition}

\begin{lemma}
\label{lem:squarelattice}
Let $\mathcal{S}$ be either $\{XX+YY+ZZ\}$, $\{XX+YY\}$ or $\{XX,Z\}$.
Then any $\mathcal{S}$-Hamiltonian $H$ on $n$ qubits can be simulated by a $\mathcal{S}$-Hamiltonian on a square lattice of $\poly(n)$ qubits using weights of $O(n\Lambda_0(1/\epsilon+1/\eta))^{\poly(n)}$ size, where $\Lambda_0$ is the size of the largest weight in $H$. Furthermore if the target Hamiltonian is spatially sparse, then the weights need only be of size $O(\poly(n\Lambda_0(1/\epsilon+1/\eta)))$.

\end{lemma}

\begin{proof}
  The final part of the statement concerning spatially sparse Hamiltonians was originally shown in \cite{Oliveira-Terhal} for $\{XX,Z\}$-Hamiltonians.
  The proof used three gadgets called fork, crossing and subdivision gadgets pictured in \cref{fig:OTgadgets}, which we briefly describe here.

  The subdivision gadget simulates an $XX$ interaction between two non-interacting qubits $a,b$ using a mediator qubit $e$, as pictured in \cref{fig:subdivision}.
  This can be used $O(\log k)$ times in series to simulate an interaction between two qubits separated by $k$ qubits.
  The fork gadget simulates the interactions $X_aX_b+X_aX_c$ using only one interaction involving qubit $a$, as pictured in \cref{fig:fork}.
  This can be used multiple times in parallel to reduce the degree of the vertex $a$ in the interaction graph.
  The crossing gadget is used to simulate $X_aX_c+X_bX_d$, for 4 qubits $a,b,c,d$ arranged as shown in \cref{fig:crossing}, via an interaction graph that has no crossings.

  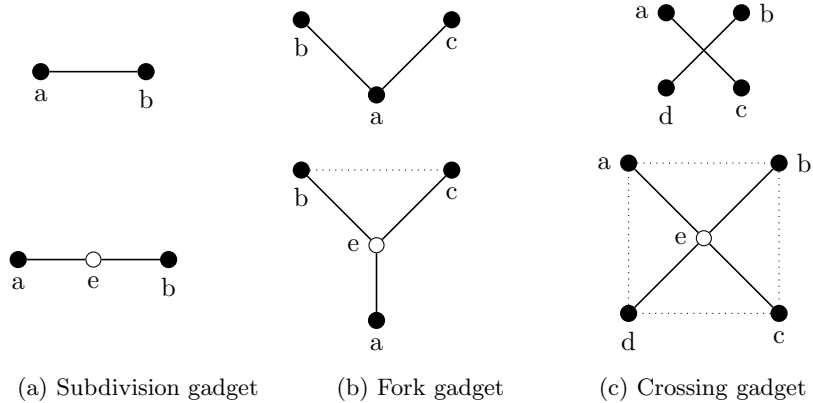
\begin{figure}
    \centering
    \begin{subfigure}[b]{0.3\textwidth}
      \begin{tikzpicture}
        \begin{scope}
          \node[qubit] (a) at (-1,1)[label=below:a]{};
          \node[qubit] (b) at (1,1) [label=below:b] {};
          \node[mqubit] (e) at (0,1)[label=below:e]{};
          \node at (0,-0.2){};
          \draw [black, semithick] (a) to (e);
          \draw [black, semithick] (b) to (e);
        \end{scope}

        \begin{scope}[shift={(0,2.5)}]
          \node[qubit] (a) at (-0.7,1)[label=below:a]{};
          \node[qubit] (b) at (0.7,1) [label=below:b] {};

          \draw [black, semithick] (a) to (b);
        \end{scope}
      \end{tikzpicture}
      \caption{Subdivision gadget}
      \label{fig:subdivision}
    \end{subfigure}
    \begin{subfigure}[b]{0.3\textwidth}
      \begin{tikzpicture}
        \begin{scope}
          \node[qubit] (b) at (-1,2) [label=below:b] {};
          \node[qubit] (c) at (1,2)[label=below:c]{};
          \node[mqubit] (e) at (0,1)[label=left:e]{};
          \node[qubit] (a) at (0,0)[label=below:a]{};

          \draw [black, semithick] (a) to (e);
          \draw [black, semithick] (b) to (e);
          \draw [black, semithick] (c) to (e);
          \draw [dotted] (b) to (c);
        \end{scope}

        \begin{scope}[shift={(0,3)}]
          \node[qubit] (b) at (-1,1) [label=below:b] {};
          \node[qubit] (c) at (1,1)[label=below:c]{};
          \node[qubit] (a) at (0,0)[label=below:a]{};

          \draw [black, semithick] (a) to (b);
          \draw [black, semithick] (a) to (c);

        \end{scope}
      \end{tikzpicture}
      \caption{Fork gadget}
      \label{fig:fork}
    \end{subfigure}
    \begin{subfigure}[b]{0.3\textwidth}
      \centering
      \begin{tikzpicture}
        \begin{scope}
          \node[qubit] (a) at (-1,2) [label=left:a] {};
          \node[qubit] (b) at (1,2)[label=right:b]{};
          \node[mqubit] (e) at (0,1)[label=left:e] {};
          \node[qubit] (c) at (1,0)[label=below:c]{};
          \node[qubit] (d) at (-1,0)[label=below:d]{};

          \draw [black, semithick] (a) to (e);
          \draw [black, semithick] (b) to (e);
          \draw [black, semithick] (c) to (e);
          \draw [black, semithick] (d) to (e);
          \draw [dotted] (a) to (b);
          \draw [dotted] (b) to (c);
          \draw [dotted] (c) to (d);
          \draw [dotted] (d) to (a);
        \end{scope}

        \begin{scope}[shift={(0,3)}]
          \node[qubit] (a) at (-0.5,1) [label=left:a] {};
          \node[qubit] (b) at (0.5,1)[label=right:b]{};
          \node[qubit] (c) at (0.5,0)[label=below:c]{};
          \node[qubit] (d) at (-0.5,0)[label=below:d]{};
          \draw [black, semithick] (a) to (c);
          \draw [black, semithick] (b) to (d);
        \end{scope}
      \end{tikzpicture}
      \caption{Crossing gadget}
      \label{fig:crossing}
    \end{subfigure}

    \caption[Subdivision, fork and crossing gadgets.]{%
      Subdivision, fork and crossing gadgets.
      In each case the top interaction pattern is simulated using the gadget underneath.
      White vertices denote mediator qubits with heavy 1-local terms applied.}
    \label{fig:OTgadgets}
  \end{figure}

  These gadgets can be used to simulate a spatially sparse Hamiltonian on a square lattice using only O(1) rounds of perturbation theory; we defer to~\cite{Oliveira-Terhal} for the technical details.
  The gadgets were generalised for the interactions $XX+YY+ZZ$ and $XX+YY$ in \cite{Piddock-Montanaro}, where the mediator qubit $e$ is replaced with a pair of mediator qubits, in order to prove the result for $\{XX+YY+ZZ\}$-Hamiltonians and $\{XX+YY\}$-Hamiltonians in the same way.

  \begin{figure}
    \centering
    \begin{subfigure}[t]{\textwidth}
      \begin{tikzpicture}[scale=0.3]
        \def\n{5};
        \foreach \i in {0,...,4}
        {\node[qubit] (\i) at (8*\i,0){};
          \node[mqubit] (\i 1) at (8*\i-3,4){};
          \node[mqubit] (\i 2) at (8*\i-1,4){};
          \node[mqubit] (\i 3) at (8*\i+1,4){};
          \node[mqubit] (\i 4) at (8*\i+3,4){};
          \foreach \j in {1,...,4}
          {\draw (\i) -- (\i\j);}}

        \foreach \i in {0,...,3}
        {\pgfmathsetmacro\k{\i+1}
          \foreach \j in {\k,...,4}
          {\draw  (8*\i+2*\j-5,4)node[mqubit]{} to[out=60, in=120]  (8*\j+2*\i-3,4)node[mqubit]{};}}

      \end{tikzpicture}
      \caption{First subdivide each edge to isolate each of the high degree vertices.}
      \label{fig:complete1}
      \vspace*{0.75cm}
    \end{subfigure}

    \begin{subfigure}{\textwidth}
      \begin{tikzpicture}[scale=0.25]
        \def\m{2}
        \pgfmathsetmacro\n{2^\m}
        \draw[step=1cm,gray,very thin] (-3,0) grid (35,15);
        \foreach \i in {0,...,\n}
        {\node[qubit] (\i 0) at (8*\i,0){};
          \draw[thick] (\i 0) to (8*\i,2);
          \node[qubit] (a) at (8*\i-2,2){};
          \node[qubit] (b) at (8*\i+2,2){};
          \draw[thick] (a) to (b);

          \draw[thick] (a) to ($(a)+(0,2)$);
          \node[qubit](\i 1) at ($(a)+(-1,2)$){};
          \node[qubit](\i 2) at ($(a)+(1,2)$){};
          \draw[thick] (\i 1) to (\i 2);

          \draw[thick] (b) to ($(b)+(0,2)$);
          \node[qubit](\i 3) at ($(b)+(-1,2)$){};
          \node[qubit](\i 4) at ($(b)+(1,2)$){};
          \draw[thick] (\i 3) to (\i 4);}

        \draw[thick](01) -- ($(01)+(0,1)$) -- ($(11)+(0,1)$) --(11);
        \draw[thick](02) -- ($(02)+(0,2)$) -- ($(21)+(0,2)$) --(21);

        \draw[thick](03) -- ($(03)+(0,3)$) -- ($(31)+(0,3)$) --(31);
        \draw[thick](04) -- ($(04)+(0,4)$) -- ($(41)+(0,4)$) --(41);
        \draw[thick](12) -- ($(12)+(0,5)$) -- ($(22)+(0,5)$) --(22);
        \draw[thick](13) -- ($(13)+(0,6)$) -- ($(32)+(0,6)$) --(32);
        \draw[thick](14) -- ($(14)+(0,7)$) -- ($(42)+(0,7)$) --(42);
        \draw[thick](23) -- ($(23)+(0,8)$) -- ($(33)+(0,8)$) --(33);
        \draw[thick](24) -- ($(24)+(0,9)$) -- ($(43)+(0,9)$) --(43);
        \draw[thick](34) -- ($(34)+(0,10)$) -- ($(44)+(0,10)$) --(44);

        \draw [decorate,decoration={brace,amplitude=8pt}] (-4,0) -- (-4,3.9) node [black,midway,xshift=-1.2cm] {$O(\log n)$};
        \draw [decorate,decoration={brace,amplitude=8pt}] (-4,4.1) -- (-4,15) node [black,midway,xshift=-1cm] {$O(n^2)$};
        \draw [decorate,decoration={brace,amplitude=8pt}] (35,-1) -- (-3,-1) node [black,midway,yshift=-0.7cm] {$O(n^2)$};
      \end{tikzpicture}
      \caption[Use the fork gadget $O(\log n)$ times at each of the high degree vertices.]{%
        Use the fork gadget $O(\log n)$ times at each of the high degree vertices, and lay out the resulting interaction pattern on a 2D lattice as shown above.
        Finally use the subdivision and crossing gadgets until the Hamiltonian is on the 2D square lattice.}
      \label{fig:complete2}
      \vspace*{0.25cm}
    \end{subfigure}

    \caption{How to simulate a Hamiltonian whose interaction pattern is the complete graph on $n=5$ qubits with a Hamiltonian on a 2D square lattice.}
  \end{figure}
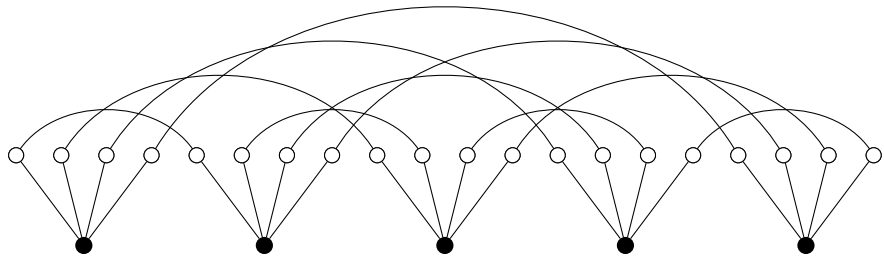
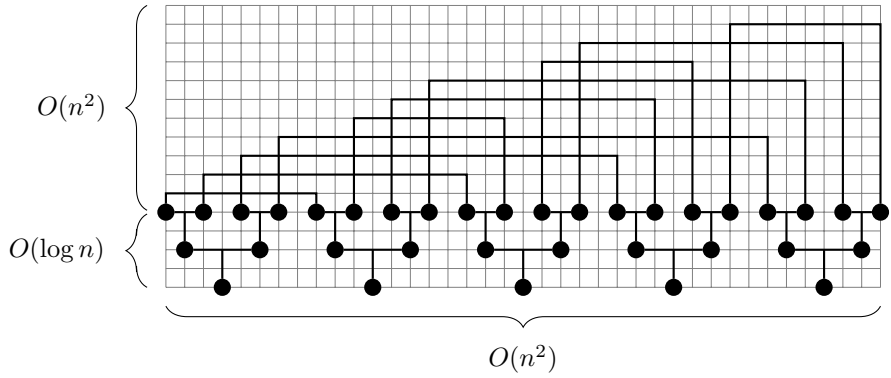

  Here we show how, if we allow more than O(1) rounds of perturbation theory, the same gadgets can be used to simulate a 2\nbd-local Hamiltonian whose interaction pattern is the complete graph on $n$ qubits, via a simulator Hamiltonian on a square lattice of size $O(n^2)\times O(n^2)$.
  Any interaction graph which is a subgraph of the complete graph can easily be simulated using the same construction, simply by setting some weights to zero.

  First, lay out the $n$ qubits in a line.
  Each vertex in the interaction graph has $n-1$ incoming edges.
  Subdivide each edge just once to isolate these high degree vertices to obtain an interaction graph as shown in \cref{fig:complete1}.
  Then using the fork gadget $O(\log n)$ times in series allows us to replace these with binary trees of depth $O(\log n)$, which can be placed directly onto a square lattice as shown in \cref{fig:complete2}.
  The long range interactions in this graph (which are of length at most $O(n^2)$), can be fitted to the edges of the square lattice using $O(\log n)$ applications of the subdivision gadget.

  At each crossing, we also need to use a crossing gadget -- note that the interactions $X_aX_b$, $X_bX_c$, $X_cX_d$, and $X_dX_a$ in \cref{fig:crossing} may be subdivided using a subdivision gadget so that the crossing gadget fits on the square lattice.
  If there is not enough space to put two crossing gadgets next to each other, then the lattice spacing can be made twice as narrow to make space.
  This only makes a constant factor difference to the number of qubits used and the number of rounds of perturbation theory required.

  The whole procedure therefore requires a total of $O(\log n)$ rounds of perturbation theory.
  By \cref{lem:secondorder}, second-order perturbation theory requires the weights of the simulator Hamiltonian to be of size $O(\Lambda^6/\epsilon^2 + \Lambda^2/\eta^2)$, where $\Lambda$ is the size of the terms $H_1$ and $H_2$.
  Given the simple nature of the gadgets used here, $\Lambda = O(\poly(n)\Lambda_0)$ where $\Lambda_0$ is the size of the largest weight in $H$.
  Therefore $r$ rounds of perturbation theory requires weights of size
  \begin{equation}
    \Lambda_{sim}=O\left(\poly(n)\Lambda_0\left(\frac{1}{\epsilon}+\frac{1}{\eta}\right)\right)^{6^r}
  \end{equation}
  Simulating the complete graph as described above requires $r=O(\log n)$, so the weights of the simulator system are $\Lambda_{sim} =(n\Lambda_0(1/\epsilon+1/\eta))^{\poly(n)}$.
  However, for a spatially sparse Hamiltonian simulated using only $r=O(1)$ rounds of perturbation theory as described in~\cite{Oliveira-Terhal} the weights scale as $\Lambda_{sim}=\poly(n\Lambda_0(1/\epsilon+1/\eta))$.
\end{proof}


\section{Consequences of universality}

We finally discuss some implications of our results for quantum computation.

\subsection{QMA-completeness}
Oliveira and Terhal showed in~\cite{Oliveira-Terhal} that the local Hamiltonian problem for spatially sparse qubit Hamiltonians is QMA-complete.
It is observed in~\cite{Cubitt-Montanaro} that this spatially sparse Hamiltonian may be assumed to not contain any $Y$ terms in its Pauli decomposition, by combining the work of~\cite{Oliveira-Terhal} with a result of \cite{Biamonte-Love}.
Notice that the simulations in \cref{thm:pauliterms} and \cref{thm:heisenberg} result in a spatially sparse simulator Hamiltonian if the target Hamiltonian is spatially sparse.
Combined with \cref{lem:squarelattice}, these results show that the Heisenberg interaction on a square lattice can efficiently simulate any spatially sparse qubit Hamiltonian with no $Y$ terms, and is therefore QMA-complete.
This was previously shown by Schuch and Verstraete~\cite{Schuch-Verstraete} in the case where arbitrary 1-local terms are allowed at every site; the novelty here is that QMA-completeness still holds even if these terms are not present.

This removes the caveat of Theorem~3 in~\cite{Piddock-Montanaro}, which can now be fully stated as:
\begin{theorem}
  Let $\mathcal{S}$ be a set of interactions on at most 2 qubits.
  Assume that there does not exist $U \in SU(2)$ such that, for each 2\nbd-qubit matrix $H_i \in \mathcal{S}$, $U^{\ox 2} H_i (U^\dg)^{\ox 2} = \alpha_i Z^{\ox 2} + A_i \ox \1 + \1 \ox B_i$, where $\alpha_i \in \R$ and $A_i$, $B_i$ are arbitrary single-qubit Hamiltonians.
  Then the local Hamiltonian problem for $\mathcal{S}$-Hamiltonians is QMA-complete even if the interactions are restricted to the edges of a 2D square lattice.
\end{theorem}
Using further gadget constructions from~\cite{Piddock-Montanaro}, one can even show that the \emph{antiferromagnetic} Heisenberg interaction is QMA-complete on a triangular lattice.


\subsection{Quantum computation by simulation}

We can connect universal quantum Hamiltonians to universality for quantum computation.
Many constructions are now known (e.g.~\cite{janzing05a,lloyd08,Nagaj,childs13,Thompson,Seifnashri}) which show that Hamiltonian simulation is sufficient to perform universal quantum computation.
Indeed, this was already shown for universal classical computation by Feynman~\cite{feynman85}, and his construction can easily be extended to quantum computation.
See~\cite{nagaj08} for much more on this ``Hamiltonian quantum computer'' model, and many further references.

One representative example is a result of Nagaj~\cite{Nagaj}, who showed that for any polynomial-time quantum computation on $n$ qubits there is a 2\nbd-local Hamiltonian $H$ on $\poly(n)$ qubits with $\|H\| = O(\poly(n))$, a time $t = O(\poly(n))$, and an easily constructed product state $\ket{\phi_0}$, such that the output of the computation can be determined (with high probability) by applying $e^{-iHt}$ to $\ket{\phi_0}$ and measuring the resulting state $\ket{\phi_t}$ in the computational basis.
The description of $H$ can be constructed in polynomial time.

Because of the strong consequences of universality, we can use universal Hamiltonians to simulate an encoded version of $H$.
Let $\mathcal{F}$ be an efficiently universal family of qubit Hamiltonians (see \cref{dfn:sim}).
For simplicity, assume further that the efficiently constructable state in \cref{dfn:sim} is $\ket{00\dots0}$, i.e.\ $P\ket{00\dots 0}=\ket{00\dots 0}$.
This is the case for all the simulations constructed in this paper.
Our definition of efficient simulation then implies that, for any polynomial-time quantum computation on $n$ qubits, there is a protocol of the following form to obtain the output of the computation:
\begin{enumerate}
\item Prepare a pure state $U \ket{\phi_0} \ket{0}^{\ox m}$ of $\poly(n)$ qubits, for some encoding map $U$ such that $U$ is a product of unitaries, each of which acts on $O(1)$ qubits.
\item Apply $e^{-iH't}$ for some Hamiltonian $H' \in \mathcal{F}$ such that $\|H'\| = \poly(n)$, and some time $t = \poly(n)$.
\item Decode the output by applying $U^\dg$.
\item Measure the resulting state in the computational basis.
\end{enumerate}
Observe that the first and third steps can be implemented by quantum circuits of depth $O(1)$.
By universality of $\mathcal{F}$, there exists $H' \in \mathcal{F}$ such that $H'$ is a $(\Delta,\eta,\epsilon)$-simulation of $H$ for arbitrary $\epsilon > 0$.
By \cref{time-evolution}, if we take $\eta,\epsilon = 1/\poly(n)$ and evolve according to $H'$ for time $t = \poly(n)$, the resulting state $\ket{\psi}$ is distance $1/\poly(n)$ from an encoded version of $e^{-iHt} \ket{\phi_0}$; call that state $\cEs(\phi_t)$.
By \cref{prop:encodingswork}, the expectation of any encoded measurement operator $\mathcal{E}(A)$ applied to $\cEs(\phi_t)$ is the same as that of $A$ applied to $\phi_t$. Thus applying $U^\dg$ to $\mathcal{E}(\phi_t)$ in order to undo $\mathcal{E}$, and then measuring in the computational basis, would result in the same distribution on measurement outcomes as measuring $\phi_t$ in the computational basis. So the distribution obtained by measuring in step (iv) is close (i.e.\ at total variation distance $1/\poly(n)$) to the distribution that would have been obtained from the measurement at the end of the simulated computation.

Thus our results show that these steps, together with time-evolution according to apparently rather simple interactions are sufficient to perform arbitrary quantum computations.
For example, time-independent Heisenberg interactions with a carefully crafted pattern of coupling strengths, but no additional types of interaction, are sufficient for universal quantum computation; the same holds for XY interactions.
Note that a similar statement was already known for the case of time-dependent Heisenberg interactions~\cite{divincenzo00,kempe00}: the proof of universality there was also based on encoding, though made substantially simpler by the additional freedom afforded by time-dependence.
Also note that, though not stated explicitly there, universality of the Heisenberg interaction on arbitrary graphs for quantum computation should follow from the techniques in~\cite{childs13}.
Universality of the XY interaction for quantum computation, when augmented by some additional restricted types of interactions, was shown in~\cite{childs13,Thompson,Seifnashri}.

We also showed that any universal set of 2\nbd-qubit interactions can efficiently simulate any spatially sparse Hamiltonian, even if all interactions in the simulator Hamiltonian occur on a square lattice.
As there exist families of spatially sparse Hamiltonians which are universal for quantum computation (e.g.~\cite{nagaj08a,Oliveira-Terhal}) this implies that these interactions remain universal for quantum computation on a square lattice.
For example, Heisenberg interactions are universal for quantum computation even when restricted to a 2D square lattice; as are XY interactions.

The converse perspective on this is that these Hamiltonians are more complicated to simulate than one might have previously thought.
Following Lloyd's original quantum simulation algorithm~\cite{Lloyd}, a number of works have developed more efficient algorithms for quantum simulation, whether of general Hamiltonians~\cite{Berry07,Berry15} or Hamiltonians specific to particular physical systems, such as those important to quantum chemistry~\cite{Hastings15,Poulin15}.
However, although these algorithms use very different techniques, one property which they share is that they are highly sequential; to simulate a Hamiltonian on $n$ qubits for time $t$, each of the algorithms requires a quantum circuit of depth $\poly(n,t)$.
Quantum simulation is predicted to be one of the earliest applications of quantum computers, yet maintaining coherence for long times is technically challenging.
So it would be highly desirable for there to exist a Hamiltonian simulation algorithm with low depth; for example, an algorithm whose quantum part consisted of a quantum circuit of depth $\poly(\log(n))$.

Our results give some evidence that such a simulation algorithm is unlikely to exist, even for apparently very simple Hamiltonians such as the Heisenberg model.
If there existed a Hamiltonian simulation algorithm for simulating a Heisenberg Hamiltonian on $n$ qubits for time $t$, whose quantum part were depth $\poly(\log(n,t))$, then the quantum part of \emph{any} polynomial-time quantum computation on $n$ qubits could be compressed to depth $\poly(\log(n))$.
This can be seen as a complexity-theoretic analogue of a query complexity argument~\cite{Berry07} that lower-bounds the time to simulate an arbitrary sparse Hamiltonian.
Unlilke the query complexity approach, using computational complexity theory gives evidence for hardness of simulating explicitly given local Hamiltonians.
In complexity-theoretic terms, our results show that, roughly speaking\footnote{This statement is only approximately true, for several reasons.
  The Hamiltonian simulation problem as we have defined it is intrinsically quantum: the task is to produce the state $e^{-iHt} \ket{\psi}$, given an input state $\ket{\psi}$.
  To formalise this complexity-theoretic claim, one would have to define a suitable notion of quantum reductions which encompassed such ``state transformation'' problems.
  And technically, the hardness result we prove is that the Hamiltonian simulation problem is at least as hard as $\mathsf{PromiseBQP}$, the complexity class corresponding to determining whether measuring the first qubit of the output of a quantum computation is likely to return 0 or 1, given that one of these is the case.
  We choose to omit a discussion of these technical issues.}, simulating any universal class of Hamiltonians is $\mathsf{BQP}$-complete under $\mathsf{QNC}_0$ reductions, where $\mathsf{BQP}$ is the complexity class corresponding to polynomial-time quantum computation, and $\mathsf{QNC}_0$ is the class of depth-$O(1)$ quantum circuits.


\subsection{Adiabatic quantum computation}

The model of adiabatic quantum computation allows arbitrary polynomial-time quantum computations to be performed in the ground state of a family of Hamiltonians~\cite{Aharonov}.
A continuously varying family of Hamiltonians $H(t)$ is used, where $0 \le t \le 1$.
$H(0)$ and $H(1)$ are chosen such that the ground state of $H(0)$ is easily prepared, while the ground state of $H(1)$ encodes the solution to some computational problem.
For example, it could be the computational history state~\cite{Kitaev-Shen-Vyalyi} encoding the entirety of a polynomial-length quantum computation.
At time $t=0$, the system starts in the ground state of $H(0)$.
If the rate of change of $t$ is slow enough, the system remains in its ground state throughout, and at time $t=1$ the solution can be read out from the state by measuring in the computational basis.
In order to perform the adiabatic computation in time $\poly(n)$, it is sufficient that the spectral gap of $H(t)$ is at least $\delta$ for all $t$, for some $\delta \ge 1/\poly(n)$, and that $\|H(t)\|$ and $\| \frac{d}{dt} H(t)\|$ are upper-bounded by $\poly(n)$ for all $t$~\cite{Jansen-Ruskai-Seller}.

It was shown in~\cite{Kempe-Kitaev-Regev} that universal adiabatic quantum computation can be achieved using 2\nbd-local Hamiltonians.
Here we argue, following a similar argument for stoquastic Hamiltonians~\cite{Bravyi-Hastings}, that any of the classes of universal Hamiltonian we have considered here can perform adiabatic quantum computation, given the ability to perform local encoding and decoding unitary operations before and after the adiabatic evolution.

Let $H(t)$ be a family of Hamiltonians used to implement an adiabatic quantum computation.
For each $t$ we define $H'(t)$ to be a $(\Delta,\eta,\epsilon)$-simulation of $H(t)$ using one of the previously discussed classes of universal simulators, where $\eta,\epsilon \le n^{-c}$ for a sufficiently small constant $c$, and let $V(t)$ be the corresponding local isometry.
From the definition of universal simulation, and the fact that the simulations increase the norm of the simulated Hamiltonian by at most a $\poly(n)$ factor, $H'(t)$ has spectral gap at least $\delta - 1/\poly(n)$ and $\|H'(t)\| = O(\poly(n))$.
The ground state of $H'(0)$ can be prepared efficiently by applying $V(0)$ to the ground state of $H(0)$, and the ground state of $H'(1)$ can be read off efficiently by applying $V^\dg(1)$ and measuring in the computational basis.

It remains to show that $\| \frac{d}{dt} H'(t)\| = O(\poly(n))$.
The map $H(t) \mapsto H'(t)$ could in principle introduce singularities, as implementing an effective interaction of weight $\alpha$ using a second-order perturbative reduction requires weights whose scaling with $\alpha$ is $\alpha^{1/2}$, so for $\alpha \rightarrow 0$ the derivative becomes infinite; a similar issue applies to third-order reductions.
This can be avoided, for example, by choosing a cutoff $\alpha_{\min}$, and forming a Hamiltonian $\widetilde{H}$ by replacing each weight $\alpha$ in the original Hamiltonian $H$ with $\widetilde{\alpha} = \sgn(\alpha) \sqrt{\alpha^2 + \alpha_{\min}^2}$.
If $\alpha_{\min}$ is sufficiently small (yet still inverse-polynomial in $n$), $\|\widetilde{H}-H\| \le n^{-c'}$ for an arbitrarily small constant $c'$, and also $\| \frac{d}{dt} \widetilde{H}'(t)\| = O(\poly(n))$.



We would like to thank Gemma De las Cuevas and David Gosset for helpful conversations about the topic of this work. TSC was supported by the Royal Society, and by grant \#48322 from the John Templeton Foundation. AM was supported by an EPSRC Early Career Fellowship (EP/L021005/1). SP was supported by the EPSRC. The opinions expressed in this publication are those of the authors and do not necessarily reflect the views of the John Templeton Foundation. No new data were created during this study.


\clearpage

\bibliographystyle{alpha}
\bibliography{universalH}

\end{document}